\documentclass[letterpaper, twocolumn, 10pt]{article}
\usepackage{usenix-2020-09}

\usepackage{xspace}
\usepackage{lipsum}
\usepackage{graphicx}%插入图片
\usepackage{subfigure} %子图

\usepackage{amsmath}
\newtheorem{theorem}{Theorem}
\newtheorem{lemma}{Lemma}
\newtheorem{proof}{Proof}[section]
\usepackage{amsfonts}
\usepackage{color, xcolor}
\usepackage[linesnumbered,ruled]{algorithm2e}
\usepackage{booktabs}
\usepackage{multirow}
\usepackage{threeparttable}
\usepackage{array}
\usepackage{bbding}
\usepackage{colortbl}
\usepackage{xcolor}

\newcommand{\sys}{{\sc V-Cloak}\xspace}

\begin{document}
%-------------------------------------------------------------------------------

%don't want date printed
\date{}

\title{\Large \bf \sys: Intelligibility-, Naturalness- \& {Timbre-}Preserving\\Real-Time Voice Anonymization}

%for single author (just remove % characters)
% \author{
% {\rm Demo: \url{https://v-cloak.com}}\\
% \rm Source Code: \url{https://github.com/V-Cloak/V-Cloak}
% % Institution
% % \and
% % {\rm Second Name}\\
% % Second Institution
% % copy the following lines to add more authors
% % \and
% % {\rm Name}\\
% %Name Institution
% } % end author

\normalsize{
	\author{
	\rm Jiangyi Deng$^1$, Fei Teng$^1$, Yanjiao Chen$^1$\footnote[1]{Corresponding author.}~, Xiaofu Chen$^2$, Zhaohui Wang$^2$, Wenyuan Xu$^1$\\
	$^1$Zhejiang University, $^2$Wuhan University\\\\
	\rm Demo: \url{https://v-cloak.com}

    }
}

\maketitle

\renewcommand{\thefootnote}{\fnsymbol{footnote}} 
\footnotetext[1]{Corresponding author.}

\renewcommand{\thefootnote}{\arabic{footnote}}

\begin{abstract}

Voice data generated on instant messaging or social media applications contains unique user voiceprints that may be abused by malicious adversaries for identity inference or identity theft.
Existing voice anonymization techniques, e.g., signal processing and voice conversion/synthesis, suffer from  degradation of perceptual quality.
In this paper, we develop a voice anonymization system, named \sys, which attains real-time voice anonymization while preserving the intelligibility, naturalness and timbre of the audio. Our designed anonymizer features a one-shot generative model that modulates the features of the original audio at different frequency levels. We train the anonymizer with a carefully-designed loss function. Apart from the anonymity loss, we further incorporate the intelligibility loss and the psychoacoustics-based naturalness loss. The anonymizer can realize untargeted and targeted anonymization to achieve the anonymity goals of unidentifiability and unlinkability.

We have conducted extensive experiments on four datasets, i.e., LibriSpeech (English), AISHELL (Chinese), CommonVoice (French) and CommonVoice (Italian), five Automatic Speaker Verification (ASV) systems (including two DNN-based, two statistical and one commercial ASV), and eleven Automatic Speech Recognition (ASR) systems (for different languages). Experiment results confirm that \sys outperforms five baselines in terms of anonymity performance. We also demonstrate that \sys trained only on the VoxCeleb1 dataset against ECAPA-TDNN ASV and DeepSpeech2 ASR has transferable anonymity against other ASVs and cross-language intelligibility for other ASRs. Furthermore, we verify the robustness of \sys against various de-noising techniques and adaptive attacks.
Hopefully, \sys may provide a cloak for us in a \emph{prism} world.

\end{abstract}

% \begin{figure}[t]
% \centering
% \setlength{\abovecaptionskip}{5pt}% 
% \setlength{\belowcaptionskip}{-0.3cm}%
% \includegraphics[width=3.5in]{figures/threat_scenario.pdf}
% \caption{Privacy Leakage Scenario. The operating system passes raw audio data to applications by default. Those service providers with ulterior motives might infer and collect speaker's identity/voiceprint with the raw audio.}
% \label{fig:scenario}
% \end{figure}
% \raggedbottom
\begin{figure}[t]
    \centering
\setlength{\abovecaptionskip}{4pt}
\setlength{\belowcaptionskip}{-0.40cm}
\subfigcapskip = -0.05cm

% \subfigure[Privacy leakage scenarios]{
%     \includegraphics[width=3.35in, trim=192 200 190 200, clip]{figures_v2/threat_scenario.pdf}
%     \label{fig:scen-leakage}
%     }
% \\\vspace{-0.0cm}
% \subfigure[\sys application scenarios]{
%     \includegraphics[width=3.35in, trim=192 200 190 200, clip]{figures_v2/application_scenario.pdf}
%     \label{fig:scen-application}
%     }
\includegraphics[width=3.4in, trim=225 155 230 155, clip]{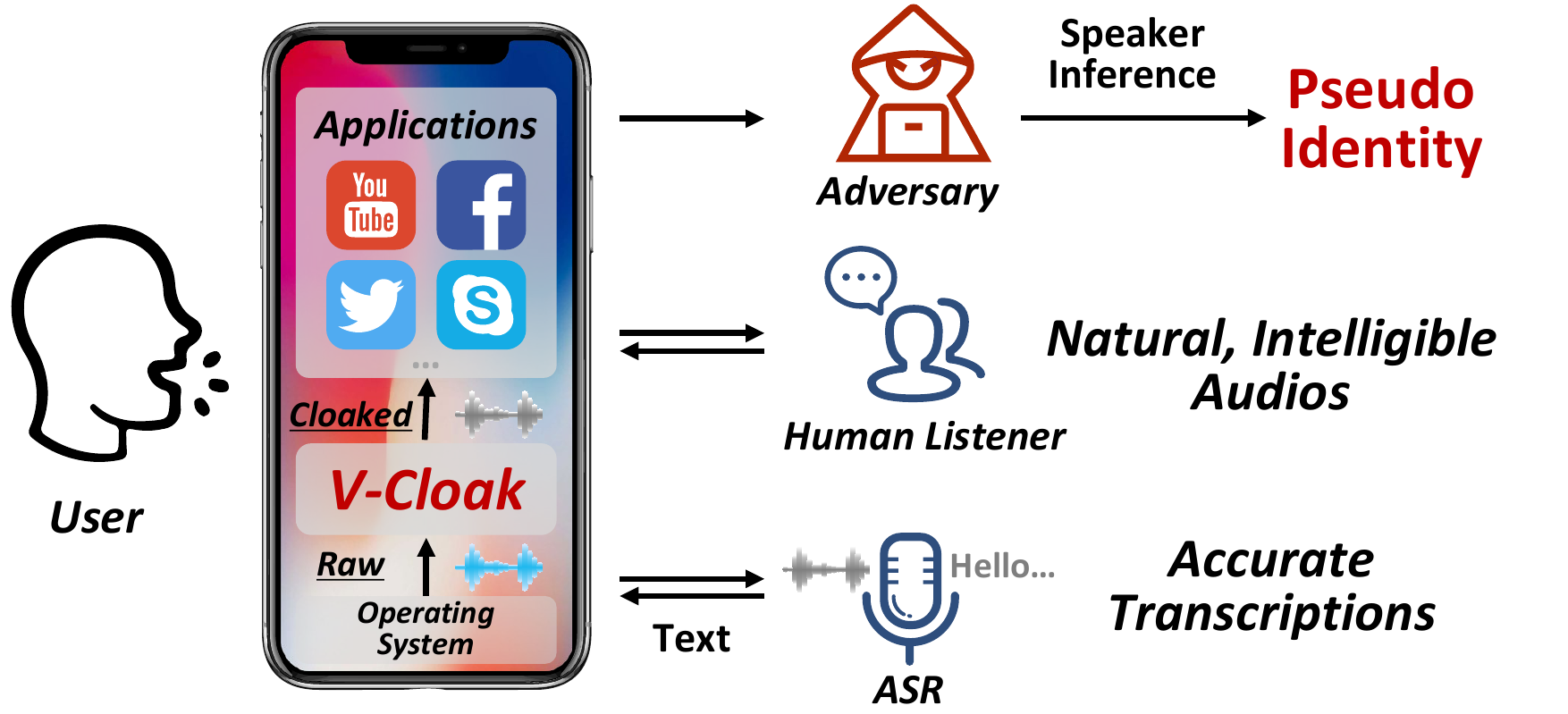}

\caption{Voiceprint in voice data may be leveraged by malicious adversaries for identity inference or identity theft. The raw audio is cloaked with \sys before being passed to applications, thus malicious service providers or third parties can only obtain a pseudo identity/voiceprint.}\label{fig:scenarios}
\end{figure}

\begin{table*}\centering
\begin{threeparttable}[t]
\setlength{\abovecaptionskip}{5pt}% 
\setlength{\belowcaptionskip}{0pt}%

\caption{\sys versus existing works.}

\begin{tabular}{@{}lcccccc@{}}
\toprule
 \textbf{Method} 	& \textbf{Type$^*$}	 & \textbf{Intelligibility$^\#$} & \textbf{Naturalness$^\#$} & \textbf{Timbre-preserving} 				 & \textbf{Real-Time Coef.$^\downarrow$} 	& \textbf{User-agnostic$^\dagger$} \\ \midrule
VoiceMask~\cite{qian2018hidebehind, qian2019speech}   & VC                               		& \XSolidBrush           & \XSolidBrush                                                         & \XSolidBrush      & 0.041    & \Checkmark                                                              \\
Yoo~\cite{DBLP:journals/access/YooLLOKY20}       & VC                          	& \XSolidBrush           & \Checkmark                                                       & \XSolidBrush         & \emph{N.K.}   & \XSolidBrush                                                              \\
NSF~\cite{fang2019speaker, han2020voice}      	& VS                           			& \Checkmark         & \Checkmark                                                      & \XSolidBrush          & 0.110   & \XSolidBrush                                                                 \\
HFGAN~\cite{DBLP:journals/corr/abs-2202-13097}	& VS									& \Checkmark			 &\Checkmark														& \XSolidBrush		 & 0.104 	& \Checkmark																\\
Justin~\cite{DBLP:conf/fg/JustinSDVIM15}       	& VS                            	& \Checkmark             & \Checkmark                                                        & \XSolidBrush      & \emph{N.K.}     & \Checkmark                                                                \\
McAdams~\cite{DBLP:conf/interspeech/0001TTNE21} 	& SP                           		& \XSolidBrush              & \XSolidBrush                                                         & \XSolidBrush       & 0.030   & \Checkmark                                                             \\
Vaidya~\cite{DBLP:conf/sp/VaidyaS19}           	& SP                    		& \XSolidBrush             & \XSolidBrush                                                         & \XSolidBrush     & \emph{N.K.}	    & \Checkmark                                                             \\
\textbf{\sys (Ours)}                                                    & \textbf{Adv}                        & \Checkmark             & \Checkmark                                                           & \Checkmark         & \textbf{0.011}     & \Checkmark                                                                \\ \bottomrule
\end{tabular}
\begin{tablenotes}[flushleft]
\item[] \vspace{-2pt}\hspace{-2pt}\small (i) $^*$: Voice Conversion (VC). Voice Synthesis (VS). Signal Processing (SP). Adversarial examples (Adv). (ii) $^\#$: whether the method has explicit constraints on intelligibility or naturalness. (iii) $^\downarrow$: Real-time coefficient (RTC), the ratio between the processing time and the duration of the audio. \textbf{The lower the RTC, the more efficient the method.} We measure the five methods under the same computing resource conditions. \emph{N.K.}, not known, the authors did not evaluate the efficiency of their methods or make their codes available.   (iv) $^\dagger$: whether the method needs to be trained for a new user. 
\end{tablenotes}

\label{tab:compare}
\end{threeparttable}
\end{table*}

\section{Introduction}

{Voiceprint is a critical biometric that can uniquely identify a person. As massive personal data is collected and processed by online services, there are rising concerns for privacy leakage. In 2018, the European Union enforced the General Data Protection Regulation (GDPR) \cite{GDPR} for personal data protection, especially for biometric data. However, an avalanche of voice data is generated daily on social media (e.g. Facebook/Meta, WeChat, TikTok) and in communication applications (e.g. Zoom, Slack, Microsoft Teams, Ding Talk), and automated processing methods, e.g., ASV, can easily extract voiceprint for ill use. For example, as shown in Figure~\ref{fig:scenarios}, an adversary may infer the speaker identity of a private conversation from voice messages uploaded to the cloud with an ASV~\cite{Facebook, TikTok, iFlytek}. Therefore, there is an urgent demand for voice anonymization to help users protect voiceprint while enjoying voice-related services (e.g., speech recognition by ASR) and interpersonal communication (e.g., human listeners can identify the speaker).}

Existing voice anonymization methods are mainly based on voice signal processing (SP), voice conversion (VC) and voice synthesis (VS). SP~\cite{DBLP:conf/interspeech/0001TTNE21, DBLP:conf/sp/VaidyaS19} methods directly apply signal processing techniques to modify speaker-related features in audios to obscure voiceprints. Nonetheless, SP-based voice anonymization usually induces large quality degradation as intelligibility and naturalness are not considered. VC~\cite{qian2018hidebehind,  qian2019speech, DBLP:journals/access/YooLLOKY20, srivastava2020evaluating} and VS \cite{fang2019speaker, han2020voice, DBLP:journals/corr/abs-2202-13097,DBLP:conf/fg/JustinSDVIM15} methods convert the original audio into another audio that sounds completely different from the original speaker. Although VC and VS may achieve anonymity, they are not suitable for scenarios where the user wants to hide their identity from ASVs but hopes to preserve their personal timbre to human audiences, e.g., posts of celebrities on social media, voice messages with acquaintances.

In this paper, we make the first attempt to design a real-time voice anonymization system, named \sys, which achieves anonymity while preserving  intelligibility, naturalness, and timbre of the audios. A comparison of \sys with existing works is shown in Table~\ref{tab:compare}. Nonetheless, to realize these design goals with a practical real-time system is challenging in three aspects.

\begin{itemize}
    \item \emph{How to achieve real-time voice anonymity against adaptive attacks?}
\end{itemize}

Different from traditional signal processing and voice conversion \& synthesis, we are inspired by the adversarial examples that can trick ASV into misidentifying the speaker but induce imperceptible differences to the human auditory system. Nonetheless, directly applying adversarial examples to voice anonymization has two major issues. First, most of the existing ASV adversarial examples~\cite{zheng2021black, li2020practical, li2020advpulse, chen2021real, DBLP:conf/ccs/DuJLGWB20} are constructed via iterative updates, which cannot 
% fulfil the requirement of 
achieve real-time voice anonymization. As far as we are concerned, there is only one ASV adversarial attack named FAPG that creates adversarial examples using a one-shot generative model~\cite{xie2020enabling}. Unfortunately, FAPG needs to train a feature map for each potential target speaker and
%  our experiments show that FAPG cannot converge when the pool of target speakers is large (
the original paper only evaluates for an ASV with 10 speakers. Furthermore, the adversary may be informed of the anonymization method and the model (anonymizer), and then launches an adaptive attack to de-anonymize the anonymized audio. 

To tackle these problems, 
% we develop a real-time voice anonymization system named \sys, built upon 
we adapt a lightweight generative model Wave-U-Net~\cite{DBLP:conf/ismir/StollerED18} for \sys. We equip Wave-U-Net with two novel components, i.e., \emph{VP-Modulation} and \emph{Throttle}. \emph{VP-Modulation} modulates the feature elements of the original audio at each frequency level according to the voiceprint of a target speaker.
\emph{Throttle} adjusts the weights of features of the original audio at different frequency levels to conform to the constraint on the anonymization perturbations. The trained anonymizer can produce anonymized audios targeting any speaker/voiceprint under any anonymization perturbation constraint without re-training. Furthermore, we conduct theoretical analysis and experiments to verify the anonymity of \sys in the case of adaptive attacks.

\begin{itemize}
    \item \emph{How to maintain objective and subjective intelligibility of anonymized audios?}
\end{itemize}

It is desirable for the anonymized audios to be intelligible to ASRs (objective intelligibility) such that the users can still enjoy speech-to-text services; and to humans (subjective intelligibility) such that voice messages can be understood. However, SP- and VC-based anonymization, as well as voice adversarial examples, do not consider intelligibility constraint and may introduce noises that greatly degrade intelligibility.

To address this issue, we impose an intelligibility loss when training the anonymizer. The intelligibility loss is based on the decoding error rate of the ASR. Instead of the commonly-used Connectionist Temporal Classification (CTC) loss of ASR, we acquire the graphemic posteriorgram (GPG) loss, which preserves the full alignment of the transcription and the grapheme of each frame.
% without trimming blank tokens. 
The subjective intelligibility is achieved by constraining the anonymization perturbations by our proposed \emph{Throttle} module and better masking the anonymization perturbations based on psychoacoustics.

\begin{itemize}
    \item \emph{How to preserve naturalness and timbre of anonymized audios?}
\end{itemize}

Naturalness and timbre preservation are important to human audiences or listeners of anonymized audios. Signal processing and existing ASV adversarial examples did not consider naturalness such that the processed audios may sound mechanical. In addition, signal processing, voice conversion and voice synthesis all distort the timbre of the original speaker such that the anonymized audio sounds unlike being spoken by the original speaker (e.g., a friend or a celebrity). 

To cope with this problem, we introduce a naturalness \& timbre loss when training the anonymizer based on the psychoacoustic theory of masking effects. Our user study verifies that the anonymized audios of \sys receive high naturalness and timbre scores.

We implement a fully-functional prototype of \sys, evaluated with extensive experiments on five ASVs (anonymity) and eleven ASRs (intelligibility) with datasets of four languages (English, Chinese, French, Italian). 
The comparison with five baselines demonstrates that \sys achieves the best anonymization performance with the second-best intelligibility performance. Cross-language experiments show that the anonymizer of \sys trained on one ASV and one ASR can be transferred to other ASVs and ASRs (with different languages). A user study with 102 volunteers confirms the intelligibility-, naturalness- and timbre-preserving properties of \sys. 

We summarize our main contributions as follows.
\begin{itemize}
    \item We propose \sys, an intelligibility-, naturalness- and timbre-preserving voice anonymization system. \sys is proved and evaluated to fulfil the anonymization goals of unidentifiability and unlinkability against naive and adaptive adversaries.
    
    \item We develop a real-time anonymizer that transforms the original audio into targeted or untargeted anonymized audios. The anonymizer is trained with anonymity, intelligibility, naturalness and timbre loss, generalizing to any new original speaker or new target speaker without the need for re-training.
    
    \item We conduct extensive experiments to verify the effectiveness and efficiency of \sys under various testing conditions and a user study to confirm the practicality and applicability of \sys.
\end{itemize}

\section{Background}

\subsection{Voice Data}

In the digital world, voice data of a user is massively generated and distributed for various purposes, e.g., communications via voice messages or video posts on social media. These wildly exposed voice data may be easily collected by service providers or third parties. For instance, Facebook is collecting audio data from voice messages on its social network platform, and even attempts to transcribe the content of these private messages~\cite{Facebook}. TikTok revised its privacy policy to legitimize faceprints and voiceprints collection from the videos uploaded by users, and even claimed the possibility of data sharing for business purposes~\cite{TikTok}.

Voice data contains two kinds of information, i.e., speech contents and phonetic features. 

\begin{itemize}
    \item {Speech contents}. Speech contents refer to the linguistic information contained in the voice data, i.e., "what are the words spoken." Speech contents determine the intelligibility of the voice data. 
    \item {Phonetic features}. Phonetic features refer to the way the speech contents are conveyed in the voice data, i.e., "how are the words spoken." Phonetic features affect the timbre of the voice data. 
\end{itemize}

Voiceprint is a phonetic feature that can uniquely identify a speaker. However, voiceprint contained in voice data may be abused for identity inference or identity theft. On the one hand, voiceprint may be used to infer the identity of speakers of a private conversation by automatic speaker verification (ASV) systems. On the other hand, the voiceprint of a speaker may be extracted from audios to synthesize audios to pass voiceprint-based authentication systems. For example, WeChat, a popular messaging app in China, allows users to login via voiceprint~\cite{wechat}. In face of these potential privacy leakages, it is essential for users to anonymize voice data before sending voice messages or publishing videos on social media.

\subsection{{Psychoacoustics}}

Psychoacoustics is the study of the relationship between subjective psychological perceptions (e.g., perceived volume, pitch) and objective physical parameters (e.g., sound pressure level, frequency)~\cite{yost2015psychoacoustics}. The masking effect is one of the most common psychoacoustic phenomena~\cite{gelfand2017hearing}. There are two forms of masking: temporal and spectral. Temporal masking refers to the situation where a sound cannot be perceived if a sudden louder sound appears immediately preceding or following the first one. The louder sound is called the \emph{masker}. Spectral masking refers to the imperceptibility of a sound component due to other frequency components played simultaneously. The \emph{perception threshold} of this component varies due to both sound signals (e.g., frequency) and the listener. We leverage the spectral masking effect to make anonymization perturbations more imperceptible to human users.

\subsection{Automatic Speaker Verification \& Automatic Speech Recognition}

An Automatic Speaker Verification (ASV) system aims to deduce the speakers of audios based on their voiceprints. Speaker inference via an ASV system includes the enrollment phase and the inference phase. In the enrollment phase, clean audio samples of the speaker to be recognized are fed into the ASV such that the voiceprint can be extracted and stored in the ASV. In the inference phase, the ASV takes an audio sample as input and outputs whether the input audio belongs to the enrolled speaker.
% to which enrolled speaker the audio belongs. 
There are two mainstream methods of extracting and matching voiceprints, i.e., statistical models and Deep Neural Network (DNN)-based models. Gaussian mixture model (GMM) is a traditional statistical model to extract ivector voiceprints. ivector-PLDA is a popular ASV implementation that matches ivector voiceprints via probabilistic linear discriminant analysis (PLDA). X-vector is a DNN-based voiceprint extractor, which outperforms GMM as DNNs are more effective in extracting feature representations from large-scale voice datasets. ECAPA-TDNN \cite{desplanques2020ecapa} 
is the state-of-the-art ASV implementation using end-to-end training, i.e., training the front-end and the back-end jointly as an integrated network~\cite{wang2017deep}.

An Automatic Speech Recognition (ASR) system aims to transcribe the speech contents from audio samples (without the need to know the speaker). In the training process, audios are first transformed into a sequence of spectral frames. Each frame is then transformed into a feature vector. Commonly used features include Filter Bank (FBank) \cite{richmond2002estimating}, Mel-Frequency Cepstral Coefficients (MFCC) \cite{muda2010voice}, Spectral Subband Centroid (SSC) \cite{thian2004spectral} and Perceptual Linear Predictive (PLP) \cite{hermansky1990perceptual}. Then, the posterior probability of the lingueme (e.g., phoneme, grapheme, or word) contained in each frame is estimated. The linguemes are usually represented as tokens. For example, 29 tokens are used for the English language, i.e., letters a$\sim$z, space, apostrophe and the special {blank} token $\phi$. Next, the Connectionist Temporal Classification (CTC) module sums the probability of all possible alignments that reduce to the ground-truth sequence. For example, a three-frame sequence of $[a~b~\phi], [a~\phi~b]$ and $[\phi~a~b]$ will all be reduced to the ground-truth sequence of $[a~b]$. Finally, the model is updated to increase the probability of producing the ground-truth sequence. In the inference phase, a language model may be used to provide a prior probability to find the lingueme sequence of the highest probability.

\begin{figure}[tt]
\centering
\setlength{\abovecaptionskip}{5pt}% 
\setlength{\belowcaptionskip}{-0.0cm}%
\includegraphics[width=3.2in, trim=220 40 180 55, clip]{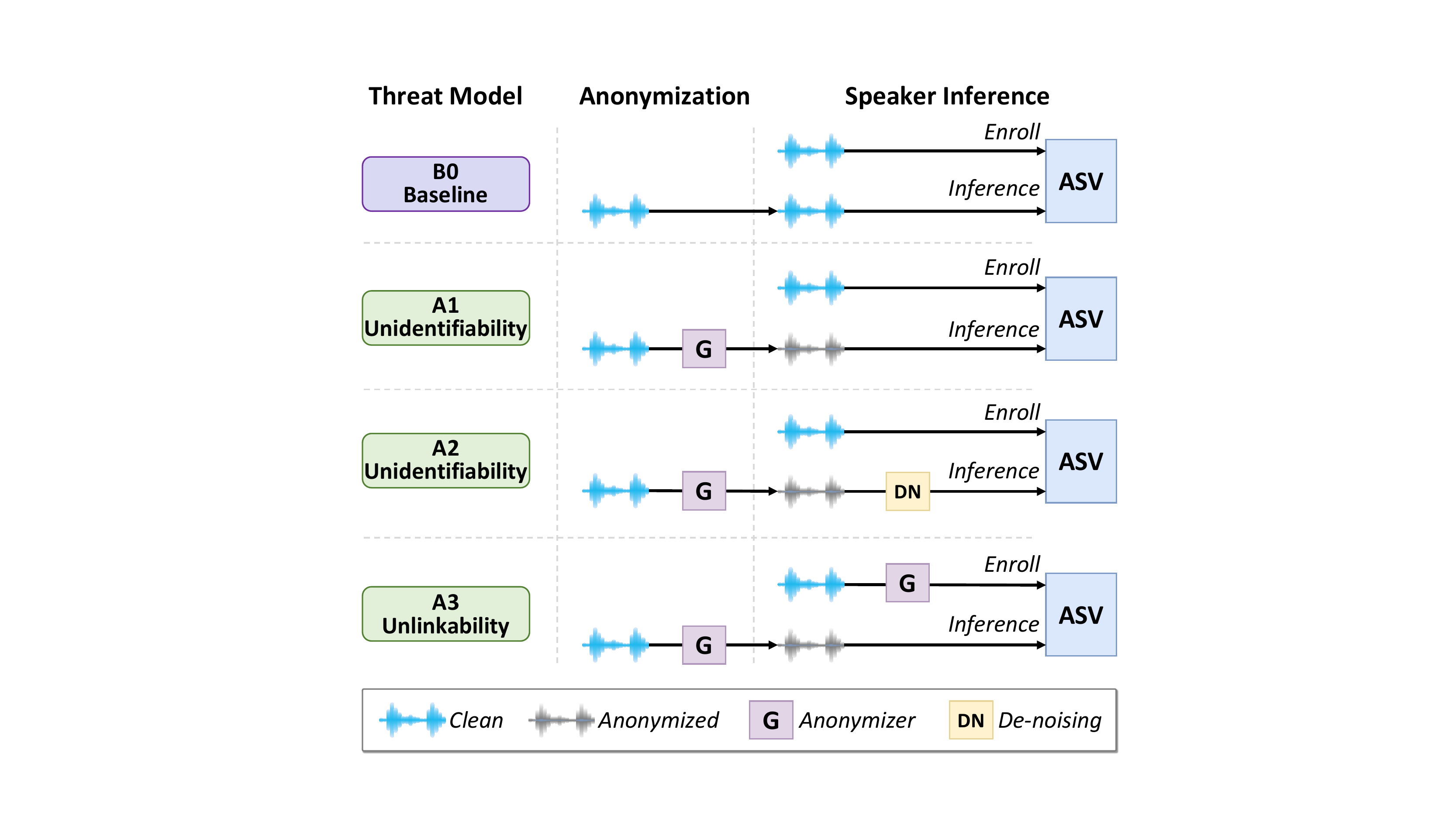}
\caption{Threat model. \textbf{A1:} \emph{ignorant} adversary who enrolls clean audios into the ASV and feeds anonymized audio into the ASV to infer the speaker. \textbf{A2:} \emph{semi-informed} adversary who enrolls clean audios into the ASV and feeds de-noised anonymized audio into the ASV to infer the speaker. \textbf{A3:} \emph{informed} adversary who enrolls anonymizer-processed audios into the ASV and feeds the anonymized audio into the ASV to infer the speaker.}
\label{fig:metrics}
\end{figure}
\raggedbottom

\subsection{Voice Anonymization}\label{sec:voice-anonymization}

Voice anonymization refers to the practice of removing voiceprint from voice data. A voice anonymization system needs to satisfy various requirements to fulfil different purposes. Regarding the digital voice data privacy, we aim to achieve the following performance goals.

\textbf{Anonymity}. An anonymized audio should not reveal the identity of the speaker. More specifically, we consider concealing speaker identities in the digital domain from ASVs. 

\textbf{Intelligibility}. An anonymized audio should be intelligible to both humans and ASRs. More specifically, the speech contents of the anonymized audio can be correctly understood by humans and transcribed by ASRs.

\textbf{Naturalness \& Timbre}. An anonymized audio should sound natural and like the timbre of the original speaker to humans. Studies show that most people find highly mechanical audios irritating and discomforting to listen to, thus natural-sounding anonymized audios are more user-friendly~\cite{taylor2009text}. For voice messages and video posts on social media, it is ideal to make the anonymized audios sound authentic as the original speaker to audiences, especially for communications between acquaintances and publicity of celebrities.

Voice anonymization can be realized in various ways, as summarized in Table~\ref{tab:compare}.

\emph{Voice signal processing}. Signal processing techniques attempt to contort the voiceprint by directly modifying the voice signals in terms of formant positions, pitch, tempo, or pause~\cite{DBLP:conf/interspeech/0001TTNE21, DBLP:conf/sp/VaidyaS19}. Though simple and fast, signal processing may degrade the intelligibility and naturalness of audios. 

\emph{Voice conversion \& synthesis}. Voice conversion \& synthesis techniques aim to replace the voiceprint of the original speaker in an audio with the voiceprint of another speaker~\cite{qian2018hidebehind,  qian2019speech, DBLP:journals/access/YooLLOKY20, srivastava2020evaluating, fang2019speaker, han2020voice, DBLP:journals/corr/abs-2202-13097,DBLP:conf/fg/JustinSDVIM15}. Voice conversion \& synthesis preserve the intelligibility and naturalness, but alter the timbre so that the audio sounds unlike the original speaker. This may reduce the authenticity of voice messages to acquaintances and video posts of celebrities.

\emph{Voice adversarial examples}. Adversarial example attacks against ASVs add imperceptible noises to audios such that the ASV cannot recognize the speaker~\cite{zheng2021black, li2020practical, li2020advpulse, chen2021real, DBLP:conf/ccs/DuJLGWB20, xie2020enabling}. Adversarial perturbations can be generated in two ways. 

\begin{itemize}
\item \emph{Iterative optimization.} Optimization-based methods formulate the problem of adversarial perturbation generation as a constrained optimization problem~\cite{zheng2021black, li2020practical, li2020advpulse, chen2021real, DBLP:conf/ccs/DuJLGWB20}. As the formulated optimization problems are usually NP-hard, the solutions can only be approximated through iterative updates, which is quite time-consuming. Therefore, iterative optimization methods cannot be applied to real-time services.

\item \emph{One-shot generative model.} Generative models can be trained to produce adversarial perturbations in one shot. Commonly used generative models include Generative Adversarial Networks (GAN) and autoencoders~\cite{DBLP:conf/cvpr/PoursaeedKGB18, DBLP:conf/nips/SongSKE18, DBLP:conf/aaai/PhanXLC020}. As far as we know, there is only one study on generative model-based adversarial examples against ASV, named FAPG~\cite{xie2020enabling}. However, FAPG mainly focuses on deceiving ASVs but not preserving intelligibility and naturalness of audios. 
\end{itemize}

% \begin{itemize}
%     \item Existing works
%     \item Adversarial Example
%     \item Iterative Optimization Approaches
%     \item Generator-oriented Approaches
%     \item Adversarial Example Generator (AEG)
%     \item Generative Adversarial Network (GAN)
% \end{itemize}

\subsection{Threat Model}\label{subsec:threatmodel}

We define the threat model in terms of the adversary's knowledge and capability, then we elaborate the performance goals of voice anonymization under the defined threat model. As shown in Figure~\ref{fig:metrics}, we consider three kinds of adversaries, i.e., \emph{ignorant} (A1), \emph{semi-informed} (A2), and \emph{informed} (A3), with different knowledge and capabilities.

\textbf{Knowledge.} The adversary has an anonymized audio whose speaker is unknown. The adversary has collected a few clean samples of a pool of potential speakers to help with identity inference.  Adversary A1 does not know that the audio is anonymized. Adversary A2 knows that the audio is anonymized but does not know the specific anonymizer. Adversary A3 has full knowledge of the anonymizer.

\textbf{Capability.} Adversary A1, A2, and A3 can use any ASVs to infer the speaker of the anonymized audio. As shown in Figure~\ref{fig:metrics}, A1 and A2 enroll potential speakers in the ASV using clean audios, and A3 enrolls potential speakers in the ASV using audio samples processed by the anonymizer. In the inference phase, A1 directly feeds the anonymized audio into the ASV; A2 applies de-noising methods to the anonymized audio and feeds the de-noised audio into the ASV; A3 also directly feeds the anonymized audio into the ASV.

In the face of the knowledge and the capability of adversaries, we further elaborate the goal of achieving anonymity regarding different types of adversaries. More specifically, in the case of \emph{ignorant} and \emph{semi-informed} adversaries, the speaker of the anonymized audio should be unidentifiable, and in the case of \emph{informed} adversaries, the speaker and the anonymized audio should be unlinkable. 

\textbf{Unidentifiability.} For A1 and A2 who enroll clean voiceprints into the ASV, the speaker of an anonymized audio should not be identified during the inference phase.

\textbf{Unlinkability.} For adversary A3 who enrolls anonymizer-processed voiceprints into the ASV, the speaker of an anonymized audio should be undistinguishable from other speakers.

\section{Problem Formulation}

Before delving into the design details of \sys, in this section, we formally formulate the voice anonymization as a constrained optimization problem. 

Given an audio sample ${x}=[x_1, \cdots, x_{D}]\in \mathbb{R}^{1\times D}$, where $\mathbb{R}^{1\times D}$ is a $D$-dimensional real number field, and $D$ is the length of the audio. Without loss of generality, we assume $x_i\in[-1,1]$. We aim to obtain an anonymized audio $\tilde{x}$ such that the ASV cannot match the voiceprint of $\tilde{x}$ with that of $x$.  Let $\mathcal{V}:\mathbb{R}^{1\times \cdot}\rightarrow\mathbb{R}^{1\times N}$, denote the voiceprint extraction function that outputs a voiceprint of a fixed length $N$, and $\mathcal{G}:\mathbb{R}^{1\times D}\rightarrow\mathbb{R}^{1\times D_o}$ denote the anonymizer function.

\noindent\textbf{\underline{Basic Formulation:}}
\begin{equation}\label{equ:obj1}
\begin{aligned}
&\mathop{\min}\limits_{\mathcal{G}} ~\ L_{\mathrm{ASV}} \\
&~\mathrm{s.t.}~\ \|\tilde{x}- x\|_{\infty}\leq \epsilon~\mathrm{and}~x, \tilde{x}\in[-1,1],\\
 \mathrm{where}&\\
&L_{\mathrm{ASV}} = 
\begin{cases}
\mathcal{S}\big(\mathcal{V}(\tilde{x}), \mathcal{V}(x)\big), & \mathrm{untargeted~anonymization},\\
-\mathcal{S}\big(\mathcal{V}(\tilde{x}), v\big), & \mathrm{targeted~anonymization}.
\end{cases}\\
&\tilde{x} = 
\begin{cases}
\mathcal{G}(x), & \mathrm{untargeted~anonymization},\\
\mathcal{G}(x,v), & \mathrm{targeted~anonymization}.
\end{cases}
\end{aligned}
\end{equation}
where $\epsilon$ constrains the $l_{\infty}$ norm difference between $x$ and $\tilde{x}$, $\mathcal{S}(\cdot, \cdot)$ is the scoring function measuring the similarity between the voiceprints of $x$ and $\tilde{x}$, and $v$ is the voiceprint of a speaker other than $x$. With untargeted anonymization, the voiceprint of the anonymized audio is diverted from that of the original audio as much as possible, which guarantees \emph{unidentifiability}, i.e., the voiceprint of the anonymized audio will not match the voiceprint of the original audio. With targeted anonymization, the voiceprints of two anonymized audios with different original speakers but the same target speaker $v$ will both be matched with $v$ (thus be matched together), which guarantees both \emph{unidentifiability} and \emph{unlinkability}. We theoretically analyze the \emph{unidentifiability} and the \emph{unlinkability} of targeted and untargeted anonymizations in Appendix~\ref{sec:proof} and perform corresponding evaluations in \S\ref{sec:evaluation}.

The anonymized audio obtained by Equation (\ref{equ:obj1}) satisfies the basic goal of anonymity but may suffer from quality degradation in terms of intelligibility, naturalness and timbre. To tackle this problem, we equip the basic optimization problem with loss terms that address the performance goals of intelligibility, naturalness and timbre preservation. More specifically, we introduce an ASR-related loss term, which maintains the intelligibility of the anonymized audio for ASR. We also add a psychoacoustic-related loss term and an $l_2$-norm loss term to improve naturalness of the anonymized audio. Overall, the refined optimization problem is

\noindent\underline{\textbf{\sys Formulation:}}
\begin{equation}\label{equ:obj3}
\begin{aligned}
\mathop{\min}\limits_{\mathcal{G}} & \ \ L_{\mathrm{ASV}}+\alpha\cdot L_{\textrm{ASR}}(x, \tilde{x}) 
+\beta\cdot L_{\textrm{PSY}}(x,\tilde{x}) +\gamma\cdot\|\tilde{x}-x\|_2,\\
\textrm{s.t.}~&\|\tilde{x}- x\|_{\infty}\leq \epsilon~\textrm{and}~x, \tilde{x}\in[-1,1],
\end{aligned}
\end{equation}
where parameters $\alpha,~\beta,~\gamma$ balance the trade-off among the performance goals.

\begin{figure}[t]
    \centering
\setlength{\abovecaptionskip}{5pt}
\setlength{\belowcaptionskip}{-0.0cm}
\subfigcapskip = -0.05cm
\includegraphics[width=3.2in, trim=270 150 270 150, clip]{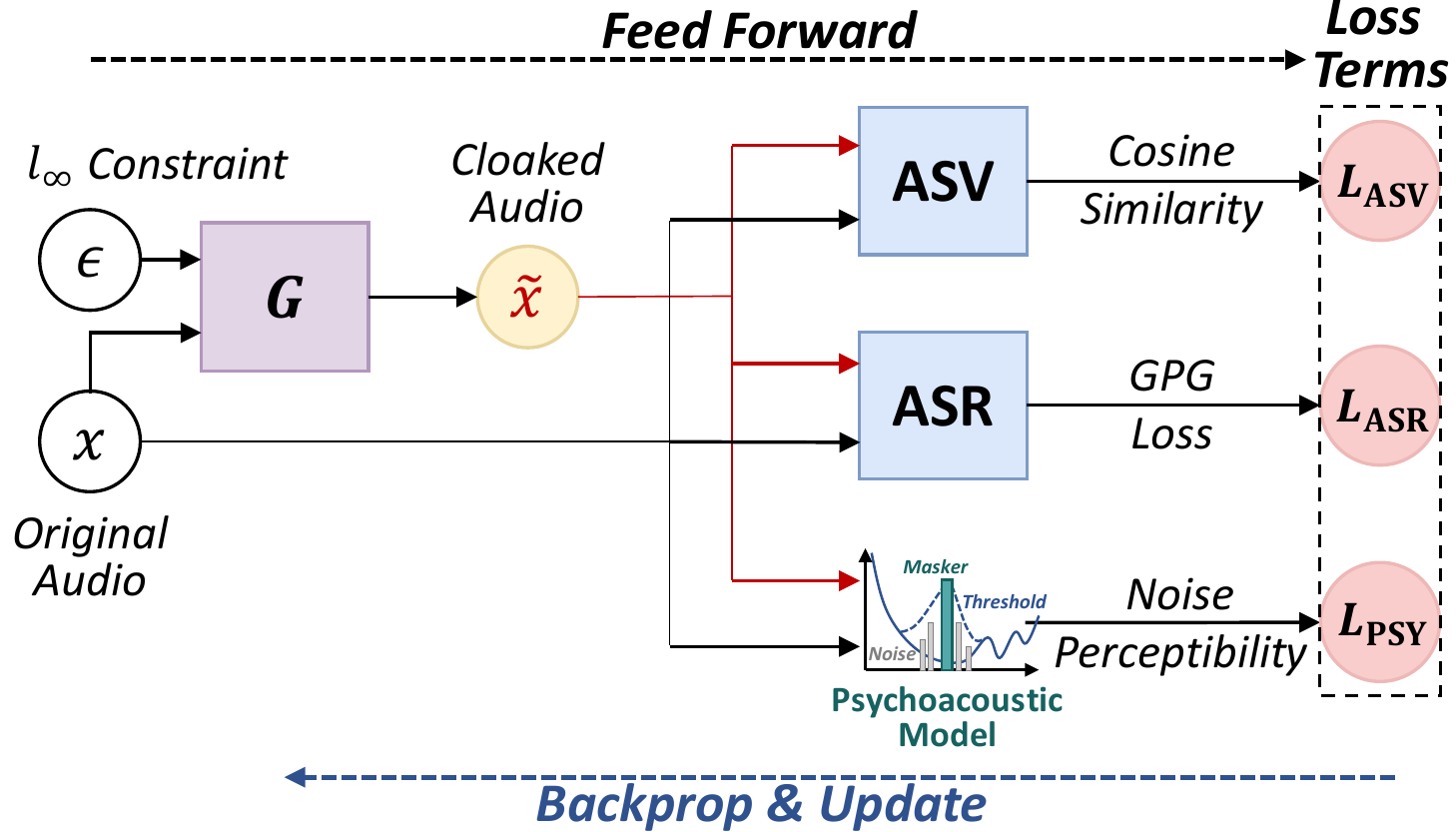}

% \subfigure[Training Generator]{
%     \includegraphics[width=3.5in]{figures/train_generator.pdf}
%     \label{fig:train_g}
%     }
% \\\vspace{-0.1cm}
% \subfigure[Training Discriminator]{
%     \includegraphics[width=3.5in]{figures/train_discriminator.pdf}
%     \label{fig:train_d}
%     }

\caption{Architecture of \sys. The anonymizer $G$ produces the anonymized audio $\tilde{x}$ given the original audio $x$ and the threshold $\epsilon$. $G$ is trained to minimize the loss function related to the performance goals of anonymity, intelligibility, naturalness and timbre.}\label{fig:training}
\end{figure}

\section{\sys: Design Details}

The optimization problem in Equation (\ref{equ:obj3}) is difficult to solve directly. Therefore, we propose a framework named \sys to derive the solution of Equation (\ref{equ:obj3}) based on a generative model. As shown in Figure~\ref{fig:training}, the main component of \sys is the anonymizer $\mathcal{G}$. $\mathcal{G}$ takes the original audio $x$ and the threshold $\epsilon$ as inputs, and creates the anonymized audio in one shot. We first introduce the model architecture of $\mathcal{G}$, and then elaborate the training process of $\mathcal{G}$. 

\begin{figure}[tt]
\centering
\setlength{\abovecaptionskip}{5pt}% 
\setlength{\belowcaptionskip}{-0.3cm}%
\includegraphics[width=3in, trim=250 40 250 50, clip]{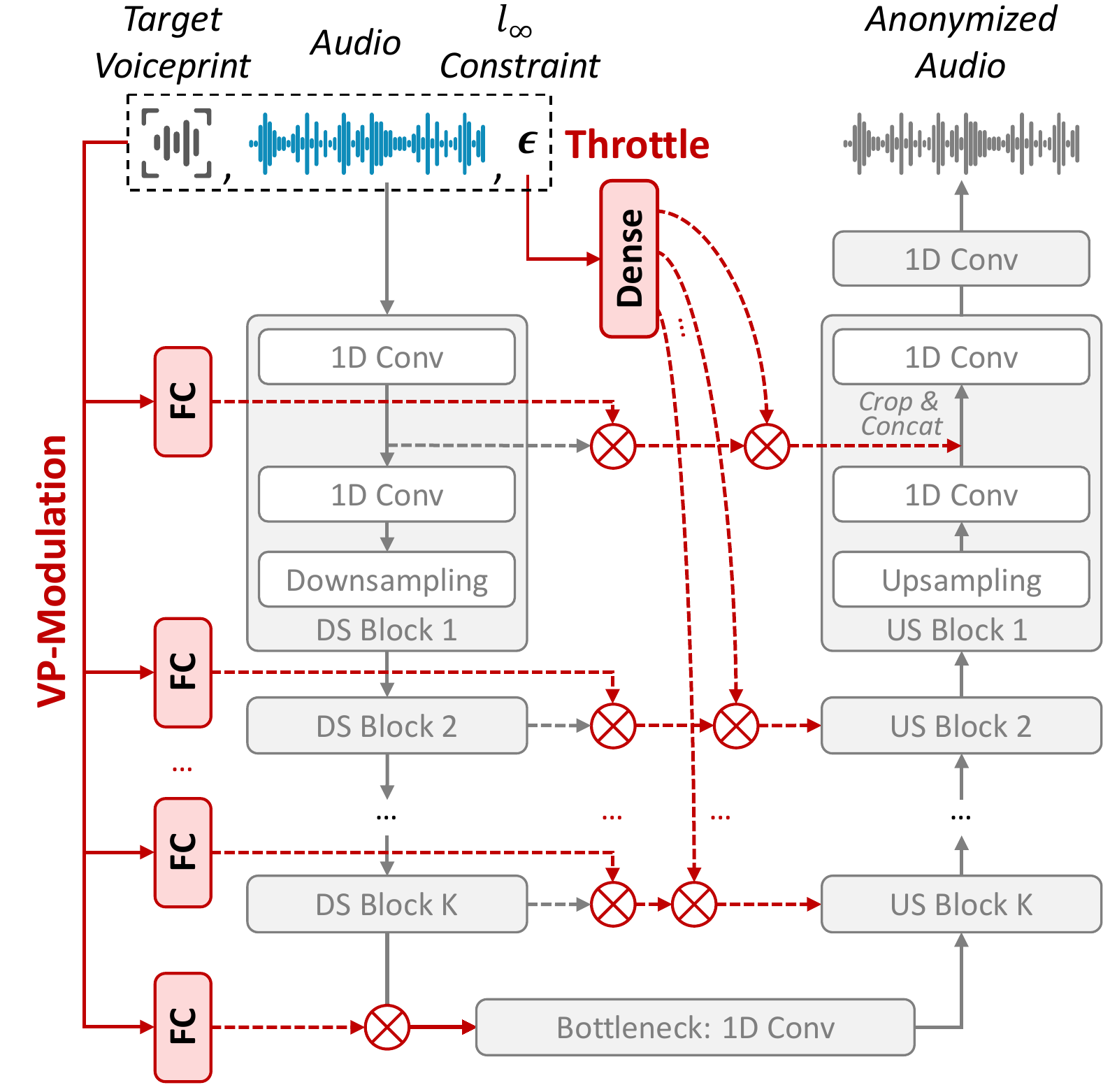}
\caption{The design of the anonymizer $\mathcal{G}$. \emph{VP-Modulation} modulates the elements in the feature vector extracted by each downsampling (DS) block (the same frequency level) based on the target voiceprint. \emph{Throttle} adjusts features at different frequency levels according to the constraint $\epsilon$.}
\label{fig:waveunet}
\end{figure}
\raggedbottom

\subsection{Anonymizer Design}

To realize real-time voice anonymization, we create anonymized audios through one-shot generation instead of iterative updates. We develop a generative model-based anonymizer based on Wave-U-Net~\cite{DBLP:conf/ismir/StollerED18}, as shown in Figure~\ref{fig:waveunet}. 

Wave-U-Net is originally used for audio source separation~\cite{DBLP:conf/ismir/StollerED18}. The vanilla Wave-U-Net is U-shaped with the downsampling (DS) and the upsampling (US) sub-networks, as illustrated by the grey blocks in Figure~\ref{fig:waveunet}. The input audio passes through a sequence of DS blocks, where a deeper DS block extracts a longer feature vector at a lower frequency level. There is a shortcut that transports the output of the first convolutional layer in each DS block to the final convolutional layer in each US block to combine the features at different frequency levels.

The anonymizer of \sys innovates Wave-U-Net in two ways with the \emph{VP-Modulation} and the \emph{Throttle} modules as shown in Figure~\ref{fig:waveunet}. 

\subsubsection{VP-Modulation}\label{subsec:sight}

% Wave-U-Net was used in FAPG~\cite{xie2020enabling} to 
Wave-U-Net in the previous work~\cite{xie2020enabling} creates targeted ASV adversarial examples by converting the audio of the original speaker towards a target speaker with a \emph{feature map} of the target speaker inserted in the \emph{bottleneck} layer at the bottom of the Wave-U-Net. Unfortunately, this suffers from two limitations. First, the feature map of every potential target speaker needs to be trained from scratch with relatively high overhead. The feature map needs to be replaced if another speaker is targeted, thus targeting an untrained speaker is not possible. Second, the feature map resides on the bottom of the Wave-U-Net, which means that the feature map represents the feature at the lowest frequency level. This lowest-frequency feature is too coarse-grained to capture the distinctive traits of different speakers, limiting the ability of the model to target a larger pool of speakers. 
% With a larger pool of target speakers, the extracted feature maps of different speakers are close, and the converted audio is far from the target speaker.  

To address this limitation, we design \emph{VP-Modulation} to guide the audio conversion process by the target voiceprint. As shown in Figure~\ref{fig:waveunet}, at each frequency level, the target voiceprint $v$ is transformed by a fully-connected layer (FC) into a modulation vector with a dimension consistent with the output at each shortcut and the output of the last DS block. The modulation vector rescales the shortcut features extracted by each DS block at each frequency level. In this way, the trained bottleneck layer needs no modification for a new target. In addition, the voiceprint of any target speaker can be fed into the network to realize a targeted anonymization without the need for re-training.  

\subsubsection{Throttle}\label{subsec:throttle}
% The Wave-U-Net network used in FAPG 
Wave-U-Net in the previous work~\cite{xie2020enabling} imposes a fixed constraint on the converted audio during training to ensure that the difference between the converted audio and the original audio is beneath a threshold $\epsilon$. However, during the voice anonymization phase, the threshold $\epsilon$ cannot be flexibly changed according to different requirements.

To cope with this problem, we design \emph{Throttle} to learn to adapt anonymization perturbations under different constraints during training. {In particular, \emph{Throttle} takes constraint $\epsilon$ as input, and computes a $K$-dimensional adjustment vector, where $K$ is the number of DS/US blocks. The adjustment vector controls the magnitude of each shortcut feature during combination. The output perturbation will be clipped to conform to the $l_\infty$-norm constraint of $\epsilon$, and back-propagate the loss to the \emph{Throttle} to alter the adjustment vector. In this way, the \emph{Throttle} learns the optimal adjustment vector under different $l_\infty$-norm constraints. } Note that the modulation vector produced by \emph{VP-Modulation} weights the elements in the feature vector at a specific frequency level, and the adjustment vector produced by \emph{Throttle} weights the feature vectors at different frequency levels.

\subsection{Anonymizer Training}

To train the anonymizer $\mathcal{G}$ to fulfil the performance goals, we materialize the loss function in Equation (\ref{equ:obj3}) as follows.

\subsubsection{Anonymity Loss}

In $L_{\mathrm{ASV}}$, we utilize cosine similarity to measure the resemblance of two voiceprints $\mathcal{S}(\cdot, \cdot)$. In general, the voiceprint extraction function $\mathcal{V}$ can follow any ASV. In our experiments, we utilize the state-of-the-art DNN-based ASV, ECAPA-TDNN~\cite{desplanques2020ecapa}, to encode varied-length audios into fixed-length voiceprints, $\mathcal{V}:\mathbb{R}^{1\times \cdot}\rightarrow \mathbb{R}^{1\times 192}$. We demonstrate that the anonymized audio is transferable to different ASVs in our experiments. 

For untargeted anonymization, minimizing $L_{\mathrm{ASV}}$ pushes the voiceprint of $\tilde{x}$ away from that of $x$, such that the voiceprint of the anonymized audio cannot be matched with that of the original audio, i.e., unidentifiability is guaranteed. For targeted anonymization, minimizing $L_{\mathrm{ASV}}$ drives the voiceprint of $\tilde{x}$ towards the target voiceprint $v$, such that the anonymized audio cannot be matched with that of the original audio when a proper $v$ is selected, i.e., unidentifiability is guaranteed. Furthermore, for targeted anonymization, voiceprints of anonymized audios of two different speakers will be matched together (both matched with the same voiceprint $v$), i.e., unlinkability is guaranteed. A theoretical proof of unidentifiability and unlinkability is provided in Appendix \ref{sec:proof}.

\begin{figure}[tt]
\centering
\setlength{\abovecaptionskip}{5pt}% 
\setlength{\belowcaptionskip}{-0.3cm}%
\includegraphics[width=2.8in, trim=280 150 280 150, clip]{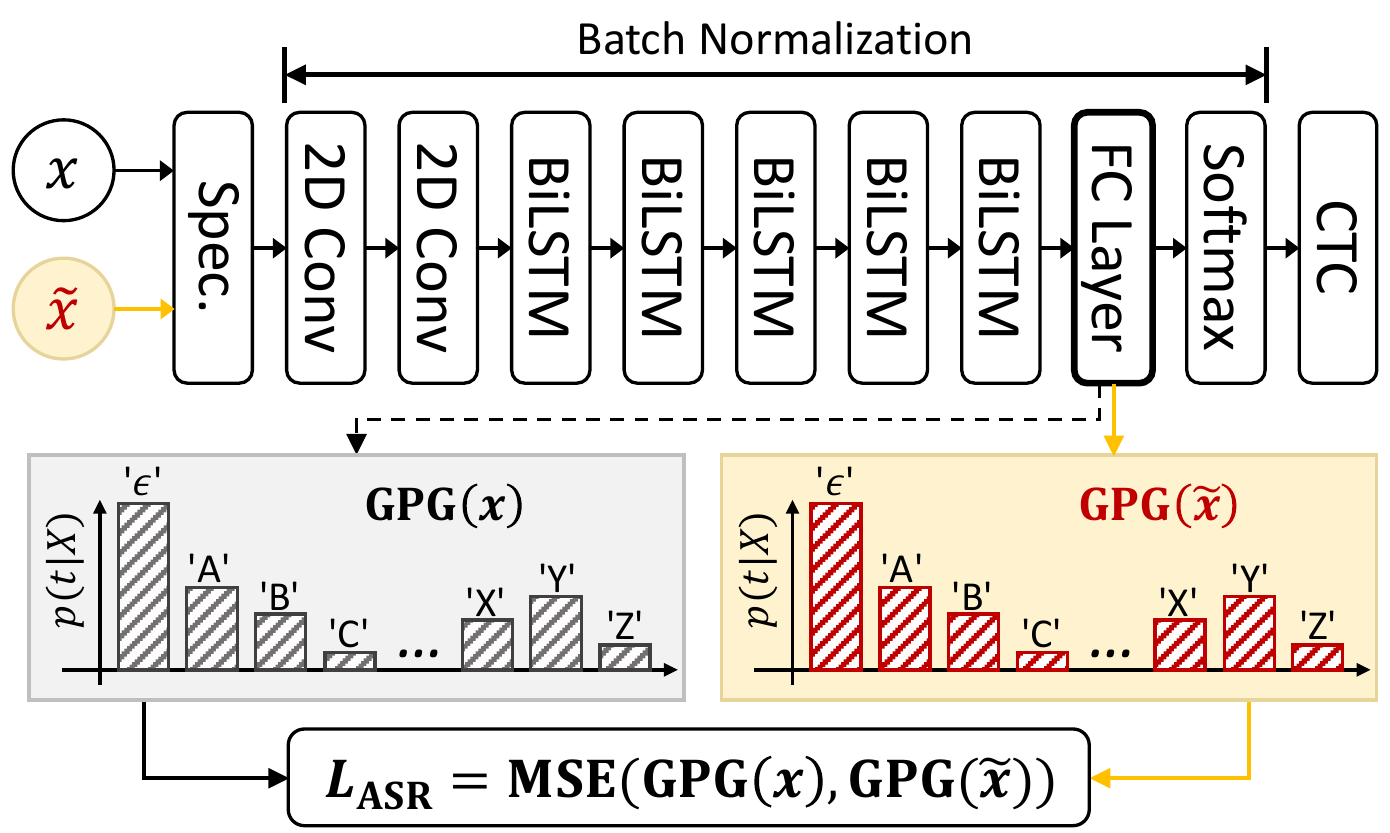}
\caption{Intelligibility Loss. Above is the structure of DeepSpeech2. Graphemic posteriorgram (GPG) output by the FC layer is used to constrain the linguistic distortion, and GPG loss is measured with the MSE between the original audio and the cloaked audio.}
\label{fig:asr_loss}
\end{figure}
\raggedbottom

\subsubsection{Intelligibility Loss}\label{subsec:asrloss}

For a similar reason, $L_{\mathrm{ASR}}$ can be instantiated with any ASR. In our experiments, we adopt DeepSpeech2~\cite{pmlr-v48-amodei16} with two 2D convolutional layers, five bidirectional LSTM layers, a fully-connected layer and a softmax layer, as shown in Figure~\ref{fig:asr_loss}. DeepSpeech2 uses graphemes, i.e., the smallest functional unit in a writing system in linguistics, as tokens. For the English language, the last softmax layer outputs the posterior probability of 29 tokens, i.e., a-z, space, apostrophe and the special $\phi$ token. The $\phi$ token indicates 'blank' or no label.

Connectionist temporal classification (CTC) loss is commonly used in ASR to train a sequence-to-sequence model when the alignment between the input spectral-frame sequence and the output token sequence is unknown. However, the CTC loss does not preserve the exact grapheme sequence. For instance, when the input is a three-frame audio, the final fully-connected layer in  DeepSpeech2 outputs $[a~\phi~b]$ and $[a~b~\phi]$, where $a, b$ and $\phi$ are the tokens of each frame. The CTC loss will ignore the blank token and reduce both sequences to $[a~b]$.

To better improve intelligibility, we utilize another loss term instead, i.e., the graphemic posteriorgram (GPG) loss. As shown in Figure~\ref{fig:asr_loss}, GPG loss is based on the posterior probability output by the softmax layer (we use the output of the preceding fully connected layer in practice).
% , normally used for phonetic posteriorgram (PPG) analysis~\cite{DBLP:conf/asru/HazenSW09}. 
Therefore, we have $L_{\mathrm{ASR}}(x, \tilde{x}) = \textrm{MSE}(\textrm{GPG}(x), \textrm{GPG}(\tilde{x}))$, where $\textrm{MSE}(\cdot, \cdot)$ is the mean squared error. We evaluate the effects of the CTC loss and the GPG loss in our ablation study in \S\ref{subsec:ablation}.

\subsubsection{Naturalness \& Timbre Loss}

We ensure naturalness and timbre preservation of the anonymized audio based on the psychoacoustic theory of masking effect. In particular, we leverage the spectral masking effect as the original audio and the anonymization perturbations are played at the same time. We treat the original audio as the \emph{masker}, which masks the presence of anonymization perturbations.

To materialize $L_{\mathrm{PSY}}$, we first compute the \emph{masking threshold}~\cite{DBLP:conf/icml/QinCCGR19} of the original audio, $\mathcal{M}(x)$, an $F$-dimensional vector, in which each element represents the maximum tolerable perturbation at a certain frequency level (a total of $F$ frequency components). Then $L_{\mathrm{PSY}}$ is computed as the sum of excesses of the perturbations in the anonymized audio. 
% \begin{equation}\label{equ:masking}
% \begin{aligned}
% L_{\mathrm{PSY}}(\tilde{x}, x)=\min\{\psi, \frac{1}{F}\mathrm{Relu}\big(\mathrm{PSD}(\tilde{x}-x)-\mathcal{M}(x)\big)\},
% \end{aligned}
% \end{equation}
\begin{equation}\label{equ:masking}
\begin{aligned}
L_{\mathrm{PSY}}(\tilde{x}, x)=\min\{\psi,~\frac{1}{F}\mathrm{max}\{0,~\mathrm{PSD}(\tilde{x}-x)-\mathcal{M}(x)\}\},
\end{aligned}
\end{equation}
where $\mathrm{PSD}(\cdot)$ computes the log-magnitude power spectral density, and $\mathrm{max}\{0, \cdot\}$ preserves the positive but not negative parts of a function. 
We constrain the magnitude of the loss by $\psi$, since we find in our experiments that the $L_{\mathrm{PSY}}$ term is unstable and of large variance, especially in the early stage of training.  Note that we further add an $l_2$-norm $\|\tilde{x}-x\|_2$ in the loss function in order to limit the energy of the anonymization perturbations.

\section{Evaluation}\label{sec:evaluation}

\subsection{Experiment Setup}

\textbf{Prototype}.
We have implemented a prototype of \sys on the PyTorch~\cite{DBLP:conf/nips/PaszkeGMLBCKLGA19} platform and trained the model according to Equation~(\ref{equ:obj3}) using two NVIDIA 3090 GPUs. We set the default configuration as $D=41,641$, $D_o=32,089$, $N=192$, $F=1,025$, $K=5$, and a batchsize of 64. The $l_\infty$-norm constraint, $\epsilon$, is sampled from a normal distribution $\mathcal{N}(\mu, \sigma)$ with the mean $\mu=0.05$ and the variance $\sigma=0.05$. {Note that we randomize $\epsilon$, the constraint value of $l_\infty$ to train the \emph{Throttle} module in \sys to learn to adjust the magnitude of each feature under different constraints. } 
In the training phase, we use an Adam~\cite{kingma2014adam} optimizer to update the parameters of the anonymizer $\mathcal{G}$ for 50 epochs, with a learning rate of 4e-4. The default adversary is A1. 

\begin{table}\centering
\begin{threeparttable}[t]
\setlength{\abovecaptionskip}{5pt}% 
\setlength{\belowcaptionskip}{0pt}%

\caption{{Dataset statistics.}}

\footnotesize
\setlength{\tabcolsep}{1.1mm}{
\begin{tabular}{@{}>{\hspace{0.1cm}}p{2.8cm}ccc>{\hspace{-0.1cm}\centering\arraybackslash}p{1.4cm}@{}}
\toprule
\textbf{Dataset}                       & \textbf{Subset}          & {\textbf{\#Speaker}} & \textbf{\#Utterance}  & \textbf{Duration (s)} \\ \midrule
LibriSpeech (English)                  & \textit{test-clean}      & {40}                 & 2,620                 & $1.3\sim35.0$         \\ \midrule
AISHELL (Chinese)                      & \textit{test}            & {20}                 & 7,176                 & $1.9\sim14.7$      \\ \midrule
CommonVoice (French)                   & \textit{test}            & {-$^*$}              & 5,000$^\dagger$       & $1.6\sim11.5$        \\ \midrule
CommonVoice (Italian)                  & \textit{test}            & {-$^*$}              & 5,000$^\dagger$       & $3.2\sim11.4$         \\ \bottomrule
\end{tabular}}

\begin{tablenotes}[flushleft]
\item[] \vspace{-1pt}\hspace{-2pt}\footnotesize (i) $^*$: CommonVoice datasets have no credible speaker identity information. \\(ii) $^\dagger$: We use the first 5,000 utterances in CommonVoice for evaluation.
\end{tablenotes}

\label{tab:dataset}
\end{threeparttable}
\end{table}
\begin{table}\centering
\begin{threeparttable}[t]
\setlength{\abovecaptionskip}{5pt}% 
\setlength{\belowcaptionskip}{0pt}%
\setlength\tabcolsep{4.5pt}
\caption{ASVs used for evaluation.}

\footnotesize
\setlength{\tabcolsep}{2.3mm}{
\begin{tabular}{@{}llllr@{}}
\toprule
\textbf{Model} & \textbf{Alias} & \textbf{Category} & \textbf{Source}	&  \textbf{EER$^\dagger$(\%)}	\\ \midrule
ECAPA-TDNN     & \textbf{EP}    & DNN-based         & SpeechBrain   	& 0.70~                      	\\
X-vector         & \textbf{XV}    & DNN-based         & SpeechBrain    	& 6.53~                      	\\
GMM-UBM        & \textbf{GMM}   & Statistical       & Kaldi           	& 11.39~                     	\\
ivector-PLDA   & \textbf{IV}    & Statistical       & Kaldi         	& 6.03~                          \\
iFlytek        & \textbf{IF}    & Commercial        & iFlytek        	& 9.44~                          \\ \bottomrule
\end{tabular}}

\begin{tablenotes}[flushleft]
\item[] \vspace{-1pt}\hspace{-2pt}\footnotesize $^\dagger$: We test the EERs of the five ASVs on \emph{test-clean} of LibriSpeech (B0).
\end{tablenotes}

\label{tab:ASVs}
\end{threeparttable}
\end{table}
\noindent\textbf{Dataset}. We train \sys on VoxCeleb1~\cite{Nagrani17}, an English dataset with 352-hour audios from 1,251 speakers. Four widely-used datasets are adopted to evaluate the effectiveness of \sys, i.e., LibriSpeech (English)~\cite{panayotov2015librispeech}, AISHELL (Chinese)~\cite{bu2017aishell}, CommonVoice (French)~\cite{ardila2019common}, and CommonVoice (Italian)~\cite{ardila2019common}. We use the test sets of the four datasets for evaluation. Note that CommonVoice datasets are used for ASR with no 
credible speaker identity information, so we only use them to test whether the anonymized audios maintain intelligibility. More details about the datasets are listed in Table~\ref{tab:dataset}.

\noindent{\textbf{ASV}. As shown in Table~\ref{tab:ASVs}, we use five widely-used ASVs to test the effectiveness and transferability of \sys, i.e., ECAPA-TDNN, X-vector~\cite{snyder2018x}, ivector-PLDA~\cite{dehak2010front}, GMM-UBM~\cite{reynolds2000speaker} and iFlytek ASV~\cite{iFlytekOP}. ECAPA-TDNN and X-vector are DNN-based ASVs implemented using SpeechBrain~\cite{speechbrain} and trained on VoxCeleb1\&2. GMM-UBM and ivector-PLDA are traditional statistical model-based ASVs implemented using Kaldi toolkit~\cite{povey2011kaldi} and trained on VoxCeleb1\&2. iFlytek ASV is a commercial ASV API provided by iFlytek Open Platform~\cite{iFlytekOP}, with private training data that is unknown. We download pretrained models of the five ASVs. The baseline performance of the five ASVs presented in Table~\ref{tab:ASVs} is tested on the \emph{test-clean} set of LibriSpeech.
}

\noindent\textbf{ASR}. We evaluate the decoding error of anonymized audios on eleven ASRs, including five English ones, two Chinese ones, two French ones and two Italian ones. These ASRs are trained on different datasets or have different architectures. More details about the ASRs are listed in Table~\ref{tab:ASRs} (Appendix). The baseline performance of the ASRs presented in Table~\ref{tab:ASRs} is tested on the corresponding datasets in Table~\ref{tab:dataset}.

\begin{figure*}[t]
    \centering
\setlength{\abovecaptionskip}{0pt}
\setlength{\belowcaptionskip}{0cm}
% \subfigcapskip = -0.05cm
\includegraphics[width=7in]{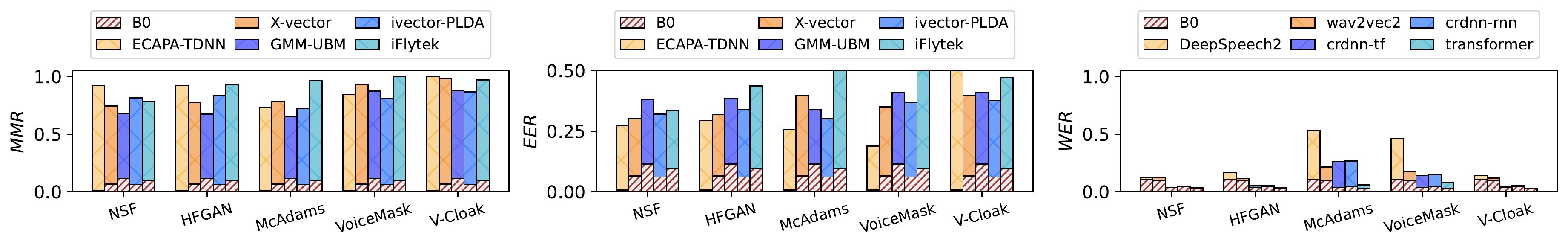}

\caption{{Comparison with existing works. (a) MMR. (b) EER. (c) WER. \sys yields the highest average MMR of 94.02\% and the highest average EER of 46.10\%. \sys obtains a low average WER of 7.65\% second only to the NSF (7.19\%).}}\label{fig:comparison}
\end{figure*}

\noindent\textbf{Evaluation metrics}.
Six metrics are used to evaluate the performance of voice anonymization. 
\begin{itemize}
    \item \emph{Miss-Match Rate (MMR)}, the probability that the voiceprint of the anonymized audio cannot be matched with that of the original speaker by the ASV, which is like the False Rejection Rate (FRR) of the ASV. 
    \item \emph{Wrong-Match Rate (WMR)}, the probability that the voiceprint of the anonymized audio is matched with that of a wrong speaker, which is like the False Acceptance Rate (FAR) of the ASV.
    \item \emph{Equal Error Rate (EER)}, the rate at which MMR equals WMR (FRR equals FAR), which measures the overall anonymization power. 
    \item \emph{Word Error Rate (WER)/Character Error Rate (CER)}, metrics that measure the differences between the transcription given by the ASR and the ground-truth. WER (resp. CER) is calculated as,
\begin{equation*}\label{equ:wer}
\begin{aligned}
\textrm{WER~(resp.~CER)} = \frac{N_{sub}+N_{del}+N_{ins}}{N_{ref}},
\end{aligned}
\end{equation*}

\noindent where $N_{sub}$, $N_{del}$ and $N_{ins}$ are the numbers of substitution, deletion, and insertion errors of words (resp. characters), respectively. $N_{ref}$ is the ground-truth number of words (resp. characters).

\item \emph{Signal-to-noise Ratio (SNR)}, computed as $\mathrm{SNR}(\mathrm{dB}) = 10\log_{10}(P_x/P_\delta)$, where $\delta=\tilde{x}-x$, $P_x$ and $P_\delta$ are the average power of the original audio and the anonymization perturbation, respectively. SNR is used to evaluate the objective naturalness of the anonymized audios. We also evaluate the subjective naturalness and timbre of the anonymized audios with a user study.
\item \emph{Real-time coefficient (RTC)}, computed as $\textrm{RTC}=T_{cvt}/T_{audio}$, where $T_{audio}$ is the duration of the original audio, and $T_{cvt}$ is the time to anonymize the audio. RTC is used to evaluate the efficiency of the voice anonymization system. 

{MMR (i.e., false negative/rejection rate) and WMR  (i.e., false positive/acceptance rate) are commonly-used metrics for speaker verification systems. EER and WER (CER) 
% are standard evaluation metrics in VoicePrivacy 2022 Challenge \cite{vpc2022} and 
are widely used by existing works, e.g., NSF, HFGAN, and McAdams.  SNR is a commonly-used metric to measure the imperceptibility of adversarial perturbations. RTC is often used to gauge the efficiency of voice processing methods \cite{qian2018hidebehind}. To achieve desirable anonymization performance, MMR, WMR, and EER are better to be higher. To maintain intelligibility, WER is better to be lower, and the SNR is better to be higher. To realize real-time voice anonymization, RTC is better to be lower.}

\end{itemize}

\noindent\textbf{Baselines}.
We compare \sys with the baselines. The performance of ASVs and ASRs with clean/unprocessed audios is referred to as {B0}. Moreover, we reproduce four state-of-the-art voice anonymization methods, i.e., NSF~\cite{fang2019speaker}, HFGAN~\cite{DBLP:journals/corr/abs-2202-13097}, McAdams~\cite{DBLP:conf/interspeech/0001TTNE21}, and VoiceMask~\cite{qian2018hidebehind, qian2019speech}. Details of the baselines are in the Appendix \ref{appdix:baseline}.

\begin{figure}[t]
    \centering
\setlength{\abovecaptionskip}{0pt}
\setlength{\belowcaptionskip}{0cm}
% \subfigcapskip = -0.05cm
\hspace{-0.5cm}\includegraphics[width=3in, trim=5 5 5 5, clip]{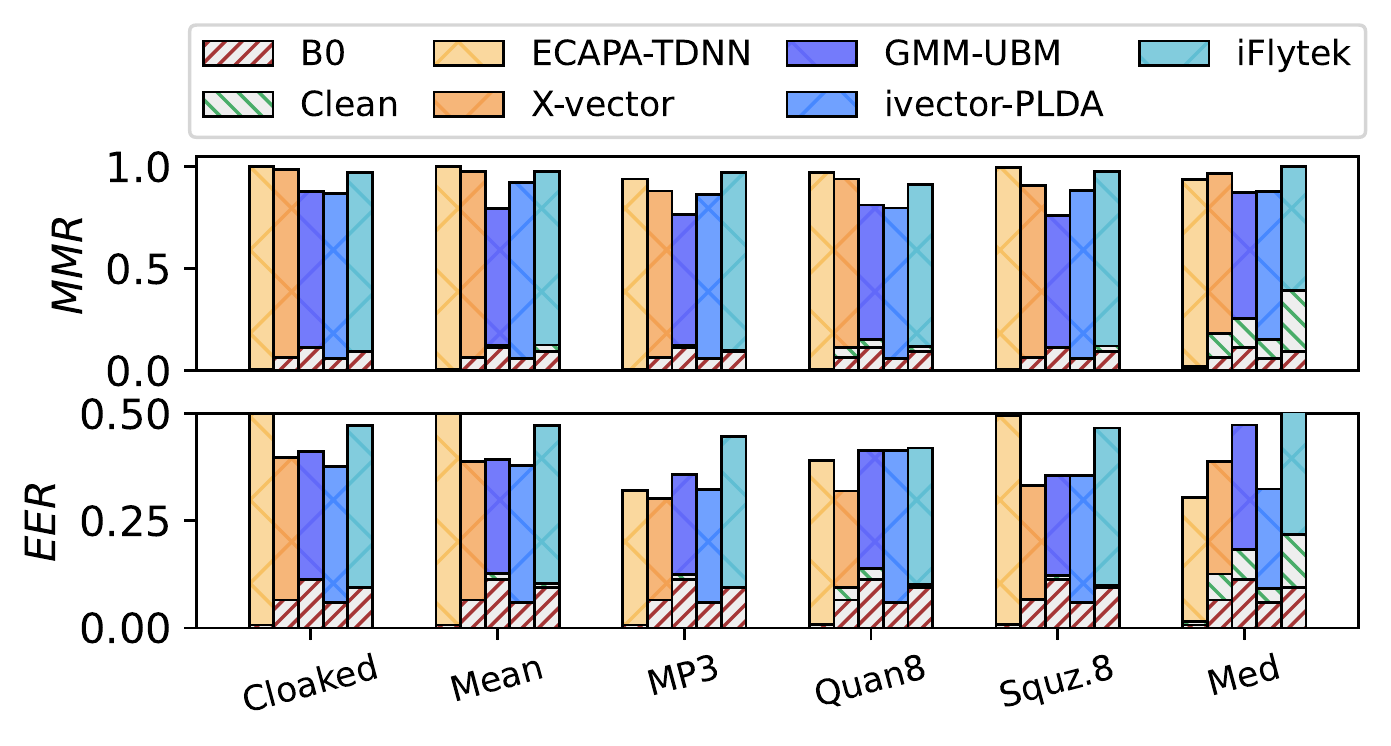}

\caption{{Unidentifiability under adversary A2. The most effective de-noising method, MP3 compression, only causes a decrease in the MMR of 3.27\% and the EER of 9.26\%.}}\label{fig:denoise}
\end{figure}

\noindent{\textbf{Hyperparameter tuning}. Hyperparameters affect the convergence, performance and generalization of the anonymizer. When balancing different optimization goals of anonymity, intelligibility and naturalness, we need to tune the hyperparameters $\alpha$, $\beta$ and $\gamma$. We adopt a stepwise hyperparameter tuning, which increases the weights of other loss terms as the anonymity loss decreases. When $L_{\mathrm{ASV}}\geq 0.3$, $\alpha=0.5, \beta=$1e-6, $\gamma=0.1$. When $0.15\leq L_{\mathrm{ASV}}<0.3$, $\alpha=0.8, \beta=$2e-6, $\gamma=0.1$. When $L_{\mathrm{ASV}}<0.15$, $\alpha=1, \beta=$3e-6, $\gamma=0.1$. }

\subsection{Comparison With Existing Works}\label{subsec:comparsion}
As shown in Figure~\ref{fig:comparison} (Table~\ref{tab:existingwork} in the Appendix), we compare \sys with four existing anonymization methods on the \emph{test-clean} set of LibriSpeech against adversary A1. {\sys ($\epsilon=0.1$) achieves the highest average MMR of 94.02\% and the highest average EER of 46.10\%.} {In the worst case where the attacker adopts the best ASV, the lowest EER of \sys is still 8.23\% higher than the lowest EER of all other baselines.} \sys gives a low average WER of 7.65\% second only to NSF (7.19\%). Although NSF achieves the lowest WER, it only gains an MMR of 78.74\% and an EER of 32.27\%. Although VoiceMask has high unidentifiable performance, it causes severe linguistic distortions, with an average WER of 20.05\%. Overall, \sys provides the best anonymity while maintaining a low ASR decoding error (high intelligibility).

\textbf{Discussion.} As far as we are concerned, existing works on voice anonymization have not standardized the threshold of EER for \emph{enough} anonymization.  
% The GDPR \cite{GDPR} defines anonymous data as data that ``does not relate to an identified or identifiable natural person or to personal data rendered anonymous" so ``the data subject is not or no longer identifiable." 
% The upper bound on the EER is 0.5 for perfect anonymization (equivalent to random guess), which may not be realized in practice. 
We expect that in the future, standards on data anonymization, especially on biometric data, will be implemented, e.g., by NIST. 

\subsection{Cross-Language Performance}
Apart from English, we also test \sys on Chinese, French and Italian datasets. Note that \sys is trained only on the English ASV and ASR. 

As shown in Figure~\ref{fig:chinese} and Table~\ref{tab:chinese} in the Appendix, \sys can effectively anonymize Chinese audios, with an average MMR of 98.19\% and an EER of 51.44\%. The cross-language transferability of \sys may be attributed to the language-agnostic ECAPA-TDNN used in the training process. Moreover, we test \sys with two Chinese ASRs, and the CER results are presented in Figure~\ref{fig:chinese}. \sys induces a CER increase of only 2.68\% with $\epsilon=0.1$. Due to a lack of identity information provided by CommonVoice datasets, we only use the French and the Italian datasets for ASR intelligibility test. As shown in Figure~\ref{fig:fr_it} (Table~\ref{tab:fr}\&\ref{tab:it} in the Appendix), \sys leads to a WER increase of 6.59\% with French audios (from 16.33\% to 22.91\%) and 6.55\% with Italian audios (from 15.11\% to 21.66\%). 

It demonstrates that even if only an English dataset is used for training, the intelligibility-preserving property of \sys can  generalize to significantly different languages, i.e., Chinese, French, and Italian in our case.

\subsection{{Unidentifiability Under Adaptive Attacker A2}}\label{subsec:robustness}
{We compare the unidentifiability of \sys and baselines against adversary A2, who applies de-noising techniques to try to remove the anonymization perturbation. We consider six methods that are commonly used to remove adversarial perturbations, i.e., smoothing (mean filter and median filter with a kernel of 3, band-pass filter with a passband of 50$\sim$7,500Hz), quantization (from 32-bits to 8-bits), audio squeezing (0.8$\times$ the sampling rate), and MP3 compression. The results are shown in Figure~\ref{fig:denoise} (Table~\ref{tab:meanfilter}-\ref{tab:mp3} in the Appendix). Most of the de-noising techniques have little influence on the effectiveness of \sys. Among them, MP3 compression has the most obvious influence, i.e., an average MMR decrease of 3.27\% (to 90.75\%) and an EER decrease of 9.26\% (to 36.84\%). Compared with the four baselines, \sys achieves the highest EERs in the worst-case scenario against all six de-noising methods. It indicates that the anonymized audios of \sys are robust in terms of unidentifiability under de-noising techniques. }

\begin{table*}\centering
\begin{threeparttable}[tt]
\setlength{\abovecaptionskip}{5pt}% 
\setlength{\belowcaptionskip}{0pt}%

\caption{The performance under adaptive attacker A3.}

\scriptsize
\setlength{\tabcolsep}{1mm}{
\begin{tabular}{cl|>{\centering}p{1.5cm}|>{\centering}p{1.5cm}|>{\centering}p{1.5cm}|>{\centering}p{1.5cm}|>{\centering}p{1.5cm}|ccc}
\hline
\multicolumn{2}{c|}{\multirow{2}{*}{\textbf{Model}}} & \multicolumn{1}{c|}{\multirow{2}{*}{\textbf{B0 (\%)}}}      & \multicolumn{1}{c|}{\multirow{2}{*}{\textbf{NSF (\%)}}}     & \multicolumn{1}{c|}{\multirow{2}{*}{\textbf{HFGAN (\%)}}} & \multicolumn{1}{c|}{\multirow{2}{*}{\textbf{McAdams (\%)}}} & \multicolumn{1}{c|}{\multirow{2}{*}{\textbf{VoiceMask (\%)}}}         & \multicolumn{3}{c}{\textbf{\sys (\%)}}                                                                                                                                     \\
\multicolumn{2}{c|}{}                                & \multicolumn{1}{c|}{}                                       & \multicolumn{1}{c|}{}                                       & \multicolumn{1}{c|}{}                                   & \multicolumn{1}{c|}{}                                   & \multicolumn{1}{c|}{}                                          & \multicolumn{1}{c}{\textbf{Untargeted}}               & \multicolumn{1}{c}{\textbf{Targeted w/o key$^\dagger$}}                          & \multicolumn{1}{c}{\textbf{Targeted w key$^\dagger$}}                      \\ \hline
\multirow{6}{*}{\textbf{ASV}}      & \textbf{EP}     & 0.70                                  & 24.50                                    & 24.79                                    & 9.01                                        & 8.80                                             & 22.03                                                 & 46.93                                                    & 32.98                                                 \\
                                   & \textbf{XV}     & 6.53                                  & 27.56                                    & 27.11                                    & 9.13                                        & 14.75                                            & 18.20                                                 & 44.50                                                    & 29.67                                                 \\
                                   & \textbf{GMM}    & 11.39                                 & 30.36                                    & 31.40                                    & 23.33                                       & 32.58                                            & 39.50                                                 & 43.44                                                    & 51.15                                                 \\
                                   & \textbf{IV}     & 6.03                                  & 11.17                                    & 26.51                                    & 8.76                                        & 18.25                                            & 32.90                                                 & 40.42                                                    & 37.22                                                 \\
                                   & \textbf{IF}     & 9.44                                  & 28.24                                    & 27.85                                    & 13.31                                       & 18.95                                            & 19.93                                                 & 37.56                                                    & 31.34                                                 \\\hline
                                   & \textbf{AVG}    & 6.82                                  & 24.37                                    & 27.53                                    & 12.71                                       & 18.67                                            & 26.51                                                 & 42.57                                                    & 36.47                                                 \\ 
                                   & \textbf{WCS}    & -                                     & 11.17                                    & 24.79                                    & 8.76                                        & 8.80                                             & 18.20                                                 & 37.56                                                    & 29.67                                                 \\ \hline
\end{tabular}}

\begin{tablenotes}[flushleft]
\item[] \vspace{-1pt}\hspace{-2pt}\footnotesize (i) $^\dagger$: w/ or w/o key means that the voiceprint of the target speaker is known or unknown to the adversary. (ii) \textbf{AVG}: the average-case scenario, \textbf{WCS}: the worst-case scenario. \textbf{EP}: ECAPA-TDNN, \textbf{XV}: X-vector, \textbf{GMM}: GMM-UBM, \textbf{IV}: ivector-PLDA, \textbf{IF}: iFlytek.
\end{tablenotes}

\label{tab:A3}
\end{threeparttable}
\end{table*}

\subsection{{Unlinkability Under Adaptive Attacker A3}}
{We compare the performance of \sys and baselines against adaptive attacker A3, who has the anonymizer $\mathcal{G}$ and attempts to link the voiceprint of an anonymized audio to that of an anonymizer-processed audio. As shown in Table~\ref{tab:A3}, with the untargeted anonymization, \sys achieves the second highest anonymization performance (only slighly worse than HFGAN). With the targeted anonymization, \sys achieves the highest EER among all methods. For adversary A3, without the voiceprint key, the anonymizer-processed audio (enrollment audio) still cannot be matched with the anonymized audio (test trial), achieving a highest EER of 42.57\% in the average case and 37.56\% in the worst case. In the most challenging scenario where the adversary has access to the anonymizer and the voiceprint key $v$, audio samples anonymized with $v$ from any speakers (not only the original speaker of the anonymized audio, but also other speakers)  will be matched to the same speaker $v$, thus producing a highest EER of 36.47\%. Overall, our experiment results show that unlinkability is achieved in the face of adversary A3. To verify the targeted anonymization effectiveness of \sys, we convert audios of \textit{test-clean} to 10 speakers in \textit{dev-clean}, and test the success rate. As shown in Table \ref{tab:exactlymatch}, 99.99\% of the converted audios are successfully matched with the target speakers.}
% high MMR and EER against adversary A1. As adversary A3 enrolls in the ASV with the anonymizer-processed samples of potential speakers, the WMR is as high as 99.67\%, which means that the anonymized audio is matched with almost all potential speakers by the ASV (the ASV cannot distinguish the original speaker from all other potential speakers). Compared with untargeted anonymization, the targeted anonymization achieves a higher MMR of 98.97\% and a higher EER of 52.58\% against adversary A1. For adversary A3, without the voiceprint key, the anonymizer-processed audio still cannot be matched with the original audio, leading to a high MMR of 95.10\% and a high EER of 43.95\%. In the most challenging case where the adversary has access to the anonymizer and the voiceprint key $v$, all samples anonymized with $v$ will be matched to the same speaker $v$, thus producing a WMR of 99.99\% and an EER of 33.29\%. Overall, our experiment results show that unlinkability is achieved in the face of adversary A3. (to be modified)

\begin{figure}[t]
    \centering
\setlength{\abovecaptionskip}{0pt}
\setlength{\belowcaptionskip}{0cm}
% \subfigcapskip = -0.05cm
\includegraphics[width=2.3in]{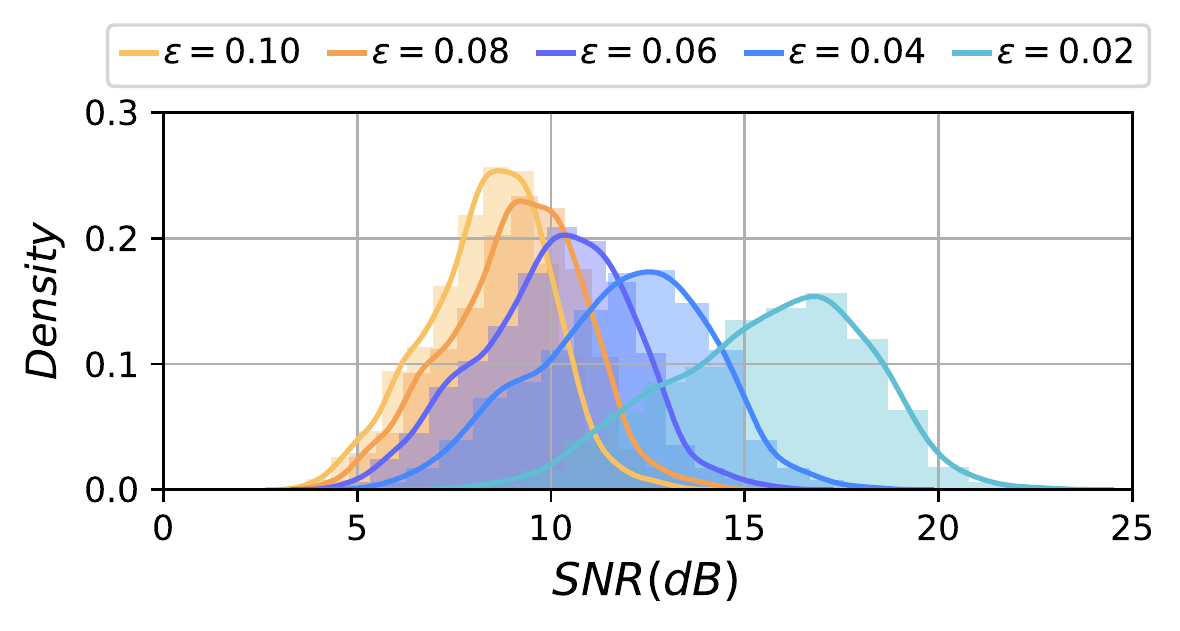}

\caption{The SNRs at different anonymization levels.}\label{fig:snr}
\end{figure}

\subsection{Anonymization Levels}\label{subset:level}
As shown in Figure~\ref{fig:epsilon} and Table~\ref{tab:eps} in the Appendix, we test the performance of \sys at different anonymization levels, i.e., vary the $\epsilon$ from $0.02\sim0.10$, on the \emph{test-clean} set of LibriSpeech. We can observe that the MMRs and EERs of \sys are much higher than those of clean audios (B0). As $\epsilon$ increases, i.e., larger anonymization perturbations, both MMRs and EERs increase. Note that although the anonymizer of \sys is trained using the ECAPA-TDNN ASV, the anonymization is also effective for the other four ASVs including a commercial ASV, iFlytek, whose architecture and training set are unknown. The transferability is gained probably because we optimize at the intermediate voiceprint layer rather than the last classification layer.

As for intelligibility, we can see that only a slight increase of WER is induced by the anonymization perturbation, i.e., increased by 1.46\% from 6.19\% to 7.65\% ($\epsilon=0.1$). Although only the DeepSpeech2 ASR is used in the training process, the intelligibility of \sys generalizes to other ASRs of different architectures and training sets.

As shown in Figure~\ref{fig:snr} in the Appendix. The average SNRs of  anonymized audios range from 8.4dB ($\epsilon=0.1$) to 15.5dB ($\epsilon=0.02$). The lower the $\epsilon$, the larger the average SNR and the variance, which means that for those \emph{hard-to-anonymize} audios, \sys adaptively generates larger anonymization perturbations, thus \sys can maintain high MMR (87.86\%) even when $\epsilon=0.02$.

\begin{figure*}[tt]
    \centering
\setlength{\abovecaptionskip}{-0.0cm}
\setlength{\belowcaptionskip}{-0.0cm}
\subfigcapskip = -0.25cm
\subfigure{
    \vspace{2cm}\hspace{0.5cm}
    \includegraphics[height=0.6cm, trim=0 155 0 0, clip]{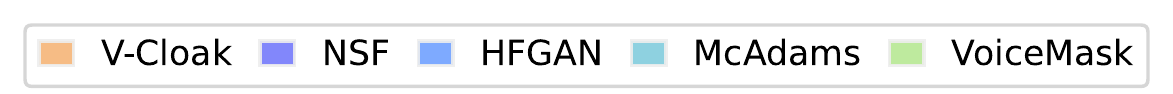}
    % \label{fig:legend}
    }
\\\vspace{-0.5cm}\setcounter{subfigure}{0}
\subfigure{
    \includegraphics[width=6.5in, trim=0 0 0 0, clip]{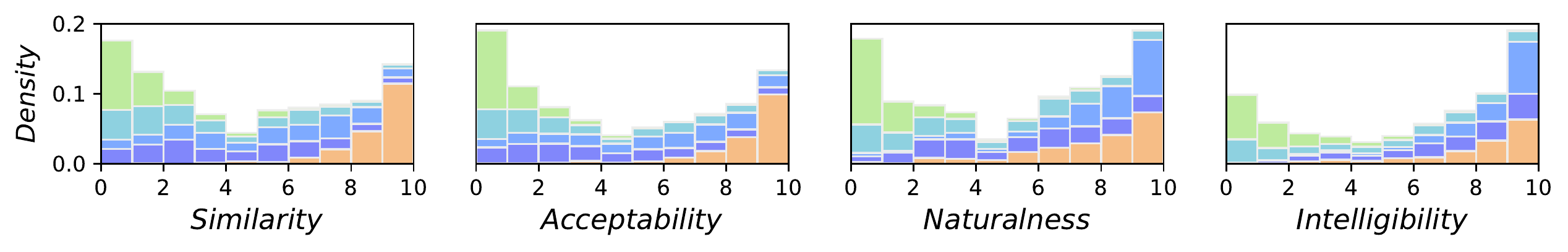}
    }
\vspace{-0.0cm}\caption{{The results of user study.}}\label{fig:userstudy}
\end{figure*}

% \subsection{{Population Size}}
% {We conduct experiments to evaluate the impact of population size (i.e., the number of speakers in the test set) on the anonymization performance. As shown in Table~\ref{tab:population} in the Appendix, the population size has little influence on the performance of \sys. The possible reason is that the application of \sys is to make ASVs unable to verify a certain individual from all other individuals, which is little affected by the total population size.}

\subsection{Ablation Study}\label{subsec:ablation}
%We conduct ablation studies to verify the effectiveness of our intelligibility loss $L_{\mathrm{ASR}}$ and the \emph{Throttle} module. 

\emph{Intelligibility loss.} As shown in Table~\ref{tab:NoASR} in the Appendix, without the intelligibility loss, the anonymizer can achieve a higher MMR of 97.16\% and a higher EER of 53.53\%, since the constraint on the anonymization process is looser. However, the intelligibility is greatly reduced (an average WER of 107.66\%) without the intelligibility loss term. In addition, we evaluate the effectiveness of the CTC loss  and the GPG loss. We find that GPG is more effective in decreasing the ASR decoding error and stabilizing the convergence of ASV loss.

\emph{Throttle.} We remove \emph{Throttle} from the anonymizer and present the results in Figure~\ref{fig:wothrottle} and Table~\ref{tab:throttle} in the Appendix. We can see that removing \emph{Throttle} induces an obvious performance degradation on GMM-UBM, ivector-PLDA, and iFlytek, which demonstrates the necessity of \emph{Throttle}.

\subsection{{Data Separation}}

{In the default experiment setup, the training datasets of \sys and baselines are partially overlapped with the training datasets of ASVs (note that the test datasets are not overlapped with any training datasets). In this section, we examine the performance of \sys and baselines when their training datasets do not overlap with those of the ASVs. More specifically, we train \sys on VoxCeleb1\&2 and LibriSpeech (train-clean-100, train-other-500). NSF and HFGAN are trained on VoxCeleb1\&2, LibriSpeech (train-clean-100, train-other-500), and LibriTTS (train-clean-100). McAdams and VoiceMask are model-free methods and do not need to be trained.
For ASVs, we train ECAPA-TDNN, X-vector, and DeepSpeaker \cite{li2017deepspeaker} all on LibriSpeech (train-clean-360). The test set is LibriSpeech (test-clean). As shown in Table \ref{tab:separation}, \sys achieves the highest average and worst-case MMR and EER compared with four baselines. This verifies the transferability of \sys, i.e., \sys is effective in fooling ASVs trained on totally different datasets.}

\begin{table*}\centering
\begin{threeparttable}[tt]
\setlength{\abovecaptionskip}{5pt}% 
\setlength{\belowcaptionskip}{0pt}%

\caption{Comparison with existing works under data separation.}

\scriptsize
\setlength{\tabcolsep}{1.5mm}{
\begin{tabular}{cl|c|crr|crr|crr|crr|crr}
\hline
\multicolumn{2}{c|}{\multirow{2}{*}{\textbf{Model}}} & \textbf{B0 (\%)}  & \multicolumn{3}{c|}{\textbf{NSF (\%)}}                                                    & \multicolumn{3}{c|}{\textbf{HFGAN (\%)}}                                                    & \multicolumn{3}{c|}{\textbf{McAdams (\%)}}                                                    & \multicolumn{3}{c|}{\textbf{VoiceMask (\%)}}                                                     & \multicolumn{3}{c}{\textbf{\sys (\%)}}                                                      \\
\multicolumn{2}{c|}{}                                & \textbf{EER} & \textbf{MMR}               & \multicolumn{1}{c}{\textbf{WMR}} & \multicolumn{1}{c|}{\textbf{EER}} & \textbf{MMR}              & \multicolumn{1}{c}{\textbf{WMR}} & \multicolumn{1}{c|}{\textbf{EER}} & \textbf{MMR}              & \multicolumn{1}{c}{\textbf{WMR}} & \multicolumn{1}{c|}{\textbf{EER}} & \textbf{MMR}               & \multicolumn{1}{c}{\textbf{WMR}} & \multicolumn{1}{c|}{\textbf{EER}} & \textbf{MMR}               & \multicolumn{1}{c}{\textbf{WMR}} & \multicolumn{1}{c}{\textbf{EER}} \\ \hline
\multirow{3}{*}{\textbf{ASV}}     & \textbf{EP}      & 3.72         & 88.89  & 3.89                             & 38.09                             & 87.33 & 3.89                             & 42.21                             & 46.53 & 3.89                            & 20.69                             & 70.15  & 3.89                            & 23.40                             & 97.90    & 3.89                            & 42.21                            \\
                                  & \textbf{XV}      & 5.74         & 87.33  & 4.73                             & 34.47                             & 88.89 & 4.73                             & 39.05                             & 84.28 & 4.73                            & 40.19                             & 95.73  & 4.73                            & 37.79                             & 100.0    & 4.73                            & 44.73                            \\
                                  & \textbf{DP}      & 3.72         & 93.97  & 4.05                             & 39.70                             & 89.39 & 4.05                             & 33.13                             & 80.15 & 4.05                            & 35.00                             & 99.24  & 4.05                            & 41.37                             & 99.77    & 4.05                            & 49.47                            \\ \hline
                                  & \textbf{AVG}     & 4.39         & 90.06  & 4.22                             & 37.42                             & 88.54 & 4.22                             & 38.13                             & 70.32 & 4.22                            & 31.96                             & 88.37  & 4.22                            & 34.19                             & 99.22    & 4.22                            & 45.47                            \\ 
                                  & \textbf{WCS}     & -            & 87.33  & 3.89                             & 34.47                             & 87.33 & 3.89                             & 33.13                             & 46.53 & 3.89                            & 20.69                             & 70.15  & 3.89                            & 23.40                             & 97.90    & 3.89                            & 42.21                            \\ \hline
\end{tabular}}

\begin{tablenotes}[flushleft]
\item[] \vspace{-1pt}\hspace{-2pt}\footnotesize \textbf{AVG}: average, \textbf{WCS}: worst-case scenario. \textbf{EP}: ECAPA-TDNN, \textbf{XV}: X-vector, \textbf{DP}: DeepSpeaker.
\end{tablenotes}

\label{tab:separation}
\end{threeparttable}
\end{table*}

\subsection{User Study}\label{subsec:userstudy}

To demonstrate the intelligibility-, naturalness-, and timbre-preserving properties of \sys, we conduct a user study, which is approved by the Institutional Review Board (IRB) of our institutes. 

{\textbf{Setup.} We have recruited 102 participants to answer two sets of questions. The participants are aged 18$\sim$28 with 72 males and 30 females.}  Before answering each question, users are asked to listen to one or two audios and give a score according to the quality of the audio(s). Users can listen to the audios for multiple times if they are unsure of the answer. The audios are played via a JBL GO2 loudspeaker in a quiet environment. {Each audio sample is clipped to 10s. Note that we anonymize an audio piece by piece so that the anonymization will not degrade over long pieces of audio.} The user study is carried out in a controlled laboratory environment, where we guarantee that no bots, scripts or automated answering tools are used in the process. In addition, we randomly insert attention-check questions to ensure that the participants pay attention to the questions.

\textbf{Timbre.} The first set of questions evaluates the timbre of the anonymized audios. We select 5 audios and anonymize each audio with \sys and four baselines, resulting in a total of 25 pairs of audios. Each participant is asked to listen to the original audio and the anonymized audio from the same speaker and to rate the similarity between the two with a score from 1$\sim$10 points (1 for completely different and 10 for completely the same). The user is also asked to assume that the speaker of the original audio is a celebrity or their acquaintances and rate whether they accept the anonymized audio as from the same speaker. There is a total of 25 similarity ratings and 25 acceptability ratings given by each participant.

\textbf{Intelligibility \&  Naturalness.} The second set of questions evaluates the intelligibility and the naturalness of the anonymized audios.  We select 5 audios and anonymize each audio with \sys and the four baselines, resulting in a total of 25 audios. Each participant is asked to listen to one audio sample and to rate the intelligibility of the audio with a score from 1$\sim$10 points (1 for completely unintelligible and 10 for completely intelligible). The user is also asked to rate the naturalness of the audio with a score from 1$\sim$10 points (1 for completely unnatural and 10 for completely natural). There is a total of 25 naturalness ratings and 25 intelligibility ratings given by each participant.

{\textbf{Results.} As shown in Figure~\ref{fig:userstudy}, \sys gains the highest scores in similarity and acceptability, with most scores distributed between 7$\sim$10, which means that \sys preserves the timbre of audios for users to trust that the audio is from the genuine speaker. For naturalness and intelligibility, \sys obtains scores as high as HFGAN. User studies verify that \sys preserves intelligibility, naturalness and timbre of the audios. }

\subsection{Efficiency}

We test \sys and the baselines \cite{qian2018hidebehind, qian2019speech, fang2019speaker, DBLP:journals/corr/abs-2202-13097, DBLP:conf/interspeech/0001TTNE21} under the same computing resources on the \emph{test-clean} set of LibriSpeech. The results are shown in Table~\ref{tab:compare}, which indicates that \sys has the highest efficiency.

\section{Related Work}
In this section, we briefly review the most related works on voice anonymization, including signal processing, voice conversion \& synthesis and voice adversarial examples.

\subsection{Voice Signal Processing}
% 1. McAdams \red{Jiangyi}
% 2. You talk too much
Conventional signal processing techniques are used to alter the voice signals. Patino et al.~\cite{DBLP:conf/interspeech/0001TTNE21} proposed to utilize the McAdams coefficient to shift the formant positions in an utterance for speaker anonymization. Vaidya et al.~\cite{DBLP:conf/sp/VaidyaS19} modified the pitch, tempo, pause, and MFCCs of an audio to alter the voice characteristics. Voice signal processing methods do not consider preserving naturalness of audios, thus they induce large distortions. In comparison, \sys limits the distortion with a psychoacoustics-based loss in the training phase to mask the introduced anonymization perturbations.

% and leverages \emph{Throttle} to constrain anonymization perturbations.
\subsection{Voice Conversion \& Synthesis}
Voice Conversion (VC)/ Synthesis (VS) methods aim to convert the speaker features of an original audio into those of the target speaker while preserving naturalness. 

\textbf{VTLN-based VC.} VTLN is a traditional voice conversion method that utilizes a frequency warping function to rescale the frequency axis of the voice spectrogram~\cite{VTLN}. Qian et al.~\cite{qian2018hidebehind, qian2019speech} utilized a bilinear warping function with randomly-chosen parameters to conceal the original voiceprint. However, they do not consider intelligibility and naturalness of the sanitized audios. According to our user study, the method has low scores of intelligibility and naturalness. Srivastava et al.~\cite{srivastava2020evaluating} investigated the performance of different target speaker selection strategies in two VTLN-based VC and one DNN-based VC. Their results show that VTLN-based VC methods suffer from a 6.5\%$\sim$10.4\% WER increase on ASR.  In contrast, we consider intelligibility and naturalness in training the anonymizer to achieve lower WER/CER and high subjective scores in user studies.

\textbf{DNN-based VC/VS.} Yoo et al.~\cite{DBLP:journals/access/YooLLOKY20} proposed to anonymize audios by VC technique based on CycleVAE-GAN, which modifies the speaker identity vectors of the VAE input.
Fang et al.~\cite{fang2019speaker} proposed to disentangle linguistic and speaker identity features from an utterance, replace the latter with a pseudo identity, and re-synthesize an anonymized audio. Han et al.~\cite{han2020voice} designed a voiceprint privacy metric according to differential privacy~\cite{dwork2006calibrating} and adapted Fang~\cite{fang2019speaker} to a voice data release mechanism that satisfies the privacy metric. Miao et al.~\cite{DBLP:journals/corr/abs-2202-13097} followed the basic framework of Fang~\cite{fang2019speaker} but proposed to use a HuBERT-based content encoder, an ECAPA-TDNN speaker encoder, and a HiFi GAN to re-synthesize the speech. Justin et al.~\cite{DBLP:conf/fg/JustinSDVIM15} proposed to transform only the linguistic content of an audio into an audio of another speaker with a speech synthesis system. However, the features of pitch, rhythm, tempo, and pause in the audio are all lost. DNN-based VC/VS methods convert the original speaker features into those of another speaker, which may not be suitable for instant messaging and social media applications. In comparison, we preserve the timbre of the original speaker while hiding the voiceprint of the speaker from the ASV. 
% Compared to their work, \sys preserves all the prosodic and spectro-temporal features in a speech, thus ensuring the subjective naturalness for human listeners.

\subsection{Voice Adversarial Examples}

As far as we know, there is only one work on generative model-based audio adversarial attacks, called FAPG. FAPG~\cite{xie2020enabling} trains an audio adversarial example generator and a series of feature maps. Each feature map is trained for each target speaker. A feature map can be concatenated with the anonymizer to produce adversarial examples of a specific target speaker. \sys differs from FAPG in four aspects. Firstly, FAPG attacks a speaker classification model with fixed and known classes and requires a re-training of feature maps for unseen classes/speakers. Secondly, FAPG utilizes the last softmax layer of the classification model, which is training-set-specific, resulting in low transferability~\cite{DBLP:conf/iccv/HuangKGHBL19}. In comparison, \sys can realize targeted anonymization with the input of any target speaker without the need to retrain the anonymizer. Thirdly, the size of the FAPG model increases with the number of feature maps, while the \sys model has a constant size. Finally, the design of FAPG only considers one fixed constraint on the added noise. In contrast, \sys introduces a learnable structure to allow diversified constraints on the anonymization perturbation.

\section{{Ethics Discussion}}

{\sys is designed to protect voiceprint, a sensitive biometric, as the General Data Protection Regulation (GDPR) \cite{GDPR} enacted by the European Union grants natural persons ``the right to the personal data protection". The processing of biometric data for the purpose of uniquely identifying a natural person is prohibited except for certain cases \cite{GDPR}, the prominent ones of which include}
\begin{itemize}\setlength{\itemsep}{0pt}
\item {Criminal convictions and offences.}
\item {Social security.}
\item {Scientific or historical research.}
\item {Public health.}
\item {Consent by the data owner.}  
\end{itemize}

{We shall take proper measures to prevent the abuse of \sys in these legitimate cases, including but not limited to 1) providing necessary details of \sys to bodies that can legally conduct voice analysis, 2) withholding the release of the code of \sys for 90 days after notification. }

\section{Conclusion \& Future Work}

In this work, we present the design, implementation and evaluation of \sys, a real-time voice anonymization system, which preserves intelligibility, naturalness and timbre of the audios. Extensive experiments on four language datasets with various ASVs and ASRs confirm the effectiveness and transferability of \sys. The user study demonstrates the high perceptual quality of the anonymized audios generated by \sys. In the future, \sys can be further improved in several aspects.

\emph{Psychoacoustics}. Apart from spectral masking, there are other psychoacoustic effects that may be leveraged to further improve the performance of \sys. For instance, we may find out the non-silent segments of audios utilizing Voice Activity Detection (VAD)~\cite{DBLP:journals/spl/SohnKS99}, and only anonymize the non-silent parts, which may further improve naturalness.

\emph{Analog voice data}. We mainly consider anonymization for digital voice data in applications such as voice messaging and social media. Anonymization for analog voice data or over-the-air digital voice data may be necessary in the case where the adversary physically records public speeches or private conversations. To realize over-the-air or analog voice data anonymization, a possible way is to incorporate room impulse responses (RIRs)~\cite{DBLP:conf/icdsp/JeubSV09} in the optimization problem of voice anonymization. This is our future direction.

{\emph{Other attacks}. We assume that the adversary knows that voice anonymization is performed and only focuses on de-anonymization. However, it is possible that the adversary does not know whether voice anonymization is conducted and tries to detect its presence. Anonymization detection can be performed by A1 $\sim$A3 before de-anonymization. The attacker may also train the ASV with (denoised) anonymized samples and feed denoised samples into ASV during inference. However, a possible loophole is that the performance of ASV on clean samples may degrade. We consider these attacks as our future direction.}

{\emph{Extension to other applications}. The approach we design for \sys may be extended to other tasks. For example, we may train an adversarial model with a similar architecture against an audio Deepfake model, preventing the generation of Deepfake audio, similar to Fawkes for facial images \cite{shan2020fawkes}. We consider the extensions of \sys as our future work.
}

\section*{Acknowledgments}
We sincerely thank our Shepherd and all the anonymous reviewers for their valuable comments. This work is supported by China NSFC Grant 61925109.

%-------------------------------------------------------------------------------
\clearpage
\bibliographystyle{plain}
\bibliography{sections/mybib}
\clearpage

\appendix

\section{Theoretical Analysis}\label{sec:proof}
In this part, we prove the unidentifiability achieved by the untargeted and targeted anonymization, and prove the unlinkability achieved by the targeted anonymization.

% We make the following basic assumptions for the ASV, i.e., 
{We first define two probabilities $p_A$ and $p_R$ related to the distinguishability of ASV, }
% \begin{subequations}
% \begin{align}
% &\forall x_1, x_2 \in X, \exists t_A, ~\mathbb{P}\Big(S(\mathcal{V}(x_1), \mathcal{V}(x_2)) \geq t_A\Big)=p_A,\label{equ:same}\\
% &\forall x \in X, y \in Y, \exists t_R, ~\mathbb{P}\Big(S(\mathcal{V}(x), \mathcal{V}(y)) <t_R\Big)=p_R,\label{equ:diff}
% \end{align}
% \end{subequations}
\begin{subequations}
\begin{align}
&\forall X , \exists t_A, ~\mathbb{P}\Big(\forall x_1, x_2 \in X,~\mathcal{S}(\mathcal{V}(x_1), \mathcal{V}(x_2)) \geq t_A\Big)=p_A,\label{equ:same}\\
&\forall X, Y, \exists t_R, ~\mathbb{P}\Big(\forall x \in X, y \in Y,~\mathcal{S}(\mathcal{V}(x), \mathcal{V}(y)) <t_R\Big)=p_R,\label{equ:diff}
\end{align}
\end{subequations}
\noindent where $X, Y$ are the data distributions of two different speakers. $t_A, t_R\in(-1,1)$ are two thresholds, and we have $t_A\geq t_R$. $\mathcal{S}(\cdot, \cdot)$ is the cosine similarity function. Without loss of generality, we further assume that $\forall x, \|\mathcal{V}(x)\|=1$. (\ref{equ:same}) states that if two utterances are from the same speaker, the ASV outputs the score no less than $t_A$ with a probability of $p_A$. (\ref{equ:diff}) states that if two utterances are from two different speakers, the ASV outputs the score less than $t_R$ with a probability of $p_R$. 
The ASV guarantees that both $p_A$ and $p_R$ are close to 1.

\subsection{Untargeted Anonymization}
According to the basic formulation (\ref{equ:obj1}), we train an untargeted anonymizer $\mathcal{G}$ that satisfies
\begin{equation}\label{equ:untarget}
    \mathbb{P}\Big(\forall x,~\mathcal{S}(\mathcal{V}(\tilde{x}), \mathcal{V}(x)) < k_{u}\Big)=p_G,
\end{equation}
\noindent where $\tilde{x}=\mathcal{G}(x), k_u\in(-1, 1)$. (\ref{equ:untarget}) states that for any audio $x$, $\mathcal{G}$ outputs $\tilde{x}$ such that the score between $x$ and $\tilde{x}$ is less than $k_u$ with a probability of $p_G$ close to 1.

\begin{theorem}
$\forall x_0, x_1 \in X$, $\forall t_R\in(\sqrt{5}-2, 1)$, $\exists k_u\in (-1, 1)$, such that $\mathcal{S}(\mathcal{V}(\tilde{x}_0), \mathcal{V}(x_1))<t_R$ with a probability higher than $p_A\cdot p_G$.
\end{theorem}

\begin{proof}
\begin{equation*}
\begin{aligned}
&~\mathcal{S}(\mathcal{V}(\tilde{x}_0), \mathcal{V}(x_1))=\mathcal{V}(\tilde{x}_0)\cdot \mathcal{V}(x_1)\\
% &=\mathcal{V}(\tilde{x}_0)\cdot \big[\mathcal{V}(x_1)-\mathcal{V}(x_0)+\mathcal{V}(x_0)\big]\\
&=\mathcal{V}(\tilde{x}_0)\cdot \big[\mathcal{V}(x_1)-\mathcal{V}(x_0)\big]+\mathcal{V}(\tilde{x}_0)\cdot \mathcal{V}(x_0)\\
&\leq\|\mathcal{V}(\tilde{x}_0)\|\cdot\|\mathcal{V}(x_1)-\mathcal{V}(x_0)\|+\mathcal{V}(\tilde{x}_0)\cdot \mathcal{V}(x_0)\\
&=\|\mathcal{V}(\tilde{x}_0)\|\cdot\sqrt{\|\mathcal{V}(x_1)\|^2+\|\mathcal{V}(x_0)\|^2-2\mathcal{V}(x_1)\cdot \mathcal{V}(x_0)}\\
&~~~~~~~~~~~~+\mathcal{V}(\tilde{x}_0)\cdot \mathcal{V}(x_0)\\
&=\sqrt{2-2\mathcal{V}(x_1)\cdot \mathcal{V}(x_0)}+\mathcal{V}(\tilde{x}_0)\cdot \mathcal{V}(x_0)\\
\end{aligned}
\end{equation*}
From (\ref{equ:same}),
\begin{equation}
    \mathbb{P}\Big(\sqrt{2-2\mathcal{V}(x_1)\cdot \mathcal{V}(x_0)}\leq \sqrt{2-2t_A}\Big)=p_A.
\end{equation}
From (\ref{equ:untarget}),
\begin{equation}
    \mathbb{P}\Big(\mathcal{V}(\tilde{x}_0)\cdot \mathcal{V}(x_0)< k_u\Big)=p_G.
\end{equation}
Thus, 
\begin{equation*}
\begin{aligned}
&~\mathbb{P}\Big(\mathcal{S}(\mathcal{V}(\tilde{x}_0), \mathcal{V}(x_1))\leq\sqrt{2-2\mathcal{V}(x_1)\cdot \mathcal{V}(x_0)}\\&~~~~~~~~~~~~+\mathcal{V}(\tilde{x}_0)\cdot \mathcal{V}(x_0)< \sqrt{2-2t_A}+k_u\Big)
\geq p_A\cdot p_G.
\end{aligned}
\end{equation*}
$\forall t_R\in (\sqrt{5}-2,1)$, $t_A\geq t_R$, 
\begin{equation*}
\begin{aligned}
&~\mathcal{S}(\mathcal{V}(\tilde{x}_0), \mathcal{V}(x_1))< \sqrt{2-2t_A}+k_u \leq \sqrt{2-2t_R}+k_u,\\
&\sqrt{2-2t_R}+k_u<t_R\Leftrightarrow k_u<\big(t_R-\sqrt{2-2t_R}\big)
\in(-1, 1).
\end{aligned}
\end{equation*}
Therefore, $\exists k_u\in (-1, 1)$, such that the probability of $\mathcal{S}(\mathcal{V}(\tilde{x}_0), \mathcal{V}(x_1))<t_R$ is higher than $p_A\cdot p_G$, which is close to 1.
\end{proof}

\subsection{Targeted Anonymization}
\subsubsection{Unidentifiability}
According to the basic formulation (\ref{equ:obj1}), we train a targeted anonymizer $\mathcal{G}$ that satisfies
\begin{equation}\label{equ:target}
    \mathbb{P}\Big(\forall x,~\mathcal{S}(\mathcal{V}(\tilde{x}), v) > k_t\Big)=p_G,
\end{equation}
\noindent where $\tilde{x}=\mathcal{G}(x, v), k_t\in(-1, 1)$. (\ref{equ:target}) states that for any audio $x$, the $\mathcal{G}$ outputs $\tilde{x}$ such that the similarity score between $v$ and $\tilde{x}$ is larger than $k_t$ with a probability of $p_G$.
We assume that the target speaker has a voiceprint $v$ that is far away from that of the original speaker, i.e.,
\begin{equation}\label{equ:voiceprint}
    \mathbb{P}\Big(\forall x\in X,~\mathcal{S}(\mathcal{V}(x), v) < t_R-\eta\Big)=p_R, 0\leq\eta<t_R+1,
\end{equation}
\noindent where $(t-\eta)$ stands for how far the target voiceprint $v$ is from the user's, $\mathcal{V}(x)$, and $\eta$ is a margin.

\begin{theorem}
$\forall x_0, x_1 \in X$, $\forall t_R\in(-\frac{1}{2}, 1)$, $\exists k_t\in(-1, 1), \eta\in(0, t_R+1)$, such that $\mathcal{S}(\mathcal{V}(\tilde{x}_0), \mathcal{V}(x_1))<t_R$ with a probability higher than $p_R \cdot p_G$.
\end{theorem}

\begin{proof}
\begin{equation}\label{equ:target_uniden}
\begin{aligned}
&~\mathcal{S}(\mathcal{V}(\tilde{x}_0), \mathcal{V}(x_1))=\mathcal{V}(\tilde{x}_0)\cdot \mathcal{V}(x_1)\\
% &=\mathcal{V}(\tilde{x}_0)\cdot \big[\mathcal{V}(x_1)+v-v\big]\\
&=\mathcal{V}(\tilde{x}_0)\cdot \big[\mathcal{V}(x_1)+v\big]-\mathcal{V}(\tilde{x}_0)\cdot v\\
&\leq\|\mathcal{V}(\tilde{x}_0)\|\cdot\|\mathcal{V}(x_1)+v\|-\mathcal{V}(\tilde{x}_0)\cdot v\\
&=\|\mathcal{V}(\tilde{x}_0)\|\cdot\sqrt{\|\mathcal{V}(x_1)\|^2+\|v\|^2+2\mathcal{V}(x_1)\cdot v}\\
&~~~~~~~~~~~~-\mathcal{V}(\tilde{x}_0)\cdot v\\
&=\sqrt{2+2\mathcal{V}(x_1)\cdot v}-\mathcal{V}(\tilde{x}_0)\cdot v\\
\end{aligned}
\end{equation}
From (\ref{equ:diff}),
\begin{equation*}
    \mathbb{P}\Big(\sqrt{2+2\mathcal{V}(x_1)\cdot v}< \sqrt{2+2(t_R-\eta)}\Big)=p_R.
\end{equation*}
From (\ref{equ:target}),
\begin{equation*}
    \mathbb{P}\Big(\mathcal{V}(\tilde{x}_0)\cdot v> k_t\Big)=p_G.
\end{equation*}
Thus,
\begin{equation*}
\begin{aligned}
&~\mathbb{P}\Big(\mathcal{S}(\mathcal{V}(\tilde{x}_0), \mathcal{V}(x_1))\leq\sqrt{2+2\mathcal{V}(x_1)\cdot v}\\
&~~~~~~~~~~~~-\mathcal{V}(\tilde{x}_0)\cdot v< \sqrt{2+2(t_R-\eta)}-k_t\Big)\geq p_R\cdot p_G.
\end{aligned}
\end{equation*}
Let $\eta=\frac{1}{2}$, $\forall t_R\in(-\frac{1}{2}, 1)$,
\begin{equation*}
\begin{aligned}
&~\sqrt{2+2(t_R-1/2)}-k_t<t_R\\
&\Leftrightarrow k_t>\big(\sqrt{1+2t_R}-t_R)
\in(\frac{1}{2}, 1).
\end{aligned}
\end{equation*}
Therefore, $\exists k_t\in(-1, 1)$, such that the probability of $\mathcal{S}(\mathcal{V}(\tilde{x}_0), \mathcal{V}(x_1))<t_R$ is higher than $p_R\cdot p_G$, which is close to 1.
\end{proof}

\subsubsection{Unlinkability}
\textbf{A3 does not have the key.} We first prove that if A3 does not have the voiceprint key $v$ and converts the clean samples of potential speakers with a random key $u$, then the anonymized audio $\tilde{x}_0$ will not be matched with any enrollment samples with high probability.

\begin{theorem}
$\forall x_0 \in X$, $\forall y$, $\exists k_t\in(-1, 1)$, such that $\mathcal{S}(\mathcal{V}(\tilde{x}_0), \mathcal{V}(\tilde{y}))<t_R$ with a high probability, where $\tilde{x}_0=\mathcal{G}(x_0, v)$, $\tilde{y}=\mathcal{G}(y, u)$ and $u\sim U(\{u\in \mathbb{R}^{1\times N}:\|u\|=1\})$.
\end{theorem}

\begin{proof}
\begin{equation}\label{equ:target_unlink}
\begin{aligned}
&~\mathcal{S}(\mathcal{V}(\tilde{x}_0), \mathcal{V}(y))\\
&=\big[\mathcal{V}(\tilde{x}_0)-v\big]\cdot\big[\mathcal{V}(y)-u\big]+u\cdot \mathcal{V}(\tilde{x}_0)+v\cdot \mathcal{V}(y) - u\cdot v\\
&=\big[\mathcal{V}(\tilde{x}_0)-v\big]\cdot\big[\mathcal{V}(y)-u\big]+u\cdot \big[\mathcal{V}(\tilde{x}_0)-v\big]\\
&~~~~~~~~~~~~+v\cdot \big[\mathcal{V}(y)-u\big] + u\cdot v\\
&\leq\|\mathcal{V}(\tilde{x}_0)-v\|\cdot\|\mathcal{V}(y)-u\|+\|u\|\cdot\|\mathcal{V}(\tilde{x}_0)-v\|\\
&~~~~~~~~~~~~+\|v\|\cdot\|\mathcal{V}(y)-u\|+ u\cdot v,
\end{aligned}
\end{equation}
\noindent of which, 
\begin{equation*}
\begin{aligned}
&~\mathbb{P}\Big(\|\mathcal{V}(\tilde{x}_0)-v\|=\sqrt{\|\mathcal{V}(\tilde{x}_0)\|^2+\|v\|^2-2\mathcal{V}(\tilde{x}_0)\cdot v}\\
&\leq\sqrt{2-2\mathcal{V}(\tilde{x}_0)\cdot v}<\sqrt{2-2k_t}\Big)=p_G.
\end{aligned}
\end{equation*}
\noindent Similarly, 
\begin{equation*}
\begin{aligned}
\mathbb{P}\Big(\|\mathcal{V}(y)-u\|< \sqrt{2-2k_t}\Big)=p_G.
\end{aligned}
\end{equation*}

\begin{lemma}
\cite{li2011concise}~The surface area of an $N$-sphere in $N$-dimensional Euclidean space, of radius $R$, can be given in the closed form:
\begin{equation}
\begin{aligned}
A_{N}(R)=\frac{2\pi^{\frac{N}{2}}}{\Gamma(\frac{N}{2})}R^{N-1},
\end{aligned}
\end{equation}
\noindent where $\Gamma$ is the gamma function. An $N$-sphere can be cut into two caps, by a hyperplane. We denote the colatitude angle, i.e., the angle between a vector of the sphere and its positive $N^{th}$-axis, as $\varphi$. We only consider the smaller cap in the following, i.e., $0\leq\varphi\leq\frac{\pi}{2}$.
\end{lemma}
 
\begin{lemma}
\cite{li2011concise}~The surface area of a hyperspherical cap described above can be given in the closed form:
\begin{equation}
\begin{aligned}
A_{N}^{cap}(R, \varphi)=\frac{1}{2}A_{N}(R)I_{sin^2_{\varphi}}\Big(\frac{N-1}{2}, \frac{1}{2}\Big),
\end{aligned}
\end{equation}
\noindent where $I_{sin^2_{\varphi}}(\cdot, \cdot)$ is the regularized incomplete beta function.
\end{lemma}
Thus,
\begin{equation}
\begin{aligned}
&~\mathbb{P}(u\cdot v\geq \mathrm{cos}\varphi) = \frac{A_{N}^{cap}(R, \varphi)}{A_{N}(R)}=\frac{1}{2}I_{sin^2_{\varphi}}\Big(\frac{N-1}{2}, \frac{1}{2}\Big)
\end{aligned}
\end{equation}
\noindent Let $\varphi=\mathrm{arccos}(\alpha \cdot t_R), 0<\alpha<1$,
\begin{equation}
\begin{aligned}
&~\mathbb{P}(u\cdot v< \alpha t_R) = 1-\frac{1}{2}I_{sin^2_{\varphi}}\Big(\frac{N-1}{2}, \frac{1}{2}\Big)\\
&=\frac{1}{2}+\frac{1}{2}I_{1-sin^2_{\varphi}}\Big(\frac{1}{2}, \frac{N-1}{2}\Big)=\frac{1}{2}+\frac{1}{2}I_{cos^2_{\varphi}}\Big(\frac{1}{2}, \frac{N-1}{2}\Big)\\
&=\frac{1}{2}+\frac{1}{2}I_{\alpha^2 t_R^2}\Big(\frac{1}{2}, \frac{N-1}{2}\Big)
\end{aligned}
\end{equation}
% Note that, when $N$ is large, $I_{\alpha^2 t_R^2/9}\Big(\frac{1}{2}, \frac{N-1}{2}\Big)$ is close to 1. 
Thus,
\begin{equation}
\begin{aligned}
&~\mathbb{P}\Big(\mathcal{S}(\mathcal{V}(\tilde{x}_0), \mathcal{V}(y))<2-2k_t+2\sqrt{2-2k_t}+\alpha\cdot t_R\Big)\\
&\geq p_G^2\cdot \Bigg(\frac{1}{2}+\frac{1}{2}I_{\alpha^2 t_R^2}\Big(\frac{1}{2}, \frac{N-1}{2}\Big)\Bigg)\\
&\rightarrow 1 (p_G\rightarrow1, N\rightarrow\infty)\\
&\lim\limits_{k_t\to1} \big(2-2k_t+2\sqrt{2-2k_t}+\alpha\cdot t_R\big) = \alpha\cdot t_R< t_R
\end{aligned}
\end{equation}

Therefore, $\forall t_R\in(-1, 1)$, $\exists k_t\in (-1, 1)$ such that $\mathcal{S}(\mathcal{V}(\tilde{x}_0), \mathcal{V}(y))<t_R$ with a probability close to 1.
\end{proof}

\textbf{A3 has the key.} Next, we prove that if A3 has the voiceprint key $v$, and converts the clean samples of potential speakers with the exact $v$, then the anonymized audio $\tilde{x}_0$ will be matched with any enrollment samples with high probability.

\begin{theorem}
$\forall x_0 \in X$, $t_A \in (-1, 1)$, $\exists k_t\in(-1, 1)$, $\forall y$, $\mathcal{S}(\mathcal{V}(\tilde{x}_0), \mathcal{V}(\tilde{y}))>t_A$ with a high probability.
\end{theorem}

\begin{proof}
From (\ref{equ:target}),
\begin{equation*}
\begin{aligned}
    &\mathbb{P}\Big(\mathcal{V}(\tilde{x}_0)\cdot v > k_t\Big)=p_G,\\
    &\mathbb{P}\Big(\mathcal{V}(\tilde{y})\cdot v> k_t\Big)=p_G.
\end{aligned}
\end{equation*}
Thus, $\forall y$,
\begin{equation*}
\begin{aligned}
&~\mathbb{P}\Big(\mathcal{S}(\mathcal{V}(\tilde{x}_0), \mathcal{V}(\tilde{y}))=\big[\mathcal{V}(\tilde{x}_0)\cdot v\big]\big[\mathcal{V}(\tilde{x}_1)\cdot v\big]> k_t^2\Big)\geq p_G^2.
\end{aligned}
\end{equation*}
Therefore, $\forall t_A\in [-1, 1]$, let $k_t=\sqrt{t_A+\alpha(1-t_A)}, 0<\alpha<1$, such that, 
\begin{equation*}
\begin{aligned}
&~\mathcal{S}(\mathcal{V}(\tilde{x}_0), \mathcal{V}(\tilde{y}))> k_t^2 = t_A+\alpha(1-t_A) > t_A.
\end{aligned}
\end{equation*}

\end{proof}

\onecolumn
\clearpage
\section{Baseline}\label{appdix:baseline}
\begin{itemize}
    \item {NSF}. NSF~\cite{fang2019speaker} is a voice synthesis-based anonymization method that separates the speech content and the voiceprint, and then replaces the voiceprint with a pseudo one to re-synthesize an anonymized audio.
    
    \item  {HFGAN}. HFGAN~\cite{DBLP:journals/corr/abs-2202-13097} is a voice synthesis-based anonymization method that uses a HuBERT-based content encoder, an ECAPA-TDNN speaker encoder, and a HiFi GAN to re-synthesize the anonymized audio. 
    
    \item {McAdams}. McAdams~\cite{DBLP:conf/interspeech/0001TTNE21} is a signal processing-based anonymization method that utilizes the McAdams coefficient to shift the formant positions of the audio.
    
    \item {VoiceMask}. VoiceMask~\cite{qian2018hidebehind, qian2019speech} is a voice conversion-based anonymization method that utilizes a bilinear warping function on the frequency domain to convert the original audio into another audio.
\end{itemize}
NSF and HFGAN are implemented following the VoicePrivacy 2022 Challenge~\cite{vpc2022}. A pool of 500 voiceprints from the \emph{LibriTTS-train-other-500} dataset are prepared for the synthesis. Among them, the 200 voiceprints that are farthest from the original speaker are chosen with the probabilistic linear discriminant analysis (PLDA). Finally, 100 voiceprints are randomly sampled from these 200 voiceprints and averaged as the pseudo voiceprint. McAdams is also implemented following the VoicePrivacy 2022 Challenge with a fixed McAdams coefficient of $0.8$. 
VoiceMask is implemented using the Github repository\footnote{\url{https://github.com/yuunin/time-invariant-anonymization}}. The warping factor in the bilinear function determines the strength of anonymization. We set the warping factor as $0.1$ following existing works \cite{qian2018hidebehind, qian2019speech}, which yields the best anonymization effect. 

\begin{table}[h]\centering
\begin{threeparttable}
\setlength{\abovecaptionskip}{0pt}% 
\setlength{\belowcaptionskip}{0pt}%

\caption{ASRs used for evaluation.}

\footnotesize
\setlength{\tabcolsep}{5mm}{
\begin{tabular}{@{}lllclr@{}}
\toprule
\textbf{Model}       &\textbf{Alias}			& \textbf{LM$^*$}   & \multicolumn{1}{l}{\textbf{Language}}                                            & \textbf{Dataset}									& \textbf{WER(\%)}		\\ \midrule
DeepSpeech2          & \textbf{DS} 				& no            & \multirow{5}{*}{English}                                                         & LibriSpeech										& 10.46					\\
wav2vec2             & \textbf{WV}				& no            &                                                                                  & CV-en-train$^\#$										& 9.66					\\
crdnn                & \textbf{CD1}				& transformer   &                                                                                  	   & LibriSpeech										& 3.61					\\
crdnn                & \textbf{CD2} 			& rnn           &                                                                                  & LibriSpeech										& 4.23					\\
transformer          & \textbf{TF}				& transformer  &                                                                                  	   & LibriSpeech										& 3.01					\\ \midrule
\begin{tabular}[c]{@{}l@{}}wav2vec2-transformer\end{tabular} & \textbf{WV} 				& no            & \multirow{2}{*}{\begin{tabular}[c]{@{}c@{}}Mandarin\\Chinese\end{tabular}} & AISHELL      & 5.58$^\ddagger$      			\\
transformer          & \textbf{TF} 				& no            &                                                                                  & AISHELL             								& 6.77$^\ddagger$					\\ \midrule
wav2vec2             & \textbf{WV} 				& no            & \multirow{2}{*}{French}                                                          & CV-fr-train$^\#$ 									& 13.18					\\
crdnn                & \textbf{CD} 				& no            &                                                                                  & CV-fr-train  									& 19.47					\\ \midrule
wav2vec2             & \textbf{WV} 				& no            & \multirow{2}{*}{Italian}                                                         & CV-it-train$^\#$  									& 12.74					\\
crdnn                & \textbf{CD} 				& no            &                                                                                  & CV-it-train									& 17.48					\\ \bottomrule
\end{tabular}}

\begin{tablenotes}[flushleft]
\item[] \vspace{-1pt}\hspace{-2pt}\footnotesize (i) $^*$: LM, language model. (ii) $^\#$: CV, CommonVoice. CV-en-train, CV-fr-train and CV-it-train are the English, French and Italian training sets of CV respectively. (iii)$^\ddagger$: The performance of Chinese models are evaluated using CER.
\end{tablenotes}

\label{tab:ASRs}
\end{threeparttable}
\end{table}
\begin{table*}\centering
\begin{threeparttable}[tt]
\setlength{\abovecaptionskip}{5pt}% 
\setlength{\belowcaptionskip}{0pt}%

\caption{Comparison with existing works.}

\scriptsize
\setlength{\tabcolsep}{2mm}{
\begin{tabular}{cl|c|crr|crr|crr|crr|crr}
\hline
\multicolumn{2}{c|}{\multirow{2}{*}{\textbf{Model}}} & \textbf{B0 (\%)}  & \multicolumn{3}{c|}{\textbf{NSF (\%)}}                                                    & \multicolumn{3}{c|}{\textbf{HFGAN (\%)}}                                                    & \multicolumn{3}{c|}{\textbf{McAdams (\%)}}                                                    & \multicolumn{3}{c|}{\textbf{VoiceMask (\%)}}                                                     & \multicolumn{3}{c}{\textbf{\sys (\%)}}                                                      \\
\multicolumn{2}{c|}{}                                & \textbf{EER} & \textbf{MMR}               & \multicolumn{1}{c}{\textbf{WMR}} & \multicolumn{1}{c|}{\textbf{EER}} & \textbf{MMR}              & \multicolumn{1}{c}{\textbf{WMR}} & \multicolumn{1}{c|}{\textbf{EER}} & \textbf{MMR}              & \multicolumn{1}{c}{\textbf{WMR}} & \multicolumn{1}{c|}{\textbf{EER}} & \textbf{MMR}               & \multicolumn{1}{c}{\textbf{WMR}} & \multicolumn{1}{c|}{\textbf{EER}} & \textbf{MMR}               & \multicolumn{1}{c}{\textbf{WMR}} & \multicolumn{1}{c}{\textbf{EER}} \\ \hline
\multirow{6}{*}{\textbf{ASV}}     & \textbf{EP}     & 0.70          & 92.06  & 1.07                             & 27.31                             & 92.44 & 1.07                             & 29.52                             & 73.32 & 1.07                             & 25.67                             & 84.66  & 1.07                             & 18.80                             & 100.0  & 1.07                             & 64.58                            \\
                                   & \textbf{XV}     & 6.53         & 74.31  & 6.91                             & 30.21                             & 77.90 & 6.87                             & 31.97                             & 78.32 & 6.87                             & 39.89                             & 93.47  & 6.87                             & 35.04                             & 98.40  & 6.87                             & 39.77                            \\
                                   & \textbf{GMM}    & 11.39        & 67.60  & 11.38                            & 38.18                             & 67.44 & 11.38                            & 38.68                             & 65.12 & 11.38                            & 33.87                             & 87.48  & 11.38                            & 41.02                             & 87.89  & 11.38                            & 41.17                            \\
                                   & \textbf{IV}     & 6.03         & 81.45  & 6.03                             & 32.06                             & 83.43 & 6.03                             & 34.00                             & 72.19 & 6.03                             & 30.10                             & 81.11  & 6.03                             & 37.05                             & 86.81  & 6.03                             & 37.75                            \\
                                   & \textbf{IF}     & 9.44         & 78.28  & 6.79                             & 33.59                             & 93.09 & 6.79                             & 43.78                             & 96.37 & 6.79                             & 76.49                             & 100.0  & 6.79                             & 86.49                             & 97.02  & 6.79                             & 47.22                            \\\hline
                                   & \textbf{AVG}    & 6.82         & 78.74  & 6.44                             & 32.27                             & 82.86 & 6.43                             & 35.59                             & 77.06 & 6.43                             & 41.20                             & 89.34  & 6.43                             & 43.68                             & 94.02  & 6.43                             & 46.10                            \\ 
                                   & \textbf{WCS}    & -            & 67.60  & 1.07                             & 27.31                             & 67.44 & 1.07                             & 29.52                             & 65.12 & 1.07                             & 25.67                             & 81.11  & 1.07                             & 18.80                             & 86.81  & 1.07                             & 37.75                            \\ \hline\hline

\multicolumn{2}{c|}{\multirow{2}{*}{\textbf{Model}}} & \textbf{B0 (\%)}  & \multicolumn{3}{c|}{\textbf{NSF (\%)}}                                                    & \multicolumn{3}{c|}{\textbf{HFGAN (\%)}}                                                    & \multicolumn{3}{c|}{\textbf{McAdams (\%)}}                                                    & \multicolumn{3}{c|}{\textbf{VoiceMask (\%)}}                                                     & \multicolumn{3}{c}{\textbf{\sys (\%)}}                                                      \\
\multicolumn{2}{c|}{}                                & \textbf{WER}                      & \multicolumn{3}{c|}{\textbf{WER}}                                                                          & \multicolumn{3}{c|}{\textbf{WER}}                                                                          & \multicolumn{3}{c|}{\textbf{WER}}                                                                          & \multicolumn{3}{c|}{\textbf{WER}}                                                                          & \multicolumn{3}{c}{\textbf{WER}}                                                                         \\ \hline
\multirow{6}{*}{\textbf{ASR}}     & \textbf{DS}     & 10.46        & \multicolumn{3}{c|}{12.08}                                                                       & \multicolumn{3}{c|}{16.72}                                                                       & \multicolumn{3}{c|}{52.98}                                                                       & \multicolumn{3}{c|}{46.10}                                                                        & \multicolumn{3}{c}{13.92}                                                                        \\
                                   & \textbf{WV}     & 9.66         & \multicolumn{3}{c|}{12.35}                                                                       & \multicolumn{3}{c|}{11.24}                                                                       & \multicolumn{3}{c|}{21.26}                                                                       & \multicolumn{3}{c|}{17.17}                                                                        & \multicolumn{3}{c}{11.74}                                                                        \\
                                   & \textbf{CD1}    & 3.61         & \multicolumn{3}{c|}{3.74}                                                                        & \multicolumn{3}{c|}{5.00}                                                                        & \multicolumn{3}{c|}{26.03}                                                                       & \multicolumn{3}{c|}{13.99}                                                                        & \multicolumn{3}{c}{4.59}                                                                         \\
                                   & \textbf{CD2}    & 4.23         & \multicolumn{3}{c|}{4.60}                                                                        & \multicolumn{3}{c|}{5.58}                                                                        & \multicolumn{3}{c|}{26.57}                                                                       & \multicolumn{3}{c|}{14.86}                                                                        & \multicolumn{3}{c}{4.95}                                                                         \\
                                   & \textbf{TF}     & 3.01         & \multicolumn{3}{c|}{3.16}                                                                        & \multicolumn{3}{c|}{3.69}                                                                        & \multicolumn{3}{c|}{5.86}                                                                        & \multicolumn{3}{c|}{8.15}                                                                         & \multicolumn{3}{c}{3.06}                                                                         \\
                                   & \textbf{AVG}    & 6.19         & \multicolumn{3}{c|}{7.19}                                                                        & \multicolumn{3}{c|}{8.45}                                                                        & \multicolumn{3}{c|}{26.54}                                                                       & \multicolumn{3}{c|}{20.05}                                                                        & \multicolumn{3}{c}{7.65}                                                                         \\ \hline
\end{tabular}}

\begin{tablenotes}[flushleft]
\item[] \vspace{-1pt}\hspace{-2pt}\footnotesize \textbf{AVG}: the average-case scenario, \textbf{WCS}: the worst-case scenario. \textbf{EP}: ECAPA-TDNN, \textbf{XV}: X-vector, \textbf{GMM}: GMM-UBM, \textbf{IV}: ivector-PLDA, \textbf{IF}: iFlytek.\\ \textbf{DS}: DeepSpeech2, \textbf{WV}: wav2vec2, \textbf{CD1}: crdnn-transformer, \textbf{CD2}: crdnn-rnn, \textbf{TF}: transformer.
\end{tablenotes}

\label{tab:existingwork}
\end{threeparttable}
\end{table*}

\begin{table*}\centering
\begin{threeparttable}[tt]
\setlength{\abovecaptionskip}{5pt}% 
\setlength{\belowcaptionskip}{0pt}%

\caption{Cross-language performance with Chinese language.}

\scriptsize
\setlength{\tabcolsep}{2mm}{
\begin{tabular}{cl|c|rrr|rrr|rrr|rrr|rrr}
\hline
\multicolumn{2}{c|}{\multirow{2}{*}{\textbf{Model}}} & \textbf{B0 (\%)}       & \multicolumn{3}{c|}{\textbf{$\epsilon=0.02$ (\%)}}                                                      & \multicolumn{3}{c|}{\textbf{$\epsilon=0.04$ (\%)}}                                                      & \multicolumn{3}{c|}{\textbf{$\epsilon=0.06$ (\%)}}                                                      & \multicolumn{3}{c|}{\textbf{$\epsilon=0.08$ (\%)}}                                                      & \multicolumn{3}{c}{\textbf{$\epsilon=0.1$ (\%)}}                                                       \\
\multicolumn{2}{c|}{}                                & \textbf{EER}           & \multicolumn{1}{c}{\textbf{MMR}} & \multicolumn{1}{c}{\textbf{WMR}} & \multicolumn{1}{c|}{\textbf{EER}} & \multicolumn{1}{c}{\textbf{MMR}} & \multicolumn{1}{c}{\textbf{WMR}} & \multicolumn{1}{c|}{\textbf{EER}} & \multicolumn{1}{c}{\textbf{MMR}} & \multicolumn{1}{c}{\textbf{WMR}} & \multicolumn{1}{c|}{\textbf{EER}} & \multicolumn{1}{c}{\textbf{MMR}} & \multicolumn{1}{c}{\textbf{WMR}} & \multicolumn{1}{c|}{\textbf{EER}} & \multicolumn{1}{c}{\textbf{MMR}} & \multicolumn{1}{c}{\textbf{WMR}} & \multicolumn{1}{c}{\textbf{EER}} \\ \hline
\multirow{6}{*}{\textbf{ASV}}     & \textbf{EP}      & 1.07                   & 100                              & 1.30                             & 69.06                             & 100                              & 1.30                             & 70.33                             & 100                              & 1.30                             & 70.78                             & 100.0                            & 1.30                             & 71.06                             & 100                              & 1.30                             & 71.36                            \\
                                   & \textbf{XV}     & 2.45                   & 99.99                            & 2.33                             & 37.83                             & 100                              & 2.33                             & 38.91                             & 100                              & 2.33                             & 40.34                             & 100.0                            & 2.33                             & 41.54                             & 100                              & 2.33                             & 42.27                            \\
                                   & \textbf{GMM}    & 10.27                  & 89.05                            & 10.33                            & 37.31                             & 89.05                            & 10.33                            & 39.14                             & 90.87                            & 10.33                            & 40.25                             & 92.08                            & 10.33                            & 40.81                             & 93.09                            & 10.33                            & 41.49                            \\
                                   & \textbf{IV}     & 6.15                   & 94.20                            & 5.03                             & 36.95                             & 96.28                            & 5.03                             & 38.25                             & 97.44                            & 5.03                             & 39.14                             & 98.02                            & 5.03                             & 39.69                             & 98.47                            & 5.03                             & 40.20                            \\
                                   & \textbf{IF}     & 6.37                   & 98.13                            & 4.42                             & 56.68                             & 99.21                            & 4.50                             & 58.61                             & 99.40                            & 4.57                             & 59.81                             & 99.33                            & 4.53                             & 58.27                             & 99.37                            & 4.74                             & 61.87                            \\\hline
                                   & \textbf{AVG}    & 5.26                   & 96.27                            & 4.68                             & 47.57                             & 96.91                            & 4.70                             & 49.05                             & 97.54                            & 4.71                             & 50.06                             & 97.89                            & 4.70                             & 50.27                             & 98.19                            & 4.74                             & 51.44                            \\ 
                                   & \textbf{WCS}    & -                      & 89.05                            & 1.30                             & 36.95                             & 89.05                            & 1.30                             & 38.25                             & 90.87                            & 1.30                             & 39.14                             & 92.08                            & 1.30                             & 39.69                             & 93.09                            & 1.30                             & 40.20                            \\ \hline\hline

\multicolumn{2}{c|}{\multirow{2}{*}{\textbf{Model}}} & \textbf{B0 (\%)}                  & \multicolumn{3}{c|}{\textbf{$\epsilon=0.02$ (\%)}}                                                     & \multicolumn{3}{c|}{\textbf{$\epsilon=0.04$ (\%)}}                                                     & \multicolumn{3}{c|}{\textbf{$\epsilon=0.06$ (\%)}}                                                     & \multicolumn{3}{c|}{\textbf{$\epsilon=0.08$ (\%)}}                                                     & \multicolumn{3}{c}{\textbf{$\epsilon=0.1$ (\%)}}                        \\
\multicolumn{2}{c|}{}                                & \textbf{CER}                      & \multicolumn{3}{c|}{\textbf{CER}}                                                                          & \multicolumn{3}{c|}{\textbf{CER}}                                                                          & \multicolumn{3}{c|}{\textbf{CER}}                                                                          & \multicolumn{3}{c|}{\textbf{CER}}                                                                          & \multicolumn{3}{c}{\textbf{CER}}                                                                         \\ \hline
\multirow{3}{*}{\textbf{ASR}}     & \textbf{WV}     & 5.58                   & \multicolumn{3}{c|}{7.53}                                                                               & \multicolumn{3}{c|}{7.73}                                                                               & \multicolumn{3}{c|}{7.89}                                                                               & \multicolumn{3}{c|}{8.12}                                                                               & \multicolumn{3}{c}{8.33}                                                                               \\
                                  & \textbf{TF}     & 6.77                   & \multicolumn{3}{c|}{8.69}                                                                               & \multicolumn{3}{c|}{8.88}                                                                               & \multicolumn{3}{c|}{9.06}                                                                               & \multicolumn{3}{c|}{9.18}                                                                               & \multicolumn{3}{c}{9.39}                                                                               \\
                                  & \textbf{AVG}    & 6.18                   & \multicolumn{3}{c|}{8.11}                                                                               & \multicolumn{3}{c|}{8.31}                                                                               & \multicolumn{3}{c|}{8.48}                                                                              & \multicolumn{3}{c|}{8.65}                                                                              & \multicolumn{3}{c}{8.86}                                                                               \\ \hline
\end{tabular}}

\begin{tablenotes}[flushleft]
\item[] \vspace{-1pt}\hspace{-2pt}\footnotesize \textbf{AVG}: the average-case scenario, \textbf{WCS}: the worst-case scenario. \textbf{EP}: ECAPA-TDNN, \textbf{XV}: X-vector, \textbf{GMM}: GMM-UBM, \textbf{IV}: ivector-PLDA, \textbf{IF}: iFlytek.\\ \textbf{WV}: wav2vec2-transformer, \textbf{TF}: transformer.
\end{tablenotes}

\label{tab:chinese}
\end{threeparttable}
\end{table*}
\begin{table*}\centering
\begin{threeparttable}[tt]
\setlength{\abovecaptionskip}{5pt}% 
\setlength{\belowcaptionskip}{0pt}%

\caption{Cross-language performance with French language (WER).}

\scriptsize
\setlength{\tabcolsep}{2mm}{
\begin{tabular}{cl|c|ccccccccccccccc}
\hline
\multicolumn{2}{c|}{\textbf{Model}}& \textbf{B0 (\%)}         & \multicolumn{3}{c|}{\textbf{$\epsilon=0.02$ (\%)}}              & \multicolumn{3}{c|}{\textbf{$\epsilon=0.04$ (\%)}}              & \multicolumn{3}{c|}{\textbf{$\epsilon=0.06$ (\%)}}              & \multicolumn{3}{c|}{\textbf{$\epsilon=0.08$ (\%)}}              & \multicolumn{3}{c}{\textbf{$\epsilon=0.1$ (\%)}} \\\hline
\multirow{3}{*}{\textbf{ASR}}     & \textbf{WV}     & 13.18                    & \multicolumn{3}{c|}{16.50}                                       & \multicolumn{3}{c|}{16.98}                                       & \multicolumn{3}{c|}{17.36}                                       & \multicolumn{3}{c|}{17.61}                                       & \multicolumn{3}{c}{17.95}                        \\
                                   & \textbf{CD}     & 19.47                    & \multicolumn{3}{c|}{25.27}                                       & \multicolumn{3}{c|}{26.30}                                       & \multicolumn{3}{c|}{26.71}                                       & \multicolumn{3}{c|}{27.32}                                       & \multicolumn{3}{c}{27.88}                        \\
                                   & \textbf{AVG}    & 16.33                    & \multicolumn{3}{c|}{20.89}                                       & \multicolumn{3}{c|}{21.64}                                       & \multicolumn{3}{c|}{22.04}                                       & \multicolumn{3}{c|}{22.47}                                       & \multicolumn{3}{c}{22.92}                        \\ \hline
\end{tabular}}

\begin{tablenotes}[flushleft]
\item[] \vspace{-1pt}\hspace{-2pt}\footnotesize \textbf{AVG}: the average-case scenario. \textbf{WV}: wav2vec2, \textbf{TF}: transformer.
\end{tablenotes}

\label{tab:fr}
\end{threeparttable}
\end{table*}
\begin{table*}\centering
\begin{threeparttable}[tt]
\setlength{\abovecaptionskip}{5pt}% 
\setlength{\belowcaptionskip}{0pt}%

\caption{Cross-language performance with Italian input (WER).}

\scriptsize
\setlength{\tabcolsep}{2mm}{
\begin{tabular}{cl|c|ccccccccccccccc}
\hline
\multicolumn{2}{c|}{\textbf{Model}}& \textbf{B0 (\%)}         & \multicolumn{3}{c|}{\textbf{$\epsilon=0.02$ (\%)}}              & \multicolumn{3}{c|}{\textbf{$\epsilon=0.04$ (\%)}}              & \multicolumn{3}{c|}{\textbf{$\epsilon=0.06$ (\%)}}              & \multicolumn{3}{c|}{\textbf{$\epsilon=0.08$ (\%)}}              & \multicolumn{3}{c}{\textbf{$\epsilon=0.1$ (\%)}} \\\hline
\multirow{3}{*}{\textbf{ASR}}     & \textbf{WV}     & 12.74                                 & \multicolumn{3}{c|}{16.06}                                       & \multicolumn{3}{c|}{16.42}                                       & \multicolumn{3}{c|}{16.81}                                       & \multicolumn{3}{c|}{17.12}                                       & \multicolumn{3}{c}{17.27}                        \\
                                   & \textbf{CD}     & 17.48                                 & \multicolumn{3}{c|}{24.38}                                       & \multicolumn{3}{c|}{25.15}                                       & \multicolumn{3}{c|}{25.57}                                       & \multicolumn{3}{c|}{25.75}                                       & \multicolumn{3}{c}{26.05}                        \\
                                   & \textbf{AVG}    & 15.11                                 & \multicolumn{3}{c|}{20.22}                                       & \multicolumn{3}{c|}{20.79}                                       & \multicolumn{3}{c|}{21.19}                                       & \multicolumn{3}{c|}{21.44}                                       & \multicolumn{3}{c}{21.66}                        \\ \hline
\end{tabular}}

\begin{tablenotes}[flushleft]
\item[] \vspace{-1pt}\hspace{-2pt}\footnotesize \textbf{AVG}: the average-case scenario. \textbf{WV}: wav2vec2, \textbf{TF}: transformer.
\end{tablenotes}

\label{tab:it}
\end{threeparttable}
\end{table*}

\begin{table*}\centering
\begin{threeparttable}[tt]
\setlength{\abovecaptionskip}{5pt}% 
\setlength{\belowcaptionskip}{0pt}%

\caption{Comparison with existing works under A2 (mean filter).}

\scriptsize
\setlength{\tabcolsep}{2mm}{
\begin{tabular}{cl|c|crr|crr|crr|crr|crr}
\hline
\multicolumn{2}{c|}{\multirow{2}{*}{\textbf{Model}}} & \textbf{Raw (\%)}  & \multicolumn{3}{c|}{\textbf{NSF (\%)}}                                                    & \multicolumn{3}{c|}{\textbf{HFGAN (\%)}}                                                    & \multicolumn{3}{c|}{\textbf{McAdams (\%)}}                                                    & \multicolumn{3}{c|}{\textbf{VoiceMask (\%)}}                                                     & \multicolumn{3}{c}{\textbf{\sys (\%)}}                                                      \\
\multicolumn{2}{c|}{}                                & \textbf{EER} & \textbf{MMR}               & \multicolumn{1}{c}{\textbf{WMR}} & \multicolumn{1}{c|}{\textbf{EER}} & \textbf{MMR}              & \multicolumn{1}{c}{\textbf{WMR}} & \multicolumn{1}{c|}{\textbf{EER}} & \textbf{MMR}              & \multicolumn{1}{c}{\textbf{WMR}} & \multicolumn{1}{c|}{\textbf{EER}} & \textbf{MMR}               & \multicolumn{1}{c}{\textbf{WMR}} & \multicolumn{1}{c|}{\textbf{EER}} & \textbf{MMR}               & \multicolumn{1}{c}{\textbf{WMR}} & \multicolumn{1}{c}{\textbf{EER}} \\ \hline
\multirow{6}{*}{\textbf{ASV}}     & \textbf{EP}      & 0.73         & 90.00  & 1.49                             & 27.56                             & 89.73 & 1.49                            & 27.37                             & 71.68 & 1.49                            & 24.47                             & 78.63  & 1.49                            & 18.32                             & 99.85  & 1.49                            & 64.96                            \\
                                   & \textbf{XV}     & 6.37         & 100.0  & 0.00                             & 34.01                             & 100.0 & 0.00                            & 37.18                             & 100.0 & 0.00                            & 43.74                             & 100.0  & 0.00                            & 40.50                             & 100.0  & 0.00                            & 42.52                            \\
                                   & \textbf{GMM}    & 12.73        & 77.48  & 9.65                             & 39.41                             & 74.38 & 9.65                            & 39.22                             & 68.69 & 9.65                            & 38.11                             & 91.69  & 9.65                            & 47.40                             & 79.55  & 13.00                           & 39.30                            \\
                                   & \textbf{IV}     & 4.90         & 84.05  & 3.95                             & 28.93                             & 80.79 & 3.95                            & 28.88                             & 90.74 & 3.95                            & 35.65                             & 90.74  & 3.95                            & 35.66                             & 92.02  & 4.70                            & 37.91                            \\
                                   & \textbf{IF}     & 10.44        & 76.41  & 6.79                             & 35.99                             & 95.99 & 6.64                            & 49.75                             & 98.70 & 6.22                            & 77.42                             & 100.0  & 6.49                            & 91.01                             & 97.67  & 6.49                            & 47.21                            \\\hline
                                   & \textbf{AVG}    & 7.03         & 85.59  & 4.38                             & 33.18                             & 88.18 & 4.35                            & 36.48                             & 85.96 & 4.26                            & 43.88                             & 92.21  & 4.32                            & 46.58                             & 93.82  & 5.14                            & 46.38                            \\ 
                                   & \textbf{WCS}    & -            & 76.41  & 0.00                             & 27.56                             & 74.38 & 0.00                            & 27.37                             & 68.69 & 0.00                            & 24.47                             & 78.63  & 0.00                            & 18.32                             & 79.55  & 0.00                            & 37.91                            \\ \hline
\end{tabular}}

\begin{tablenotes}[flushleft]
\item[] \vspace{-1pt}\hspace{-2pt}\footnotesize \textbf{AVG}: the average-case scenario, \textbf{WCS}: the worst-case scenario. \textbf{EP}: ECAPA-TDNN, \textbf{XV}: X-vector, \textbf{GMM}: GMM-UBM, \textbf{IV}: ivector-PLDA, \textbf{IF}: iFlytek.
\end{tablenotes}

\label{tab:meanfilter}
\end{threeparttable}
\end{table*}
\begin{table*}
\centering
\begin{threeparttable}[tt]
\setlength{\abovecaptionskip}{5pt}% 
\setlength{\belowcaptionskip}{0pt}%

\caption{Comparison with existing works under A2 (median filter).}

\scriptsize
\setlength{\tabcolsep}{2mm}{
\begin{tabular}{cl|c|crr|crr|crr|crr|crr}
\hline
\multicolumn{2}{c|}{\multirow{2}{*}{\textbf{Model}}} & \textbf{Raw (\%)}  & \multicolumn{3}{c|}{\textbf{NSF (\%)}}                                                    & \multicolumn{3}{c|}{\textbf{HFGAN (\%)}}                                                    & \multicolumn{3}{c|}{\textbf{McAdams (\%)}}                                                    & \multicolumn{3}{c|}{\textbf{VoiceMask (\%)}}                                                     & \multicolumn{3}{c}{\textbf{\sys (\%)}}                                                      \\
\multicolumn{2}{c|}{}                                & \multicolumn{1}{c|}{\textbf{EER}}    & \textbf{MMR}               & \multicolumn{1}{c}{\textbf{WMR}} & \multicolumn{1}{c|}{\textbf{EER}} & \textbf{MMR}              & \multicolumn{1}{c}{\textbf{WMR}} & \multicolumn{1}{c|}{\textbf{EER}} & \textbf{MMR}              & \multicolumn{1}{c}{\textbf{WMR}} & \multicolumn{1}{c|}{\textbf{EER}} & \textbf{MMR}               & \multicolumn{1}{c}{\textbf{WMR}} & \multicolumn{1}{c|}{\textbf{EER}} & \textbf{MMR}               & \multicolumn{1}{c}{\textbf{WMR}} & \multicolumn{1}{c}{\textbf{EER}} \\ \hline
\multirow{6}{*}{\textbf{ASV}}      & \textbf{EP}     & 1.49            & 89.24  & 1.49                             & 24.70                             & 93.05 & 1.49                            & 27.48                             & 75.46 & 1.49                            & 32.14                             & 86.64  & 1.49                            & 25.04                             & 93.51  & 1.49                            & 33.93                            \\
                                   & \textbf{XV}     & 12.63           & 100.0  & 0.00                             & 30.42                             & 100.0 & 0.00                            & 31.87                             & 100.0 & 0.00                            & 44.47                             & 100.0  & 0.00                            & 41.76                             & 100.0  & 0.00                            & 41.34                            \\
                                   & \textbf{GMM}    & 18.27           & 87.15  & 9.65                             & 46.46                             & 89.59 & 9.65                            & 48.54                             & 84.01 & 9.65                            & 45.30                             & 92.40  & 9.65                            & 52.53                             & 87.36  & 13.00                           & 47.35                            \\
                                   & \textbf{IV}     & 9.25            & 87.52  & 3.95                             & 30.45                             & 87.56 & 3.95                            & 31.45                             & 82.27 & 3.95                            & 33.54                             & 92.02  & 3.95                            & 39.63                             & 87.77  & 4.70                            & 32.43                            \\
                                   & \textbf{IF}     & 21.89           & 100.0  & 6.79                             & 62.52                             & 100.0 & 6.60                            & 65.21                             & 99.62 & 6.53                            & 68.09                             & 100.0  & 6.37                            & 77.48                             & 100.0  & 6.60                            & 69.67                            \\\hline
                                   & \textbf{AVG}    & 12.71           & 92.78  & 4.38                             & 38.91                             & 94.04 & 4.34                            & 40.91                             & 88.27 & 4.32                            & 44.71                             & 94.21  & 4.29                            & 47.29                             & 93.73  & 5.16                            & 44.94                            \\ 
                                   & \textbf{WCS}    & -               & 87.15  & 0.00                             & 24.70                             & 87.56 & 0.00                            & 27.48                             & 75.46 & 0.00                            & 32.14                             & 86.64  & 0.00                            & 25.04                             & 87.36  & 0.00                            & 32.43                            \\ \hline
\end{tabular}}

\begin{tablenotes}[flushleft]
\item[] \vspace{-1pt}\hspace{-2pt}\footnotesize \textbf{AVG}: the average-case scenario, \textbf{WCS}: the worst-case scenario. \textbf{EP}: ECAPA-TDNN, \textbf{XV}: X-vector, \textbf{GMM}: GMM-UBM, \textbf{IV}: ivector-PLDA, \textbf{IF}: iFlytek.
\end{tablenotes}

\label{tab:medfilter}
\end{threeparttable}
\end{table*}
\begin{table*}\centering
\begin{threeparttable}[tt]
\setlength{\abovecaptionskip}{5pt}% 
\setlength{\belowcaptionskip}{0pt}%

\caption{{Comparison with existing works under A2 (band-pass filter 50$\sim$7,500Hz).}}

\scriptsize
\setlength{\tabcolsep}{2mm}{
\begin{tabular}{cl|c|crr|crr|crr|crr|crr}
\hline
\multicolumn{2}{c|}{\multirow{2}{*}{\textbf{Model}}} & \textbf{Raw (\%)}  & \multicolumn{3}{c|}{\textbf{NSF (\%)}}                                                    & \multicolumn{3}{c|}{\textbf{HFGAN (\%)}}                                                    & \multicolumn{3}{c|}{\textbf{McAdams (\%)}}                                                    & \multicolumn{3}{c|}{\textbf{VoiceMask (\%)}}                                                     & \multicolumn{3}{c}{\textbf{\sys (\%)}}                                                      \\
\multicolumn{2}{c|}{}                                & \textbf{EER} & \textbf{MMR}               & \multicolumn{1}{c}{\textbf{WMR}} & \multicolumn{1}{c|}{\textbf{EER}} & \textbf{MMR}              & \multicolumn{1}{c}{\textbf{WMR}} & \multicolumn{1}{c|}{\textbf{EER}} & \textbf{MMR}              & \multicolumn{1}{c}{\textbf{WMR}} & \multicolumn{1}{c|}{\textbf{EER}} & \textbf{MMR}               & \multicolumn{1}{c}{\textbf{WMR}} & \multicolumn{1}{c|}{\textbf{EER}} & \textbf{MMR}               & \multicolumn{1}{c}{\textbf{WMR}} & \multicolumn{1}{c}{\textbf{EER}} \\ \hline
\multirow{6}{*}{\textbf{ASV}}      & \textbf{EP}     & 0.73         & 90.57     & 1.49                             & 29.58                             & 90.95   & 1.49                            & 30.27                             & 73.05   & 1.49                            & 24.85                             & 80.84    & 1.49                            & 20.38                             & 99.96    & 1.49                          & 64.92                              \\
                                   & \textbf{XV}     & 5.95         & 100.0     & 0.00                             & 34.73                             & 100.0   & 0.00                            & 36.22                             & 100.0   & 0.00                            & 42.79                             & 100.0    & 0.00                            & 38.36                             & 100.0    & 0.00                          & 40.27                              \\
                                   & \textbf{GMM}    & 10.90        & 76.86     & 9.65                             & 39.25                             & 73.64   & 9.65                            & 39.21                             & 67.81   & 9.65                            & 39.37                             & 93.06    & 9.65                            & 48.35                             & 78.97    & 9.65                          & 37.00                              \\
                                   & \textbf{IV}     & 5.09         & 84.30     & 3.95                             & 28.83                             & 81.45   & 3.95                            & 30.65                             & 73.47   & 3.95                            & 29.12                             & 88.31    & 3.95                            & 36.00                             & 90.70    & 3.95                          & 36.34                              \\
                                   & \textbf{IF}     & 9.16         & 80.65     & 6.57                             & 36.24                             & 94.08   & 6.45                            & 45.53                             & 97.44   & 5.88                            & 77.26                             & 100.0    & 6.41                            & 89.58                             & 98.21    & 6.68                          & 48.15                              \\\hline
                                   & \textbf{AVG}    & 6.37         & 86.48     & 4.33                             & 33.73                             & 88.03   & 4.31                            & 36.38                             & 82.35   & 4.19                            & 42.68                             & 92.44    & 4.30                            & 46.53                             & 93.57    & 4.35                          & 45.34                              \\ 
                                   & \textbf{WCS}    & -            & 76.86     & 0.00                             & 28.83                             & 73.64   & 0.00                            & 30.27                             & 67.81   & 0.00                            & 24.85                             & 80.84    & 0.00                            & 20.38                             & 78.97    & 0.00                          & 36.34                              \\ \hline
\end{tabular}}

\begin{tablenotes}[flushleft]
\item[] \vspace{-1pt}\hspace{-2pt}\footnotesize \textbf{AVG}: the average-case scenario, \textbf{WCS}: the worst-case scenario. \textbf{EP}: ECAPA-TDNN, \textbf{XV}: X-vector, \textbf{GMM}: GMM-UBM, \textbf{IV}: ivector-PLDA, \textbf{IF}: iFlytek.
\end{tablenotes}

\label{tab:bandpass}
\end{threeparttable}
\end{table*}
\begin{table*}\centering
\begin{threeparttable}[tt]
\setlength{\abovecaptionskip}{5pt}% 
\setlength{\belowcaptionskip}{0pt}%

\caption{{Comparison with existing works under A2 (quantization 32 bits $\rightarrow$ 8 bits).}}

\scriptsize
\setlength{\tabcolsep}{2mm}{
\begin{tabular}{cl|c|crr|crr|crr|crr|crr}
\hline
\multicolumn{2}{c|}{\multirow{2}{*}{\textbf{Model}}} & \textbf{Raw (\%)}  & \multicolumn{3}{c|}{\textbf{NSF (\%)}}                                                    & \multicolumn{3}{c|}{\textbf{HFGAN (\%)}}                                                    & \multicolumn{3}{c|}{\textbf{McAdams (\%)}}                                                    & \multicolumn{3}{c|}{\textbf{VoiceMask (\%)}}                                                     & \multicolumn{3}{c}{\textbf{\sys (\%)}}                                                      \\
\multicolumn{2}{c|}{}                                & \textbf{EER} & \textbf{MMR}               & \multicolumn{1}{c}{\textbf{WMR}} & \multicolumn{1}{c|}{\textbf{EER}} & \textbf{MMR}              & \multicolumn{1}{c}{\textbf{WMR}} & \multicolumn{1}{c|}{\textbf{EER}} & \textbf{MMR}              & \multicolumn{1}{c}{\textbf{WMR}} & \multicolumn{1}{c|}{\textbf{EER}} & \textbf{MMR}               & \multicolumn{1}{c}{\textbf{WMR}} & \multicolumn{1}{c|}{\textbf{EER}} & \textbf{MMR}               & \multicolumn{1}{c}{\textbf{WMR}} & \multicolumn{1}{c}{\textbf{EER}} \\ \hline
\multirow{6}{*}{\textbf{ASV}}      & \textbf{EP}     & 0.92         & 91.49   & 1.49                             & 29.14                             & 92.86  & 1.49                            & 33.17                             & 75.53  & 1.49                            & 27.82                             & 82.18   & 1.49                            & 22.75                             & 96.99    & 1.49                            & 44.58                              \\
                                   & \textbf{XV}     & 9.50         & 100.0   & 0.00                             & 32.18                             & 100.0  & 0.00                            & 35.23                             & 100.0  & 0.00                            & 48.70                             & 100.0   & 0.00                            & 44.66                             & 100.0    & 0.00                            & 34.81                              \\
                                   & \textbf{GMM}    & 13.88        & 79.09   & 9.65                             & 40.35                             & 82.36  & 9.65                            & 43.20                             & 72.77  & 9.65                            & 41.32                             & 89.71   & 9.65                            & 47.48                             & 81.12    & 13.00                           & 41.34                              \\
                                   & \textbf{IV}     & 5.12         & 83.72   & 3.95                             & 28.02                             & 83.60  & 3.95                            & 29.93                             & 77.19  & 3.95                            & 30.13                             & 86.57   & 3.95                            & 36.12                             & 79.67    & 4.70                            & 41.34                              \\
                                   & \textbf{IF}     & 10.10        & 92.82   & 6.79                             & 43.55                             & 98.05  & 6.64                            & 49.05                             & 89.12  & 6.60                            & 59.72                             & 100.0   & 6.37                            & 67.96                             & 90.99    & 6.49                            & 41.97                              \\\hline
                                   & \textbf{AVG}    & 7.90         & 89.42   & 4.38                             & 34.65                             & 91.38  & 4.35                            & 38.11                             & 82.92  & 4.34                            & 41.54                             & 91.69   & 4.29                            & 43.79                             & 89.75    & 5.14                            & 40.81                              \\ 
                                   & \textbf{WCS}    & -            & 79.09   & 0.00                             & 28.02                             & 82.36  & 0.00                            & 29.93                             & 72.77  & 0.00                            & 27.82                             & 82.18   & 0.00                            & 22.75                             & 79.67    & 0.00                            & 34.81                              \\ \hline
\end{tabular}}

\begin{tablenotes}[flushleft]
\item[] \vspace{-1pt}\hspace{-2pt}\footnotesize \textbf{AVG}: the average-case scenario, \textbf{WCS}: the worst-case scenario. \textbf{EP}: ECAPA-TDNN, \textbf{XV}: X-vector, \textbf{GMM}: GMM-UBM, \textbf{IV}: ivector-PLDA, \textbf{IF}: iFlytek.
\end{tablenotes}

\label{tab:quantization}
\end{threeparttable}
\end{table*}
\begin{table*}\centering
\begin{threeparttable}[tt]
\setlength{\abovecaptionskip}{5pt}% 
\setlength{\belowcaptionskip}{0pt}%

\caption{{Comparison with existing works under A2 (squeezing 0.8$\times$ the sampling rate).}}

\scriptsize
\setlength{\tabcolsep}{2mm}{
\begin{tabular}{cl|c|crr|crr|crr|crr|crr}
\hline
\multicolumn{2}{c|}{\multirow{2}{*}{\textbf{Model}}} & \textbf{Raw (\%)}  & \multicolumn{3}{c|}{\textbf{NSF (\%)}}                                                    & \multicolumn{3}{c|}{\textbf{HFGAN (\%)}}                                                    & \multicolumn{3}{c|}{\textbf{McAdams (\%)}}                                                    & \multicolumn{3}{c|}{\textbf{VoiceMask (\%)}}                                                     & \multicolumn{3}{c}{\textbf{\sys (\%)}}                                                      \\
\multicolumn{2}{c|}{}                                & \textbf{EER} & \textbf{MMR}               & \multicolumn{1}{c}{\textbf{WMR}} & \multicolumn{1}{c|}{\textbf{EER}} & \textbf{MMR}              & \multicolumn{1}{c}{\textbf{WMR}} & \multicolumn{1}{c|}{\textbf{EER}} & \textbf{MMR}              & \multicolumn{1}{c}{\textbf{WMR}} & \multicolumn{1}{c|}{\textbf{EER}} & \textbf{MMR}               & \multicolumn{1}{c}{\textbf{WMR}} & \multicolumn{1}{c|}{\textbf{EER}} & \textbf{MMR}               & \multicolumn{1}{c}{\textbf{WMR}} & \multicolumn{1}{c}{\textbf{EER}} \\ \hline
\multirow{6}{*}{\textbf{ASV}}      & \textbf{EP}     & 0.80         & 91.15    & 1.49                             & 30.08                             & 90.80   & 1.49                            & 30.50                             & 71.60   & 1.49                            & 25.12                             & 71.37    & 1.49                            & 16.72                             & 99.62    & 1.49                           & 55.92                              \\
                                   & \textbf{XV}     & 6.64         & 100.0    & 0.00                             & 33.15                             & 100.0   & 0.00                            & 36.49                             & 100.0   & 0.00                            & 39.73                             & 100.0    & 0.00                            & 30.12                             & 100.0    & 0.00                           & 36.76                              \\
                                   & \textbf{GMM}    & 12.34        & 77.77    & 9.65                             & 39.48                             & 74.38   & 9.65                            & 40.01                             & 67.44   & 9.65                            & 36.82                             & 89.09    & 9.65                            & 45.90                             & 76.12    & 13.00                          & 35.63                              \\
                                   & \textbf{IV}     & 4.97         & 85.08    & 3.95                             & 29.20                             & 84.63   & 3.95                            & 31.85                             & 72.60   & 3.95                            & 29.82                             & 82.89    & 3.95                            & 32.50                             & 88.14    & 4.70                           & 35.59                              \\
                                   & \textbf{IF}     & 9.98         & 91.91    & 6.72                             & 42.31                             & 98.05   & 6.64                            & 49.89                             & 87.54   & 6.57                            & 67.53                             & 100.0    & 6.53                            & 67.29                             & 97.56    & 6.53                           & 46.66                              \\\hline
                                   & \textbf{AVG}    & 6.95         & 89.18    & 4.36                             & 34.84                             & 89.57   & 4.35                            & 37.75                             & 79.84   & 4.33                            & 39.80                             & 88.67    & 4.32                            & 38.50                             & 92.29    & 5.14                           & 42.11                              \\ 
                                   & \textbf{WCS}    & -            & 77.77    & 0.00                             & 29.20                             & 74.38   & 0.00                            & 30.50                             & 67.44   & 0.00                            & 25.12                             & 71.37    & 0.00                            & 16.72                             & 76.12    & 0.00                           & 35.59                              \\ \hline
\end{tabular}}

\begin{tablenotes}[flushleft]
\item[] \vspace{-1pt}\hspace{-2pt}\footnotesize \textbf{AVG}: the average-case scenario, \textbf{WCS}: the worst-case scenario. \textbf{EP}: ECAPA-TDNN, \textbf{XV}: X-vector, \textbf{GMM}: GMM-UBM, \textbf{IV}: ivector-PLDA, \textbf{IF}: iFlytek.
\end{tablenotes}

\label{tab:squeezing}
\end{threeparttable}
\end{table*}
\begin{table*}\centering
\begin{threeparttable}[tt]
\setlength{\abovecaptionskip}{5pt}% 
\setlength{\belowcaptionskip}{0pt}%

\caption{{Comparison with existing works under A2 (MP3 compression).}}

\scriptsize
\setlength{\tabcolsep}{2mm}{
\begin{tabular}{cl|c|crr|crr|crr|crr|crr}
\hline
\multicolumn{2}{c|}{\multirow{2}{*}{\textbf{Model}}} & \textbf{Raw (\%)}  & \multicolumn{3}{c|}{\textbf{NSF (\%)}}                                                    & \multicolumn{3}{c|}{\textbf{HFGAN (\%)}}                                                    & \multicolumn{3}{c|}{\textbf{McAdams (\%)}}                                                    & \multicolumn{3}{c|}{\textbf{VoiceMask (\%)}}                                                     & \multicolumn{3}{c}{\textbf{\sys (\%)}}                                                      \\
\multicolumn{2}{c|}{}                                & \textbf{EER} & \textbf{MMR}               & \multicolumn{1}{c}{\textbf{WMR}} & \multicolumn{1}{c|}{\textbf{EER}} & \textbf{MMR}              & \multicolumn{1}{c}{\textbf{WMR}} & \multicolumn{1}{c|}{\textbf{EER}} & \textbf{MMR}              & \multicolumn{1}{c}{\textbf{WMR}} & \multicolumn{1}{c|}{\textbf{EER}} & \textbf{MMR}               & \multicolumn{1}{c}{\textbf{WMR}} & \multicolumn{1}{c|}{\textbf{EER}} & \textbf{MMR}               & \multicolumn{1}{c}{\textbf{WMR}} & \multicolumn{1}{c}{\textbf{EER}} \\ \hline
\multirow{6}{*}{\textbf{ASV}}      & \textbf{EP}     & 0.73         & 91.83  & 1.49                             & 30.23                             & 91.68 & 1.49                            & 31.64                             & 73.09 & 1.49                            & 25.19                             & 81.37  & 1.49                            & 20.34                             & 93.89   & 1.49                            & 36.18                             \\
                                   & \textbf{XV}     & 5.86         & 100.0  & 0.00                             & 33.70                             & 100.0 & 0.00                            & 36.22                             & 100.0 & 0.00                            & 43.11                             & 100.0  & 0.00                            & 38.40                             & 100.0   & 0.00                            & 35.31                             \\
                                   & \textbf{GMM}    & 12.52        & 78.64  & 9.65                             & 39.93                             & 77.19 & 9.65                            & 40.31                             & 69.88 & 9.65                            & 39.10                             & 90.95  & 9.65                            & 47.65                             & 76.45   & 13.00                           & 35.87                             \\
                                   & \textbf{IV}     & 5.11         & 83.51  & 3.95                             & 28.17                             & 81.86 & 3.95                            & 30.76                             & 75.50 & 3.95                            & 29.46                             & 88.18  & 3.95                            & 36.12                             & 86.24   & 4.70                            & 32.25                             \\
                                   & \textbf{IF}     & 8.97         & 78.28  & 6.79                             & 33.59                             & 93.09 & 6.79                            & 43.78                             & 96.37 & 6.79                            & 76.49                             & 100.0  & 6.79                            & 86.49                             & 97.17   & 5.84                            & 44.59                             \\\hline
                                   & \textbf{AVG}    & 6.64         & 86.45  & 4.38                             & 33.12                             & 88.76 & 4.38                            & 36.54                             & 82.97 & 4.38                            & 42.67                             & 92.10  & 4.38                            & 45.80                             & 90.75   & 5.01                            & 36.84                             \\ 
                                   & \textbf{WCS}    & -            & 78.28  & 0.00                             & 28.17                             & 77.19 & 0.00                            & 30.76                             & 69.88 & 0.00                            & 25.19                             & 81.37  & 0.00                            & 20.34                             & 76.45   & 0.00                            & 32.25                             \\ \hline
\end{tabular}}

\begin{tablenotes}[flushleft]
\item[] \vspace{-1pt}\hspace{-2pt}\footnotesize \textbf{AVG}: the average-case scenario, \textbf{WCS}: the worst-case scenario. \textbf{EP}: ECAPA-TDNN, \textbf{XV}: X-vector, \textbf{GMM}: GMM-UBM, \textbf{IV}: ivector-PLDA, \textbf{IF}: iFlytek.
\end{tablenotes}

\label{tab:mp3}
\end{threeparttable}
\end{table*}

\begin{table*}\centering
\begin{threeparttable}[tt]
\setlength{\abovecaptionskip}{5pt}% 
\setlength{\belowcaptionskip}{0pt}%

\caption{Targeted conversion performance.}

\scriptsize
\setlength{\tabcolsep}{2mm}{
\begin{tabular}{cl|rrrrrrrrrr|r}
\hline
\multirow{3}{*}{\textbf{ECAPA}}    & \textbf{Target Speaker}     & 1272        & 1462    & 1673   & 174     & 1919    & 1988    & 1993   & 2035   & 2078   & 2086    & \textbf{Total} \\\hline
                                   & \textbf{Success Trials}     & 2620        & 2620    & 2620   & 2620    & 2619    & 2620    & 2619   & 2619   & 2620   & 2620    & 26197 \\
                                   & \textbf{Success Rate (\%)}  & 100.0       & 100.0   & 100.0  & 100.0   & 99.96   & 100.0   & 99.96  & 99.96  & 100.0  & 100.0   & 99.99 \\\hline

\end{tabular}}
\begin{tablenotes}[flushleft]
\item[] \vspace{-1pt}\hspace{-2pt}\footnotesize Test set: LibriSpeech \textit{test-clean}. Target speakers are from LibriSpeech \textit{dev-clean}.
\end{tablenotes}

\label{tab:exactlymatch}
\end{threeparttable}
\end{table*}
\begin{table*}\centering
\begin{threeparttable}[tt]
\setlength{\abovecaptionskip}{5pt}% 
\setlength{\belowcaptionskip}{0pt}%

\caption{The effectiveness of \emph{Throttle}.}

\scriptsize
\setlength{\tabcolsep}{2mm}{
\begin{tabular}{cl|c|crr|crr|crr|crr|crr}
\hline
\multicolumn{2}{c|}{\multirow{2}{*}{\textbf{Model}}} & \textbf{B0 (\%)}                  & \multicolumn{3}{c|}{\textbf{$\epsilon=0.02$ (\%)}}                                                     & \multicolumn{3}{c|}{\textbf{$\epsilon=0.04$ (\%)}}                                                     & \multicolumn{3}{c|}{\textbf{$\epsilon=0.06$ (\%)}}                                                     & \multicolumn{3}{c|}{\textbf{$\epsilon=0.08$ (\%)}}                                                     & \multicolumn{3}{c}{\textbf{$\epsilon=0.1$ (\%)}}                        \\
\multicolumn{2}{c|}{}                                & \textbf{EER}                      & \textbf{MMR}               & \multicolumn{1}{c}{\textbf{WMR}} & \multicolumn{1}{c|}{\textbf{EER}} & \textbf{MMR}               & \multicolumn{1}{c}{\textbf{WMR}} & \multicolumn{1}{c|}{\textbf{EER}} & \textbf{MMR}               & \multicolumn{1}{c}{\textbf{WMR}} & \multicolumn{1}{c|}{\textbf{EER}} & \textbf{MMR}               & \multicolumn{1}{c}{\textbf{WMR}} & \multicolumn{1}{c|}{\textbf{EER}} & \textbf{MMR}               & \multicolumn{1}{c}{\textbf{WMR}} & \multicolumn{1}{c}{\textbf{EER}} \\ \hline
\multirow{6}{*}{\textbf{ASV}}     & \textbf{EP}     & 0.70                              & \multicolumn{1}{r}{97.75} & 1.07                             & 42.25                             & \multicolumn{1}{r}{99.89} & 1.07                             & 50.50                             & \multicolumn{1}{r}{99.96} & 1.07                             & 51.64                             & \multicolumn{1}{r}{99.96} & 1.07                             & 51.87                             & \multicolumn{1}{r}{99.96} & 1.07                             & 51.99                            \\
                                   & \textbf{XV}     & 6.53                              & \multicolumn{1}{r}{91.34} & 6.91                             & 35.84                             & \multicolumn{1}{r}{96.57} & 6.91                             & 40.95                             & \multicolumn{1}{r}{96.91} & 6.87                             & 41.81                             & \multicolumn{1}{r}{96.99} & 6.91                             & 42.06                             & \multicolumn{1}{r}{97.06} & 6.91                             & 42.14                            \\
                                   & \textbf{GMM}    & 11.39                             & \multicolumn{1}{r}{60.62} & 13.00                            & 29.12                             & \multicolumn{1}{r}{64.63} & 13.00                            & 31.04                             & \multicolumn{1}{r}{65.41} & 13.00                            & 31.55                             & \multicolumn{1}{r}{65.33} & 13.00                            & 31.49                             & \multicolumn{1}{r}{65.70} & 13.00                            & 31.62                            \\
                                   & \textbf{IV}     & 6.03                              & \multicolumn{1}{r}{69.17} & 4.70                             & 25.45                             & \multicolumn{1}{r}{76.16} & 4.70                             & 29.71                             & \multicolumn{1}{r}{77.40} & 4.70                             & 30.41                             & \multicolumn{1}{r}{77.60} & 4.70                             & 30.39                             & \multicolumn{1}{r}{77.93} & 4.70                             & 30.45                            \\
                                   & \textbf{IF}     & 9.44                              & \multicolumn{1}{r}{77.48} & 6.45                             & 34.50                             & \multicolumn{1}{r}{82.18} & 6.64                             & 36.64                             & \multicolumn{1}{r}{82.14} & 6.34                             & 36.93                             & \multicolumn{1}{r}{81.99} & 6.68                             & 37.12                             & \multicolumn{1}{r}{80.99} & 6.79                             & 36.99                            \\ \hline
                                   & \textbf{AVG}    & 6.82                              & \multicolumn{1}{r}{79.27} & 6.43                             & 33.43                             & \multicolumn{1}{r}{83.88} & 6.46                             & 37.77                             & \multicolumn{1}{r}{84.36} & 6.40                             & 38.47                             & \multicolumn{1}{r}{84.37} & 6.47                             & 38.59                             & \multicolumn{1}{r}{84.33} & 6.49                             & 38.64                            \\ 
                                   & \textbf{WCS}    & -                              & \multicolumn{1}{r}{60.62} & 1.07                             & 25.45                             & \multicolumn{1}{r}{64.63} & 1.07                             & 29.71                             & \multicolumn{1}{r}{65.41} & 1.07                             & 30.41                             & \multicolumn{1}{r}{65.33} & 1.07                             & 30.39                             & \multicolumn{1}{r}{65.70} & 1.07                             & 30.45                            \\ \hline\hline

\multicolumn{2}{c|}{\multirow{2}{*}{\textbf{Model}}} & \textbf{B0 (\%)}                  & \multicolumn{3}{c|}{\textbf{$\epsilon=0.02$ (\%)}}                                                     & \multicolumn{3}{c|}{\textbf{$\epsilon=0.04$ (\%)}}                                                     & \multicolumn{3}{c|}{\textbf{$\epsilon=0.06$ (\%)}}                                                     & \multicolumn{3}{c|}{\textbf{$\epsilon=0.08$ (\%)}}                                                     & \multicolumn{3}{c}{\textbf{$\epsilon=0.1$ (\%)}}                        \\
\multicolumn{2}{c|}{}                                & \textbf{WER}                      & \multicolumn{3}{c|}{\textbf{WER}}                                                                          & \multicolumn{3}{c|}{\textbf{WER}}                                                                          & \multicolumn{3}{c|}{\textbf{WER}}                                                                          & \multicolumn{3}{c|}{\textbf{WER}}                                                                          & \multicolumn{3}{c}{\textbf{WER}}                                                                         \\ \hline
\multirow{6}{*}{\textbf{ASR}}     & \textbf{DS}     & 10.46                             & \multicolumn{3}{c|}{12.22}                                                                       & \multicolumn{3}{c|}{12.43}                                                                       & \multicolumn{3}{c|}{12.43}                                                                       & \multicolumn{3}{c|}{12.48}                                                                       & \multicolumn{3}{c}{12.47}                                                                       \\
                                   & \textbf{WV}     & 9.66                              & \multicolumn{3}{c|}{11.14}                                                                       & \multicolumn{3}{c|}{11.31}                                                                       & \multicolumn{3}{c|}{11.22}                                                                       & \multicolumn{3}{c|}{11.23}                                                                       & \multicolumn{3}{c}{11.22}                                                                       \\
                                   & \textbf{CD1}    & 3.61                              & \multicolumn{3}{c|}{3.76}                                                                        & \multicolumn{3}{c|}{3.91}                                                                        & \multicolumn{3}{c|}{3.95}                                                                        & \multicolumn{3}{c|}{3.95}                                                                        & \multicolumn{3}{c}{3.96}                                                                        \\
                                   & \textbf{CD2}    & 4.23                              & \multicolumn{3}{c|}{3.85}                                                                        & \multicolumn{3}{c|}{3.94}                                                                        & \multicolumn{3}{c|}{4.06}                                                                        & \multicolumn{3}{c|}{4.07}                                                                        & \multicolumn{3}{c}{4.08}                                                                        \\
                                   & \textbf{TF}     & 3.01                              & \multicolumn{3}{c|}{2.74}                                                                        & \multicolumn{3}{c|}{2.78}                                                                        & \multicolumn{3}{c|}{2.82}                                                                        & \multicolumn{3}{c|}{2.84}                                                                        & \multicolumn{3}{c}{2.84}                                                                        \\
                                   & \textbf{AVG}    & 6.19                              & \multicolumn{3}{c|}{6.74}                                                                        & \multicolumn{3}{c|}{6.87}                                                                        & \multicolumn{3}{c|}{6.90}                                                                        & \multicolumn{3}{c|}{6.91}                                                                        & \multicolumn{3}{c}{6.91}                                                                        \\ \hline
\end{tabular}}

\begin{tablenotes}[flushleft]
\item[] \vspace{-1pt}\hspace{-2pt}\footnotesize \textbf{AVG}: the average-case scenario, \textbf{WCS}: the worst-case scenario. \textbf{EP}: ECAPA-TDNN, \textbf{XV}: X-vector, \textbf{GMM}: GMM-UBM, \textbf{IV}: ivector-PLDA, \textbf{IF}: iFlytek.\\ \textbf{DS}: DeepSpeech2, \textbf{WV}: wav2vec2, \textbf{CD1}: crdnn-transformer, \textbf{CD2}: crdnn-rnn, \textbf{TF}: transformer.
\end{tablenotes}

\label{tab:throttle}
\end{threeparttable}
\end{table*}
\begin{table*}\centering
\begin{threeparttable}[tt]
\setlength{\abovecaptionskip}{5pt}% 
\setlength{\belowcaptionskip}{0pt}%

\caption{The performance of \sys at different anonymization levels.}

\scriptsize
\setlength{\tabcolsep}{2mm}{
\begin{tabular}{cl|c|crr|crr|crr|crr|crr}
\hline
\multicolumn{2}{c|}{\multirow{2}{*}{\textbf{Model}}} & \textbf{B0 (\%)}                  & \multicolumn{3}{c|}{\textbf{$\epsilon=0.02$ (\%)}}                                                     & \multicolumn{3}{c|}{\textbf{$\epsilon=0.04$ (\%)}}                                                     & \multicolumn{3}{c|}{\textbf{$\epsilon=0.06$ (\%)}}                                                     & \multicolumn{3}{c|}{\textbf{$\epsilon=0.08$ (\%)}}                                                     & \multicolumn{3}{c}{\textbf{$\epsilon=0.1$ (\%)}}                        \\
\multicolumn{2}{c|}{}                                & \textbf{EER}                      & \textbf{MMR}               & \multicolumn{1}{c}{\textbf{WMR}} & \multicolumn{1}{c|}{\textbf{EER}} & \textbf{MMR}               & \multicolumn{1}{c}{\textbf{WMR}} & \multicolumn{1}{c|}{\textbf{EER}} & \textbf{MMR}               & \multicolumn{1}{c}{\textbf{WMR}} & \multicolumn{1}{c|}{\textbf{EER}} & \textbf{MMR}               & \multicolumn{1}{c}{\textbf{WMR}} & \multicolumn{1}{c|}{\textbf{EER}} & \textbf{MMR}               & \multicolumn{1}{c}{\textbf{WMR}} & \multicolumn{1}{c}{\textbf{EER}} \\ \hline
\multirow{6}{*}{\textbf{ASV}}     & \textbf{EP}     & 0.70                              & \multicolumn{1}{r}{98.93}  & 1.07                             & 48.05                             & \multicolumn{1}{r}{99.96}  & 1.07                            & 59.29                             & \multicolumn{1}{r}{99.96}  & 1.07                             & 62.44                             & \multicolumn{1}{r}{100}    & 1.07                             & 63.86                             & \multicolumn{1}{r}{100}    & 1.07                             & 64.58                            \\
                                   & \textbf{XV}     & 6.53                              & \multicolumn{1}{r}{89.73}  & 6.87                             & 32.12                             & \multicolumn{1}{r}{95.92}  & 6.87                             & 37.08                             & \multicolumn{1}{r}{97.63}  & 6.87                             & 38.34                             & \multicolumn{1}{r}{97.98}  & 6.91                             & 39.26                             & \multicolumn{1}{r}{98.40}  & 6.87                             & 39.77                            \\
                                   & \textbf{GMM}    & 11.39                             & \multicolumn{1}{r}{74.26}  & 11.38                            & 32.37                             & \multicolumn{1}{r}{82.72}  & 11.38                            & 36.69                             & \multicolumn{1}{r}{85.12}  & 11.38                            & 38.71                             & \multicolumn{1}{r}{86.74}  & 11.38                            & 40.10                             & \multicolumn{1}{r}{87.89}  & 11.38                            & 41.17                            \\
                                   & \textbf{IV}     & 6.03                              & \multicolumn{1}{r}{79.34}  & 6.03                             & 33.73                             & \multicolumn{1}{r}{83.60}  & 6.03                             & 35.75                             & \multicolumn{1}{r}{85.12}  & 6.03                             & 36.61                             & \multicolumn{1}{r}{85.54}  & 6.03                             & 37.20                             & \multicolumn{1}{r}{86.81}  & 6.03                             & 37.75                            \\
                                   & \textbf{IF}     & 9.44                              & \multicolumn{1}{r}{67.90}  & 6.76                             & 31.24                             & \multicolumn{1}{r}{85.27}  & 6.79                             & 38.91                             & \multicolumn{1}{r}{92.52}  & 6.76                             & 43.78                             & \multicolumn{1}{r}{95.42}  & 6.76                             & 46.71                             & \multicolumn{1}{r}{97.02}  & 6.79                             & 47.26                            \\ \hline
                                   & \textbf{AVG}    & 6.82                              & \multicolumn{1}{r}{82.03}  & 6.42                             & 35.50                             & \multicolumn{1}{r}{89.49}  & 6.43                             & 41.55                             & \multicolumn{1}{r}{92.07}  & 6.42                             & 43.98                             & \multicolumn{1}{r}{93.14}  & 6.43                             & 45.42                             & \multicolumn{1}{r}{94.02}  & 6.43                             & 46.11                            \\ 
                                   & \textbf{WCS}    & -                                  & \multicolumn{1}{r}{67.90}  & 1.07                             & 31.24                             & \multicolumn{1}{r}{82.72}  & 1.07                             & 35.75                             & \multicolumn{1}{r}{85.12}  & 1.07                             & 36.61                             & \multicolumn{1}{r}{85.54}  & 1.07                             & 37.20                             & \multicolumn{1}{r}{86.81}  & 1.07                             & 37.75                            \\ \hline\hline

\multicolumn{2}{c|}{\multirow{2}{*}{\textbf{Model}}} & \textbf{B0 (\%)}                  & \multicolumn{3}{c|}{\textbf{$\epsilon=0.02$ (\%)}}                                                     & \multicolumn{3}{c|}{\textbf{$\epsilon=0.04$ (\%)}}                                                     & \multicolumn{3}{c|}{\textbf{$\epsilon=0.06$ (\%)}}                                                     & \multicolumn{3}{c|}{\textbf{$\epsilon=0.08$ (\%)}}                                                     & \multicolumn{3}{c}{\textbf{$\epsilon=0.1$ (\%)}}                        \\
\multicolumn{2}{c|}{}                                & \textbf{WER}                      & \multicolumn{3}{c|}{\textbf{WER}}                                                                          & \multicolumn{3}{c|}{\textbf{WER}}                                                                          & \multicolumn{3}{c|}{\textbf{WER}}                                                                          & \multicolumn{3}{c|}{\textbf{WER}}                                                                          & \multicolumn{3}{c}{\textbf{WER}}                                                                         \\ \hline
\multirow{6}{*}{\textbf{ASR}}     & \textbf{DS}     & 10.46                             & \multicolumn{3}{c|}{13.22}                                                                        & \multicolumn{3}{c|}{13.43}                                                                        & \multicolumn{3}{c|}{13.63}                                                                        & \multicolumn{3}{c|}{13.82}                                                                        & \multicolumn{3}{c}{13.92}                                                                        \\
                                   & \textbf{WV}     & 9.66                              & \multicolumn{3}{c|}{11.48}                                                                        & \multicolumn{3}{c|}{11.35}                                                                        & \multicolumn{3}{c|}{11.71}                                                                        & \multicolumn{3}{c|}{11.70}                                                                        & \multicolumn{3}{c}{11.74}                                                                        \\
                                   & \textbf{CD1}    & 3.61                              & \multicolumn{3}{c|}{4.06}                                                                         & \multicolumn{3}{c|}{4.33}                                                                         & \multicolumn{3}{c|}{4.51}                                                                         & \multicolumn{3}{c|}{4.44}                                                                         & \multicolumn{3}{c}{4.59}                                                                         \\
                                   & \textbf{CD2}    & 4.23                              & \multicolumn{3}{c|}{4.32}                                                                         & \multicolumn{3}{c|}{4.50}                                                                         & \multicolumn{3}{c|}{4.60}                                                                         & \multicolumn{3}{c|}{4.82}                                                                         & \multicolumn{3}{c}{4.95}                                                                         \\
                                   & \textbf{TF}     & 3.01                              & \multicolumn{3}{c|}{2.93}                                                                         & \multicolumn{3}{c|}{2.97}                                                                         & \multicolumn{3}{c|}{3.02}                                                                         & \multicolumn{3}{c|}{3.02}                                                                         & \multicolumn{3}{c}{3.06}                                                                         \\
                                   & \textbf{AVG}    & 6.19                              & \multicolumn{3}{c|}{7.20}                                                                         & \multicolumn{3}{c|}{7.32}                                                                         & \multicolumn{3}{c|}{7.49}                                                                         & \multicolumn{3}{c|}{7.56}                                                                         & \multicolumn{3}{c}{7.65}                                                                         \\ \hline
\end{tabular}}

\begin{tablenotes}[flushleft]
\item[] \vspace{-1pt}\hspace{-2pt}\footnotesize \textbf{AVG}: the average-case scenario, \textbf{WCS}: the worst-case scenario. \textbf{EP}: ECAPA-TDNN, \textbf{XV}: X-vector, \textbf{GMM}: GMM-UBM, \textbf{IV}: ivector-PLDA, \textbf{IF}: iFlytek.\\ \textbf{DS}: DeepSpeech2, \textbf{WV}: wav2vec2, \textbf{CD1}: crdnn-transformer, \textbf{CD2}: crdnn-rnn, \textbf{TF}: transformer.
\end{tablenotes}

\label{tab:eps}
\end{threeparttable}
\end{table*}
\begin{table*}\centering
\begin{threeparttable}[tt]
\setlength{\abovecaptionskip}{5pt}% 
\setlength{\belowcaptionskip}{0pt}%

\caption{The effectiveness of the ASR loss.}

\scriptsize
\setlength{\tabcolsep}{2mm}{
\begin{tabular}{cl|crr|crr|crr}
\hline
\multicolumn{2}{c|}{\multirow{2}{*}{\textbf{Model}}} & \multicolumn{3}{c|}{\textbf{w/o ASR (\%)}}                                         & \multicolumn{3}{c|}{\textbf{w/o GPG (\%)}}                                        & \multicolumn{3}{c}{\textbf{w/~GPG (\%)}}                                           \\
\multicolumn{2}{c|}{}                                & \textbf{MMR}               & \multicolumn{1}{c}{\textbf{WMR}} & \multicolumn{1}{c|}{\textbf{EER}} & \textbf{MMR}              & \multicolumn{1}{c}{\textbf{WMR}} & \multicolumn{1}{c|}{\textbf{EER}} & \textbf{MMR}              & \multicolumn{1}{c}{\textbf{WMR}} & \multicolumn{1}{c}{\textbf{EER}} \\ \hline
\multirow{6}{*}{\textbf{ASV}}     & \textbf{EC}     & \multicolumn{1}{r}{100} & 1.07                             & 81.41                             & \multicolumn{1}{r}{99.77} & 1.07                             & 61.83                             & \multicolumn{1}{r}{100}   & 1.07                             & 64.58                            \\
                                   & \textbf{XV}     & \multicolumn{1}{r}{99.92}  & 6.91                             & 43.57                             & \multicolumn{1}{r}{85.12} & 6.91                             & 30.73                             & \multicolumn{1}{r}{98.40} & 6.87                             & 39.77                            \\
                                   & \textbf{GMM}    & \multicolumn{1}{r}{97.02}  & 13.00                            & 61.80                             & \multicolumn{1}{r}{89.59} & 13.00                            & 48.01                             & \multicolumn{1}{r}{87.89} & 11.38                            & 41.17                            \\
                                   & \textbf{IV}     & \multicolumn{1}{r}{96.86}  & 4.70                             & 41.28                             & \multicolumn{1}{r}{84.13} & 4.70                             & 32.63                             & \multicolumn{1}{r}{86.81} & 6.03                             & 37.75                            \\
                                   & \textbf{IF}     & \multicolumn{1}{r}{91.99}  & 6.57                             & 39.60                             & \multicolumn{1}{r}{54.73} & 6.53                             & 28.82                             & \multicolumn{1}{r}{97.02} & 6.79                             & 47.22                            \\ \hline
                                   & \textbf{AVG}    & \multicolumn{1}{r}{97.16}  & 6.45                             & 53.53                             & \multicolumn{1}{r}{82.67} & 6.44                             & 40.40                             & \multicolumn{1}{r}{94.02} & 6.43                             & 46.10                            \\ 
                                   & \textbf{WCS}    & \multicolumn{1}{r}{91.99}  & 1.07                             & 39.60                             & \multicolumn{1}{r}{54.73} & 1.07                             & 28.82                             & \multicolumn{1}{r}{86.81} & 1.07                             & 37.75                            \\ \hline\hline

\multicolumn{2}{c|}{\multirow{2}{*}{\textbf{Model}}} & \multicolumn{3}{c|}{\textbf{w/o ASR (\%)}}                                         & \multicolumn{3}{c|}{\textbf{w/o GPG (\%)}}                                        & \multicolumn{3}{c}{\textbf{w/~GPG (\%)}}                                           \\
\multicolumn{2}{c|}{}                                & \multicolumn{3}{c|}{\textbf{WER}}                                                                 & \multicolumn{3}{c|}{\textbf{WER}}                                                                & \multicolumn{3}{c}{\textbf{WER}}                                                                \\ \hline
\multirow{6}{*}{\textbf{ASR}}     & \textbf{DS}     & \multicolumn{3}{c|}{116.16}                                                                       & \multicolumn{3}{c|}{14.49}                                                                       & \multicolumn{3}{c}{13.92}                                                                       \\
                                   & \textbf{WV}     & \multicolumn{3}{c|}{66.51}                                                                        & \multicolumn{3}{c|}{13.19}                                                                       & \multicolumn{3}{c}{11.74}                                                                       \\
                                   & \textbf{CD1}    & \multicolumn{3}{c|}{128.14}                                                                       & \multicolumn{3}{c|}{4.81}                                                                        & \multicolumn{3}{c}{4.59}                                                                        \\
                                   & \textbf{CD2}    & \multicolumn{3}{c|}{105.85}                                                                       & \multicolumn{3}{c|}{5.65}                                                                        & \multicolumn{3}{c}{4.95}                                                                        \\
                                   & \textbf{TF}     & \multicolumn{3}{c|}{121.62}                                                                       & \multicolumn{3}{c|}{3.16}                                                                        & \multicolumn{3}{c}{3.06}                                                                        \\
                                   & \textbf{AVG}    & \multicolumn{3}{c|}{107.66}                                                                       & \multicolumn{3}{c|}{8.26}                                                                        & \multicolumn{3}{c}{7.65}                                                                        \\ \hline
\end{tabular}}

\begin{tablenotes}[flushleft]
\item[] \vspace{-1pt}\hspace{-2pt}\footnotesize \textbf{AVG}: the average-case scenario, \textbf{WCS}: the worst-case scenario. \textbf{EP}: ECAPA-TDNN, \textbf{XV}: X-vector, \textbf{GMM}: GMM-UBM, \textbf{IV}: ivector-PLDA, \textbf{IF}: iFlytek.\\ \textbf{DS}: DeepSpeech2, \textbf{WV}: wav2vec2, \textbf{CD1}: crdnn-transformer, \textbf{CD2}: crdnn-rnn, \textbf{TF}: transformer.
\end{tablenotes}

\label{tab:NoASR}
\end{threeparttable}
\end{table*}

\clearpage
\begin{figure*}[t]
    \centering
\setlength{\abovecaptionskip}{0pt}
\setlength{\belowcaptionskip}{0cm}
% \subfigcapskip = -0.05cm
\includegraphics[width=7in]{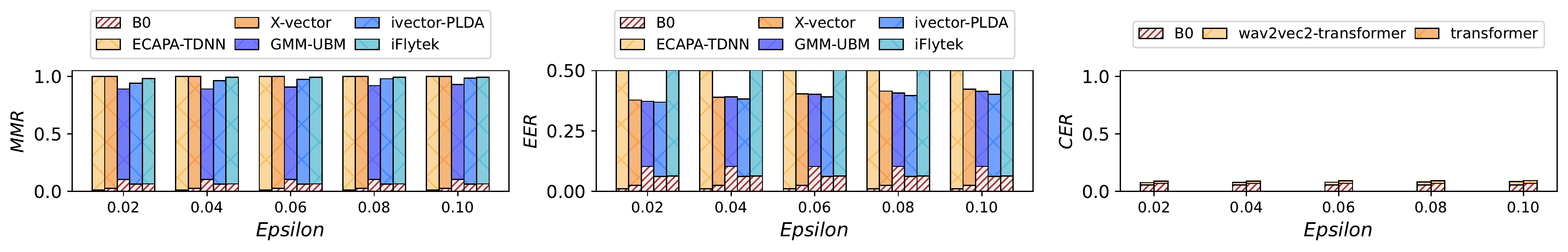}

\caption{Cross-language performance on Chinese. (a) MMR. (b) EER. (c) CER. \sys is trained on the English dataset and tested on the Chinese dataset. \sys can retain intelligibility across different languages (a CER of 8.11\%-8.86\%).}\label{fig:chinese}
\end{figure*}

\begin{figure*}[t]
    \centering
\setlength{\abovecaptionskip}{0pt}
\setlength{\belowcaptionskip}{0cm}
% \subfigcapskip = -0.05cm
\includegraphics[width=7in]{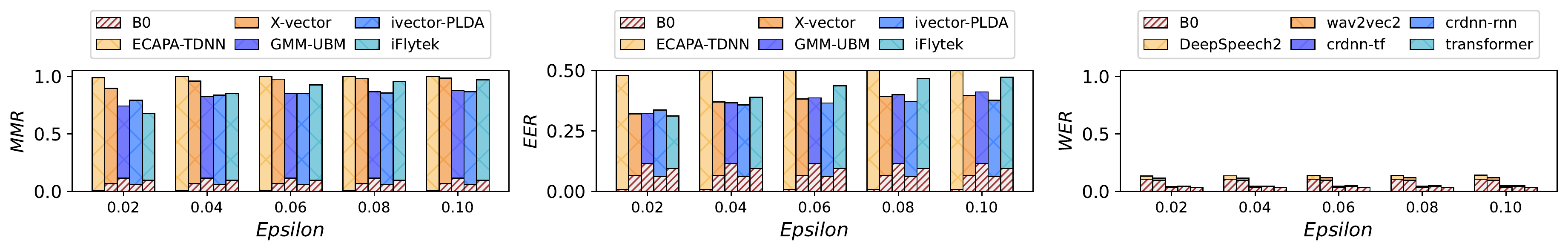}

\caption{The performance at different anonymization levels. (a) MMR. (b) EER. (c) WER. The MMRs and EERs of \sys are much higher than those of clean audios (B0) on all 5 ASVs. \sys only causes a slight increase in the WERs on 5 ASRs.}\label{fig:epsilon}
\end{figure*}

\begin{figure*}[tt]
    \centering
\setlength{\abovecaptionskip}{-0.0cm}
\setlength{\belowcaptionskip}{-0.0cm}
% \subfigcapskip = -0.25cm

\subfigure[French]{
    \includegraphics[width=1.5in, trim=5 5 5 5, clip]{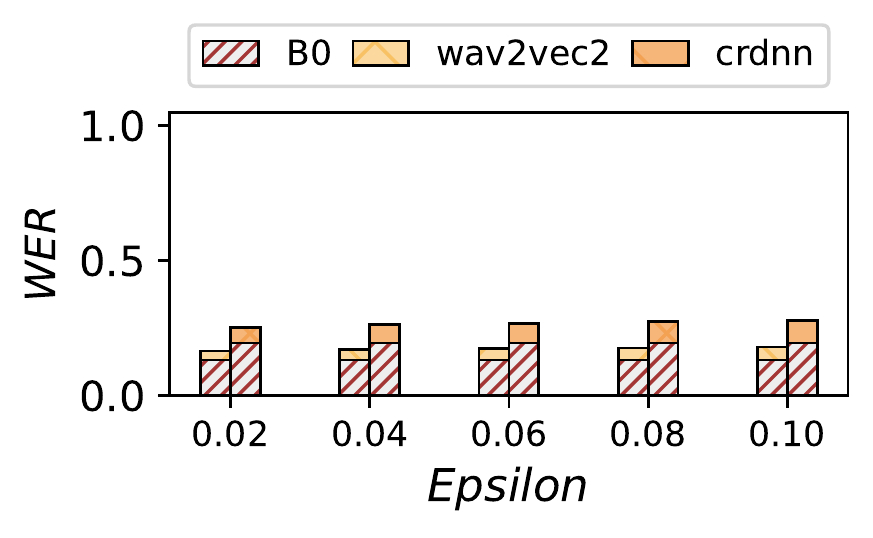}
    \label{fig:french}
    }
\quad  
% \hspace{0.2cm}
\subfigure[Italian]{
    \includegraphics[width=1.5in, trim=5 5 5 5, clip]{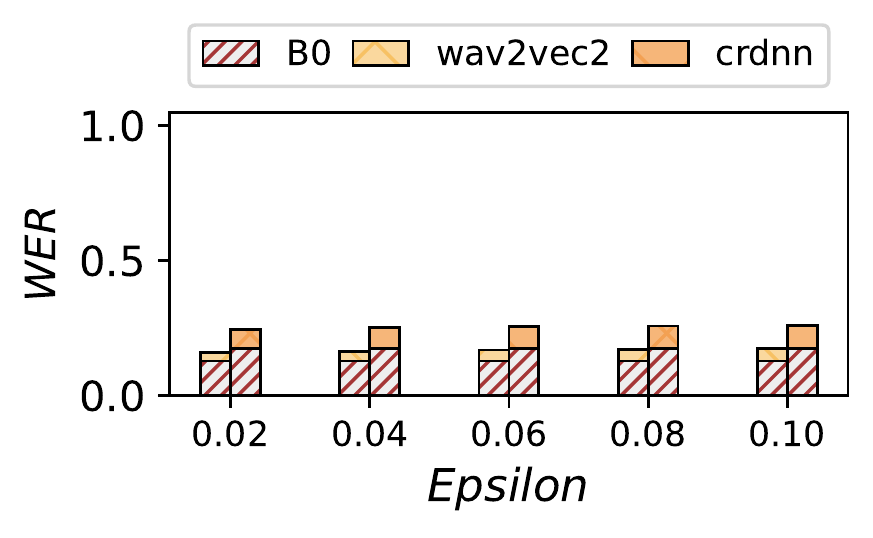}
    \label{fig:italian}
    }

\vspace{-0.0cm}\caption{Cross-language performance on French and Italian. \sys is trained on the English dataset and tested on the French and Italian datasets. \sys can retain intellibility across different languages (an average increase in WERs of 6.59\% and 6.55\% for French and Italian respectively).}\label{fig:fr_it}
\end{figure*}

\begin{figure*}[t]
    \centering
\setlength{\abovecaptionskip}{0pt}
\setlength{\belowcaptionskip}{0cm}
% \subfigcapskip = -0.05cm
\hspace{-0.5cm}\includegraphics[width=4.5in, trim=5 5 5 5, clip]{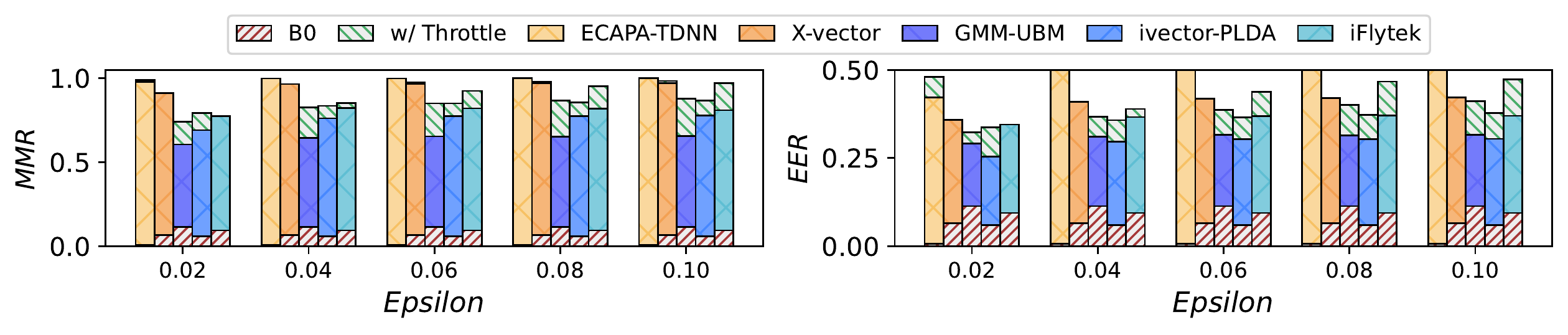}

\caption{Ablation study on \emph{Throttle}. The bars with green diagonal stripes denote the gap between the performance w/ and w/o \emph{Throttle}.}\label{fig:wothrottle}
\end{figure*}

%%%%%%%%%%%%%%%%%%%%%%%%%%%%%%%%%%%%%%%%%%%%%%%%%%%%%%%%%%%%%%%%%%%%%%%%%%%%%%%%
\end{document}